\documentclass{article}

\usepackage[english]{babel}

\usepackage[letterpaper,top=2cm,bottom=2cm,left=3cm,right=3cm,marginparwidth=1.75cm]{geometry}

\usepackage{algorithm}
\usepackage{algpseudocode}
\usepackage{amsmath}
\usepackage{amssymb}
\usepackage{amsthm}
\usepackage{xy}
\xyoption{matrix}
\xyoption{frame}
\xyoption{arrow}
\xyoption{arc}

\newtheorem{property}{Property}

\newtheorem{lemma}{Lemma}
\newcommand{\Endproof}{\hfill$\Box$\\}
\usepackage{mathtools}
\usepackage{amsmath}

\usepackage{xcolor} %% это необходимо для выделения цветом тескта
\definecolor{ForestGreen}{rgb}{	0.0,	0.478,	0.125} %% цвет Наили. Выделила то, что добавила, и то, что нужно будет добавить
\definecolor{Red}{rgb}{	0.9,	0.0, 0.0}
\usepackage{graphicx}
\usepackage[colorlinks=true, allcolors=blue]{hyperref}
\newcommand{\Beginproof}{{\em Proof.}  }
\DeclareMathOperator{\nand}{\texttt{ NAND }}
\newcommand{\bigo}[1]{{O\left({#1}\right)}}

\newcommand{\bra}[1]{{\left\langle{#1}\right\vert}}
\newcommand{\ket}[1]{{\left\vert{#1}\right\rangle}}
\newcommand{\qw}[1][-1]{\ar @{-} [0,#1]}
\newcommand{\qwx}[1][-1]{\ar @{-} [#1,0]}
\newcommand{\gate}[1]{*+<.6em>{#1} \POS ="i","i"+UR;"i"+UL **\dir{-};"i"+DL **\dir{-};"i"+DR **\dir{-};"i"+UR **\dir{-},"i" \qw}
\newcommand{\control}{*!<0em,.025em>-=-<.2em>{\bullet}}
\newcommand{\ctrl}[1]{\control \qwx[#1] \qw}
\newcommand{\lstick}[1]{*!R!<.5em,0em>=<0em>{#1}}
\newcommand{\Qcircuit}{\xymatrix @*=<0em>}
\newcommand{\braket}[2]{{\langle {#1}\!\mid\!{#2} \rangle}}
\newcommand{\Hilbert}{{\mathcal H}}

\title{Lecture Notes on Quantum Algorithms}
\author{Kamil Khadiev}
\date{Kazan Federal University, Kazan, Russia}
\begin{document}
\maketitle

\begin{abstract}
The lecture notes contain three parts. The first part is Grover's Search Algorithm with modifications, generalizations, and applications. The second part is a discussion on the quantum fingerprinting technique. The third part is Quantum Walks (discrete time) algorithm with applications.
\end{abstract}

\tableofcontents

\section{Introduction}

Preliminaries of this lecture notes are based on \cite{aazksw2019part1,nc2010}.
%%%%%%%%%%%%%%%%%%%%%%%%%%%%%%%%%%%%%%%%%%%%%%%

%%%%%%%%%%%%%%%%%%%%%%%%%%%%%%%%%%%%%%%%%%%%%%%%
\section{Notations}
\begin{itemize}
    
    \item $\mathbb{Z}$ is the set of integers.
     \item $\mathbb{Z}_+$ is the set of positive integers.
     \item $\mathbb{C}$ is the set of complex numbers.
      \item $\mathbb{R}$ is the set of real numbers.
      \item $[K]=\{0,\dots,K-1\}$, for some $K\in \mathbb{Z}_+$.
      \item $\mathcal{H}^k$ is the $k$-dimensional Hilbert space. 
\item $|K|$ is a power of the set $[K]$.
\end{itemize}
\begin{itemize}
\item $\log a=\log_2 a$.
\item $e=exp(1)$ that is a  base of a natural logarithm or an  Euler's number.
\item A norm of $a=(a_1, \dots, a_d)$ is $||a||=\sqrt{\sum_{i=1}^d |a_i|^2}$.
\item $\delta_{a,b}$ is a function such that $\delta_{a,b}=1$ if $a=b$, $\delta_{a,b}=0$ if $a\neq b$.
\item  $\ket{a}$ is a column vector $a$.
\item $\bra{b}=(\ket{b}^T)^*$ is a row vector $b^*$.
\item $\braket{a}{b}$ is an inner product of $\ket{a}$ and $\ket{b}$.
\item $\ket{a} \otimes \ket{b}$ is a tensor product of vectors $\ket{a}, \ket{b}$.
For $\ket{a}=(a_1, \dots, a_d)^T$ and $\ket{b}=(b_1,\dots, b_l)^T$, we have
$\ket{a} \otimes \ket{b} =(a_1b_1, a_1b_2, \dots ,a_1b_l, a_2b_1,
\dots, a_db_l)^T$.
\item $\ket{a}\ket{b} = \ket{a b} = \ket{a}\otimes\ket{b}$
\item $\ket{a}\bra{b} = \ket{a}\otimes\bra{b}$,  
$(\ket{a} \otimes \bra{b})[i][j] = a_ib_j.$ That is, $\ket{a}\bra{b}$ is $d\times l$ matrix $C$ with entries $c_{ij}=a_ib_j$. 
\item $A \otimes B$ is a tensor product of matrices $A, B$
\item $A^{\otimes d}=\underbrace{A\otimes A\otimes\cdots\otimes A}\limits_d$.
\item $(A\otimes B)(\ket{\phi}\otimes\ket{\psi}) = A\ket{\phi}\otimes
B\ket{\psi}$.
\item $f(n)=O(g(n))$ if there are constants $c>0$ and $n_0>0$ such that $0 \leq f(n)\leq c\cdot g(n)$ for all $n \geq n_0$.
\item $f(n)=\Omega(g(n))$ if there are constants $c>0$ and $n_0>0$ such that  $0 \leq c\cdot g(n) \leq f(n))$ for all $n \geq n_0$.
\item $f(n)=\Theta(g(n))$ if $f(n)=O(g(n))$ and $f(n)=\Omega(g(n))$.
\end{itemize}

\section{Basics of Quantum Computing}
\subsection{Qubit.}
The notion of quantum bit  (qubit) is the basis of quantum computations. Qubit is the quantum version of the classical binary bit physically realized with a two-state device. There are two possible outcomes for the measurement of a qubit usually taken to have the value "0" and "1", like a bit or binary digit. However, whereas the state of a bit can only be either 0 or 1, the general state of a qubit according to quantum mechanics can be a coherent superposition of both. This allows us to compute $0$ and $1$ simultaneously. Such a phenomenon is known as quantum parallelism. 

Formally the qubit's state is the column vector $|\psi\rangle$ from two dimensional Hilbert space $\Hilbert^2$, i.e.
\begin{equation}\label{qubit}\ket{\psi} = \alpha\ket{0} + \beta\ket{1},\end{equation}
Here the pair of vectors $\ket{0}$ and $\ket{1}$ is an orthonormal basis of 
$\Hilbert^2$, 
where $\alpha,\beta\in \mathbb{C}$ such that 
$|\alpha|^2 + |\beta|^2 = 1$.  The numbers $\alpha,\beta$ are called amplitudes.
%
%So, the state  $\ket{\psi}$ (\ref{qubit}) of qubit encodes a superposition of $0$ and $1$.  
\begin{figure}[h]
\begin{center}
\includegraphics[scale=0.2]{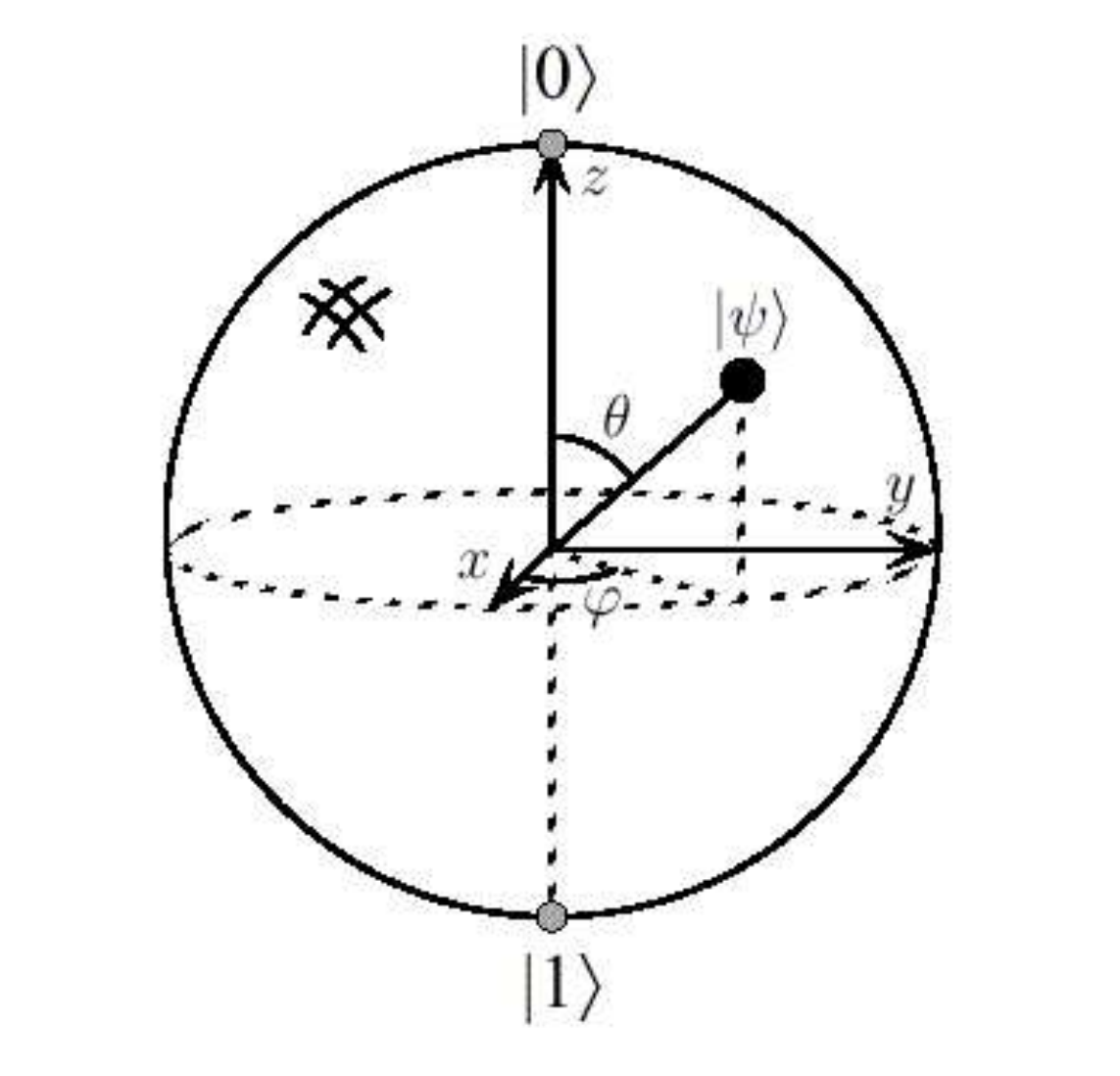}
\caption{The Bloch sphere.}\label{fig:BlochSphere}
\end{center}
\end{figure}
The Bloch sphere is a representation of a qubit's state as a point on a three-dimensional unit sphere (Fig. \ref{fig:BlochSphere}).
Let us consider a qubit in a state (\ref{qubit}) such that  $|\alpha|^2 + |\beta|^2 = 1$. Then we can represent amplitudes as
%
%\begin{equation}\alpha=e^{i\gamma}\cos{\theta'}, \beta=e^{i\phi'}\sin{\theta'},\end{equation} 
%where $0\leq\gamma<2\pi$, $0\leq\theta'\leq\pi/2$ and $0\leq\phi'<2\pi$ are real numbers.
%
%Let $\phi=\phi'-\gamma$. Because of $|e^{i\gamma}|^2=1$ we have 
%\begin{equation}\ket{\psi} = \cos{\theta'}\ket{0} + e^{i\phi}\sin{\theta'}\ket{1}.\end{equation}
%
%Let $e^{i\phi}\sin{\theta'}= x + iy$ and $z = \cos{\theta'}$. Then $x^2 + y^2 + z^2 = 1$ and the state of the qubit can be represented as a point on the unit sphere. Here $\theta'$ and $\phi$ are polar angels of the point on the sphere.
% $(\sin\theta'\cos\phi,
%\sin\theta'\sin\phi, \cos\theta')$
%
%Note that opposite point with angles $\pi-\theta'$ and $\pi+\phi$ such that
%\begin{equation}\begin{array}{rcl}
%\ket{\psi'} & = & \cos(\pi-\theta')\ket{0} + e^{i(\pi+\phi)}\sin(\pi-\theta')\ket{1}\\ & = & -\cos{\theta'}\ket{0} - e^{i\phi}\sin{\theta'}\ket{1}\\
%& = & -\ket{\psi}.\end{array}\end{equation}
%
%That is why it is enough to consider the bottom half sphere ($0\leq\theta'\leq\pi/2$). Let $\theta=2\theta'$ so the state of a quantum qubit is
%
\begin{equation}\ket{\psi} = \cos\frac{\theta}{2}\ket{0} + e^{i\phi}\sin\frac{\theta}{2}\ket{1},\mbox{ where $0\leq\phi<2\pi$, $0\leq\theta\leq\pi$}\end{equation}
\begin{figure}[h]
\begin{center}
\includegraphics[scale=0.5]{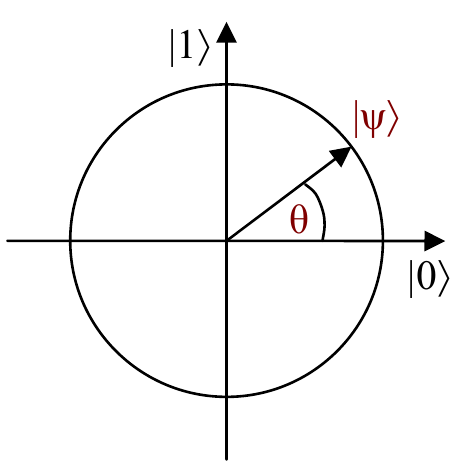}
\caption{The representation of a qubit's state with real-value amplitudes.}\label{fig2}
\end{center}
\end{figure}
If we consider only real values for $\alpha$ and $\beta$ then the state of a qubit is a point on a unit circle (Fig. \ref{fig2}). In a case of real-value amplitudes, the state of the qubit is
\begin{equation}
\ket{\psi}=\cos\theta\ket{0}+\sin\theta\ket{1}\mbox{, where $\theta\in[0,2\pi)$}.
\end{equation}

\subsection{States of Quantum Register}
 A quantum register is an isolated quantum mechanical system composed of several $n$ qubits (quantum $n$-qubits register), where  $n\ge 1$. In the case of a quantum $n$-qubits register, we have a superposition of $2^n$ states $B=\{(0\dots0), \dots, (1\dots1)\}$. This allows us to compute the set of $B$ states simultaneously. This phenomenon of quantum parallelism is a potential advantage of quantum computational models.
%end_color

Formally, a quantum state $\ket{\psi}$ of a quantum  $n$-qubits register is described as follows.  Let $\sigma=\sigma_1\dots\sigma_n$ be a binary sequence. Then,   we denote tensor product $\ket{\sigma_1}\otimes\ket{\sigma_2}\otimes\dots\otimes\ket{\sigma_n}$ by $\ket{\sigma}$.  Let  ${\cal B}=\{\ket{00\ldots0}$, $\ket{00\ldots1}$, \ldots,
$\ket{bin(i-1)}$, \ldots, $\ket{11\ldots1}\}$ be a set of orthonormal vectors, where $bin(i)$ is a binary representation of $i$. A set ${\cal B}$ forms a basis for $2^n$ dimensional Hilbert space $\Hilbert^{2^n}$. We also denote basis vectors of ${\cal B}$ for short as   $\ket{0}, \dots, \ket{2^n-1}$. Usually, the set ${\cal B}$ is called computational basis. 

A quantum state $\ket{\psi}$ of a quantum  $n$-qubits register   is a  complex valued unit vector in $2^n$-dimensional Hilbert space $\Hilbert^{2^n}$ that is described as a linear combination of basis vectors $\ket{i}$, $i\in \{0,\dots,2^n-1\}$:
\begin{equation}\label{amplitudes}
\ket{\psi}=\sum\limits_{i=0}^{2^n-1}\alpha_i\ket{i}, \quad \mbox{ with } \quad 
\sum\limits_{i=0}^{2^n-1}|\alpha_i|^2=1.\end{equation}
When we measure an $n$-qubits register of the state $\ket{i}$ with respect to ${\cal B}$, we can say that $|\alpha_i|^2$ expresses a probability to find the register in the state $\ket{i}$.
We say that the state $\ket{\psi}$ is a superposition of basis vectors  $\ket{i}$ with amplitudes $\alpha_i$. We also use notation $(\Hilbert^2)^{\otimes n}$ for Hilbert space $\Hilbert^{2^n}$ to outline the fact that its vectors represent the states of a quantum $n$-qubits register.

If a state $\ket{\psi}\in (\Hilbert^2)^{\otimes n}$ can be decomposed to a tensor product of single  qubits
%\begin{equation}
\[\ket{\psi} = \ket{\psi_1}\otimes\ket{\psi_2}\otimes\cdots\otimes\ket{\psi_n},\]
%\end{equation}
then we say that $\ket{\psi}$ is  {\em not entangled}. Otherwise, 
 the state is called  \emph{entangled}. EPR-pairs are an example of entangled states.

%\paragraph{Comment.} It is important once again to outline the difference between classical and quantum $n$-registers.  In the classical world, a state $\sigma$ of $n$-register is a binary $n$ length sequence $\sigma\in B$. While in the quantum world a state $\ket{\psi}$ (\ref{amplitudes}) of $n$-register is a  quantum superposition of all its basis states.  There are several interpretations of the formula (\ref{amplitudes}).  One interpretation of quantum mechanics describes quantum superposition as the phenomenon that the quantum state is simultaneously in all its basis states, which are manifested in accordance with their amplitudes when measuring the state.

\subsection{Transformations of a Quantum States.}
Quantum mechanic postulates that
transformations  of quantum states $\ket{\psi}\in (\Hilbert^2)^{\otimes n }$ (of quantum $n$-qubits register) are mathematically  determined by unitary operators:
\begin{equation}\label{operator}
\ket{\psi'} = U\ket{\psi},
\end{equation}
where  $U$ is a $2^n\times 2^n$ unitary matrix for transforming a vector that represents a state of a quantum $n$-qubits register. 
%This transformation  will be also presented  in the form  
% \[ U:\ket{\psi} \to \ket{\psi'}. \]
 
%The problem of realizing an effective quantum procedure $QP$ (based on the invented quantum algorithm $QA$) for solving the problem can be formulated as follows. First. A model of computation is selected. Second. A suitable quantum basis for the data representation of the problem should be chosen. Thirdly. A sequence of operators of the form realizing the transformation of the quantum system based on the quantum parallelism effect is developed

A unitary matrix $U$ can be written in the exponential form:
$$U = e^{iW},$$

where $W$ is a Hermitian matrix. Learn more from \cite{aHermitian2017}.

\subsection{Basic Transformations of  Quantum States}
Basic transformations of a quantum register use the following notations in the quantum mechanic. 
Qubit transformations are rotations on an angle $\theta$ around $\hat{x}$, $\hat{y}$
and $\hat{z}$ axis of the Bloch sphere:
\begin{equation}\begin{array}{l}
R_{\hat{x}}(\theta)=\left(
\begin{array}{cc}
\cos\frac{\theta}{2} & -i\sin\frac{\theta}{2}\\
-i\sin\frac{\theta}{2} & \cos\frac{\theta}{2}
\end{array}\right);\\[0.5cm]
R_{\hat{y}}(\theta)=\left(
\begin{array}{cc}
\cos\frac{\theta}{2} & -\sin\frac{\theta}{2}\\
\sin\frac{\theta}{2} & \cos\frac{\theta}{2}
\end{array}\right);\\[0.5cm]
R_{\hat{z}}(\theta)=\left(
\begin{array}{cc}
e^{-\frac{i\theta}{2}} & 0\\
0 & e^{\frac{i\theta}{2}}
\end{array}\right).
\end{array}\end{equation}

Several such transformations  (unitary matrices) have specific notations.
\begin{itemize}
    \item 
$I$ is an identity operator. That is, $I= R_{\hat{x}}(0)= R_{\hat{y}}(0)=R_{\hat{z}}(0) $
\[ I=\left(\begin{array}{cc}1 & 0\\0 & 1\end{array}\right).\]

\item $X$ is a NOT operator. NOT flips the state of a qubit. 
It is a special case of the $R_{\hat{x}}(\theta)$ that is a rotation around the $X$-axis of the Bloch sphere by $\pi$.

\begin{equation}\begin{array}{l}
X=\left(\begin{array}{cc}0 & 1\\1 & 0\end{array}\right).
%S=\left(\begin{array}{cc}1 & 0\\0 & i\end{array}\right);\\
\end{array}
\end{equation}

\item $S$ and $T$ are phase transformation operators
\begin{equation}\begin{array}{l}
S=\left(\begin{array}{cc}1 & 0\\0 & i\end{array}\right);\quad 
T=
e^{i\pi/8}\left(\begin{array}{cc}e^{-i\pi/8} & 0\\0 &
e^{i\pi/8}\end{array}\right)=\left(\begin{array}{cc}1 & 0\\0 &
e^{i\pi/4}\end{array}\right).
\end{array}
\end{equation}

They are special cases of rotation around the $Z$-axis of the Bloch sphere. The transformation $S$ is for rotation on $\frac{\pi}{2}$ and $T$ is for rotation on $\frac{\pi}{4}$.

\item $\sigma_z$ is a Pauli-Z gate. It is a special case of the rotation $R_{\hat{z}}(\theta)$  around the $Z$-axis on an angle $\pi$.
\begin{equation}\begin{array}{l}
\sigma_z =\left(\begin{array}{cc}1 & 0\\0 & -1\end{array}\right).
\end{array}
\end{equation}

\item $H$ is a Hadamard operator. It is a combination of a rotation around the $Z$-axis on $\pi$ and a rotation around the $Y$-axis on $\frac{\pi}{2}$.
$$H = R_{\hat{y}}\left(\frac{\pi}{2}\right)R_{\hat{z}}(\pi)$$
\begin{equation}\begin{array}{l}
H=\frac{1}{\sqrt{2}}\left(\begin{array}{cc}1 & 1\\1 & -1\end{array}\right).
\end{array}
\end{equation}

%Известно, что любое однокубитное унитарное преобразование можно с произвольной точностью $\epsilon$ представить в виде произведения $O(\log^c{1/\epsilon})$ операторов $H$ и $\pi/8$ (константа $c$ приблизительно равна 2), а произвольное унитарное преобразование, действующее на $q$ кубитах может быть представлено в виде произведения $O(q^2 4^q)$ однокубитных и $\cnot$ операторов. Фазовый оператор $S$ включается в стандартный набор для реализации контролируемых операторов и для организации помехоустойчивых вычислений.
\end{itemize}

%\color{blue}

\paragraph{Comments.}
The Hadamard operator has several useful properties.  

Firstly, let us note that  $H=H^{*}$. Therefore, for any state of a qubit $|\psi\rangle$ we have $H(H |\psi\rangle)=|\psi\rangle$.
Secondly, the Hadamard operator creates a superposition of the basis states of the qubit with equal probabilities.

    Hadamard operator maps the $\ket{0}$ to ${\frac {\ket{0} + \ket{1}}{\sqrt {2}}}$. A measurement of this qubit has equal probabilities to obtain basis states $\ket{0}$  or $\ket{1}$.     Repeated application of the Hadamard gate leads to the initial state of the qubit.

Thirdly, if a one-qubit  state is  in the form
    $\ket{\psi_0} = \frac{e^{i\phi}}{\sqrt{2}}\ket{0} + \frac{e^{-i\phi}}{\sqrt{2}}\ket{1},$
    then applying the Hadamard operator to this qubit allows transferring phase information into amplitudes.
    After applying the Hadamard gate to $\ket{\psi_0}$ we obtain:
    $\ket{\psi_1}  = H\ket{\psi_0}  = \cos{\phi}\ket{0} + \sin{\phi}\ket{1}$
    
%That is, applying Hadomar transfor for the state  $\ket{\psi_0}$ in the phase form transforms it into state $\ket{\psi_1}$   in the amplitude form.  
    
    %Imagine we have the following two-qubit quantum state:
    %$$\ket{\psi_0} = \frac{1}{\sqrt{2}}(e^{i\phi}\ket{0} + e^{-i\phi}\ket{1}) \otimes \frac{1}{\sqrt{2}}(e^{i\omega}\ket{0} + e^{-i\omega}\ket{1})$$
    
    %After applying the Hadamard gate to $\ket{\psi_0}$ we obtain:
    %$$\ket{\psi_1} = \cos{\phi}\cos{\omega}\ket{00} + \cos{\phi}\sin{\omega}\ket{01} + \sin{\phi}\cos{\omega}\ket{10} +
    %\sin{\phi}\sin{\omega}\ket{11}$$
    
%This rule can be generalized for a quantum state of an arbitrary number of qubits.

\subsection{Quantum Circuits}

Circuits are one of the ways of visually representing a register's state's transformations sequence.  

A classical Boolean circuit is a finite directed acyclic graph with AND, OR, and NOT gates. It has $n$ input nodes, which contain the $n$ input bits (a state of $n$-register). Internal nodes are AND, OR, and NOT gates, and there are one or more designated output nodes. The initial input bits are fed into
AND, OR, and NOT gates according to the circuit, and eventually the output
nodes assume some value. Circuit computes a Boolean function
$f : \{0, 1\}^n \to \{0, 1\}^m$
if the output nodes get the  value $f(\sigma)$ for every input
$\sigma \in \{0, 1\}^n$.

A quantum circuit (also called a quantum network or quantum gate array) acts on a quantum register. It is a representation of transformations of a quantum state (a state of a quantum register). Quantum circuit generalizes the idea of classical circuit families, replacing the $AND, OR$, and $NOT$ gates by elementary operators (quantum gates). A quantum gate is a unitary operator on a small (usually 1, 2, or 3) number of qubits.  Mathematically, if gates are applied in a parallel way to different parts of the quantum register, then they can be composed using a tensor product. If gates are applied sequentially, then they can be composed using the ordinary product.

In the picture, we draw qubits as lines (See Figure \ref{fig:qubitcur}) and operators or gates as rectangles or circles.  
\begin{figure}[h]
\begin{center}
\includegraphics[height=0.75cm]{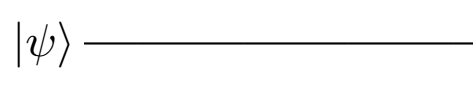}
\caption{Qubit}\label{fig:qubitcur}
\end{center}
\end{figure}

Figure \ref{fig:hgate} is the Hadamard operator (or gate) and Figure \ref{fig:notgate} is the $NOT$ or $X$ gate. If a rectangle of a gate crosses a line of a qubit, then the gate is applied to this qubit. If the rectangle crosses several lines of qubits, then it is applied to all of these qubits. 
\begin{figure}[h]
\begin{center}
\includegraphics[height=0.75cm]{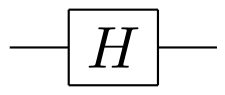}
\caption{Hadamard gate}\label{fig:hgate}
\end{center}
\end{figure}
\begin{figure}[h]
\begin{center}
\includegraphics[height=0.75cm]{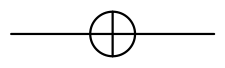}
\caption{$NOT$ or $X$ gate}\label{fig:notgate}
\end{center}
\end{figure}

%As a computational basis gate one can choose $H$ and $T$.
%
%gates because of the next property.
%
%\begin{property}
%Any transformation of one qubit  ($2\times 2$ unitary operator) for any \color{red} accuracy $\epsilon$  can be represented as a product of $O(\log^c{1/\epsilon})$ operators $H$ and $T$ for $c$ around 2.
%\end{property}
%Another 
%There are different options to choose a set of basic quantum gates that forms a computational basis.  As an option for computational basis gates we can choose $R_{\hat{x}}, R_{\hat{y}},R_{\hat{z}}$ because of the following property.
%\begin{property}
%Any  transformation of one qubit  can be represented as
%\begin{equation}U=e^{i\alpha}R_{\hat{x}}(\beta)R_{\hat{y}}(\gamma)R_{\hat{z}}(\delta)\end{equation}
%for some real $\alpha$, $\beta$, $\gamma$ and $\delta$. 
%\end{property}

%In this paper, we use $H, CNOT, T$ and $X$ gates because ????. The $CNOT$ gates are defined in the next section.

Figure \ref{circuitCnot} represents the CNOT gate that implements the following transformation on two qubits:
\[\ket{a}\ket{b}\to\ket{a}\ket{a XOR b}\]

Here XOR is excluding or operation that is $1$ iff $a\neq b$. The gate is also called the ``control-not'' or ``control-X'' gate. The $\ket{a}$ qubit is called the control qubit and $\ket{b}$ is called the target qubit. We can say that we apply $NOT$ or $X$ operator to $\ket{b}$ iff $\ket{a}$ is $1$. 
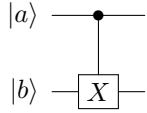
\begin{figure}[h]
$$
\begin{array}{l}\quad\quad \Qcircuit  @C=1.0em @R=1.0em {
\lstick{\ket{a}} & \ctrl{2} & \qw\\
\\
\lstick{\ket{b}} & \gate{X} & \qw\\
}\end{array}
$$
\caption{CNOT gate}\label{circuitCnot}
\end{figure}

\subsection{Information Extracting (Measurement).}

There is only one way to extract information from a state of a quantum $n$-register to a ``macro world''. It is the measurement of the state of the quantum register. Measurement can be classified as the second kind of quantum operators, while unitary transformations of a quantum system as the first kind of quantum operators.

Different measurements are considered in quantum computation theory.  In our review, we use only measurement with respect to ``computational basis''. Such a type of measurement is described as follows. If we measure the quantum state  $\sum_i\alpha_i\ket{i}$, then we get one of the basis state $\ket{i}$  with probability $|\alpha_i|^2$.

We also can use a partial measurement of the state of the quantum register. 
Consider the case of a quantum $2$-register. Let $\ket{\psi}$ be a state of such a quantum $2$-register:
\[\ket{\psi} = \alpha_{00}\ket{00} + \alpha_{01}\ket{01} + \alpha_{10}\ket{10} + \alpha_{11}\ket{11}.
\]
Imagine that we measure the first qubit of the quantum $2$-register. The probability of obtaining $\ket{0}$ is
$$Pr(\ket{0}) = Pr(\ket{00}) + Pr(\ket{01}) = |\alpha_{00}|^2 + |\alpha_{01}|^2$$
The   state of the second qubit after the measurement  is
$$\ket{\psi_0} = \frac{\alpha_{00}\ket{0} + \alpha_{01}\ket{1}}{\sqrt{|\alpha_{00}|^2 + |\alpha_{01}|^2}}.$$
Similarly, we define the probability of obtaining 1-result %is equal to the sum over all of the possible probabilities for a two-qubit measurement where we get 0 on the first qubit:
$$Pr(\ket{1}) = Pr(\ket{10}) + Pr(\ket{11}) = |\alpha_{10}|^2 + |\alpha_{11}|^2$$
The  state  of the second qubit after the measurement  is
$$\ket{\psi_1} = \frac{\alpha_{10}\ket{0} + \alpha_{11}\ket{1}}{\sqrt{|\alpha_{10}|^2 + |\alpha_{11}|^2}}.$$
%

%%%%%%%%%%%%%%%%%%%%%%%%%%%%%%%%%%%%%%%%%%%%%%%%%%%%%%%
%%% Computational model%%%%%%%%%%%%%%%%%%%%%%%
%%%%%%%%%%%%%%%%%%%%%%%%%%%%%%%%%%%%%%%%%%%%%%%%%%%%%%%%

\section{Computational model}
To present a notion of computational complexity, we need a formalization of the computational models we use. We define here the main computational models that are used for quantum search problems. 

% more information on the subject can be found in the following books and papers \cite{a2017,Wolf:2001:PhD,Weg00,av2009}.  
 
In this section, we  define computational models oriented for Boolean functions 
\[f:\{0,1\}^N\to\{0.1\}\]
realization. Denote $X=\{x_1\dots,x_N\}$ a set of variables of function $f$. 

Note, that the computational model we define here can be naturally generalized for the case of computing more general functions. 
 
 %We consider computing a function $f(x_1, \dots , x_N )$ of variables $x_i\in[d]$ for some integer $d$. Let us consider the case of $d=2$, but the description of the models is right for other $d$.

%\paragraph{Quantum Computing Model.}

\subsection{Decision Trees}
We consider a deterministic version of the Decision Tree (DT) model of computation.  We present here a specific version of DTs that is oriented for quantum generalization.

The model can be presented in two forms. 
\paragraph{Graph representation.} DT   $A$ is a directed leveled binary tree with a selected starting node (root node). All nodes $V$ of $A$ are  partitioned into levels $V_1, \dots, V_\ell$. The level $V_1$ contains the starting node. Nodes from the level $V_i$ are connected only to nodes from the level $V_{i+1}$. 
The computation of a DT starts from the start node. At each node of the graph, a Boolean variable $x\in X$ is tested.  Depending on the outcome of the query, the
algorithm proceeds to the $x = 0$ child or the $x = 1$ child of the node. Leaf nodes are marked by ``0'' or ``1''. When on input $\sigma=\sigma_1,\dots , \sigma_N$ the computation reaches a leaf of the tree, then it outputs the value listed at this leaf.

\paragraph{Linear representation.} We define a  linear presentation of the graph-based model of DT as follows. 

Let $dim(A)=\max\limits_{1\le i\le \ell} |V_i|$ and  $d=dim(A)$.  
Denote $S=\{s^1,\dots, s^d\}$ a set of $d$-dimensional  column-vectors, where $s^{(i)}$ is a vector with all components ``0'' except  one ``1'' in the $i$-th position.  We call a vector  $s\in S$ a state of $A$.  A state $s$ in each level $i$, $1\le i\le \ell$ represents a node (by ``1'' component) where $A$ can be found at that level. 

Next, let $X_i\subseteq X $ be a set of variables tested on the level $i$.
A transformation of a state $s$ on a level $i$, $1\le i\le \ell-1$, is described by a set ${\cal Q}^i(x_i)$ of matrices which depends on values of variables $x_i$ tested on that level.  The state $s^{(0)}\in S$ is an initial state. 

%A sequence of transformations ${\cal Q}^1,\dots, {\cal Q}^{\ell-1}$  and a sequence of states $s_1, \dots, s_\ell$ are determined by the input $\sigma$ and the structure of the underlying decision tree. 

Now DT $A$ can be formalized as 
\[
A= \langle S, {\cal Q}^1(x_1), \dots, {\cal Q}^{\ell-1}(x_{\ell-1}), s^{(0)} \rangle.  
\]
Computation on an input $\sigma$ by $A$ is presented as a sequence 
\[ s^{(0)} \to s^{(1)} \to \cdots  \to s^{(\ell)} 
\]
of states transformations determined by $A$  structure and input $\sigma$ as follows.  Let  $s^{(i)}$ be the current state on the level $i$, then the next state  will be $s^{(i+1)}={Q}^i(x_i,\sigma) s^{(i)}$, where matrix ${Q}^i(x_i,\sigma)\in {\cal Q}^i(x_i) $ is determined by  $\sigma$ for  $x_i$ variables.  

\paragraph{Complexity Measures.}

Two main complexity measures are used for computational models. These measures are Time and Space (Memory). For the Decision models of computation, the analogs of these complexity measures are query complexity and size complexity. The query complexity is the maximum number of queries that the algorithm can do during computation. This number is the depth $D(A)$ that is the length of the longest path from the root to a leaf of the decision tree $A$. The size complexity is the width of $A$. We denote the width $dim(A)$.  In addition, we 
use a number of bits, which are enough to encode a state on a level of $A$ that is $size(A)=\lceil \log_2 dim(A) \rceil$.  

For function $f$, complexity measures $D(f)$, $dim(f)$, and $size(f)$ denote (as usual) the minimum Time and Space needed for computing $f$.

%We can discuss two complexity measures: time and memory. The analog of time complexity is the query complexity of $A$. The query complexity is the maximum number of queries that the algorithm can do or it is the height of the tree. Let us denote it $height(A)=\ell$. Deterministic query complexity $D(f)$ is the smallest $height(A)$ of a deterministic $A$ which outputs $f(x_1, \dots , x_N )$ if the queries are answered according to $(x_1, \dots , x_N )$, whenever $f(x_1, \dots , x_N )$ is defined.

%The memory complexity or space complexity of $A$ is the width of the tree. Let us denote it $dim(A)$. Formally,  $dim(A)=\max\limits_{i\in\{1,\dots height(A)\}} |V_i|$. 

%We can say that $\lceil \log_2 dim(A) \rceil$ bits are enough to encode an index of the current state on a level. We define this number as $size(A)=\lceil \log_2 dim(A) \rceil$.

\subsection{Quantum Query Algorithm} 

Quantum Query Algorithm (QQA) is the following generalization of the DT model for a quantum case. 
A QQA  $A$ for computing a Boolean function $f(X)$ is based on a quantum  $size(A)$-register (on a quantum system composed from $size(A)$ qubits). $|\psi_{start}\rangle$ is an initial state. The computation procedure  is determined by the sequence 
\[ U_0, Q, U_1, \dots , Q, U_{\ell},\]
of operators that are $dim(A)\times dim(A)$ unitary matrices.

%defined by an initial state $|\psi_{start}\rangle$ and operators $U_0, Q, U_1, \dots , Q, U_{\ell}$.  The algorithm uses a quantum system of $size(A)$ qubits or $dim(A)$ states. All operators are represented as $dim(A)\times dim(A)$ unitary matrices.

\begin{figure}[h]
    \begin{center}
    \includegraphics[height=2cm]{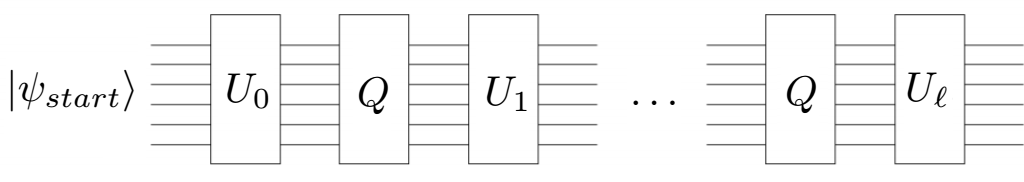}
    \caption{Quantum query algorithm as a quantum circuit.}
    \label{fig:queryalgo}
    \end{center}
\end{figure}

Algorithm  $A$ is composed of two types of operators. Operators  $U_i$ are independent of a tested input $X$. $Q$ is the query operator of a fixed form that depends on the tested input $X$. The algorithm consists of performing
$U_0, Q, U_1, \dots , Q, U_{\ell}$ on $|\psi_{start}\rangle$ and measuring the result.

The algorithm computes $f(X)$ if during the computation on an input instance $\sigma$ the initial state $|\psi_{start}\rangle$ is transformed to a final quantum state $\ket{\psi}$   that allows us to extract a value  $f(\sigma)$ when this state is measured. 
%Clearly, the  input $\sigma$ determines the final state. So, $\ket{\psi}$ can be presented as $\ket{\psi_\sigma}$. More precisely the final state is determined by $f(\sigma)$. So, $\ket{\psi}$ can be presented as $\ket{\psi(f(\sigma))}$ or $\ket{\psi_{f(\sigma)}}$. We will use all these notations in the paper.

%In the quantum case, the decision trees are interpreted as the Query model \cite{}. The Quantum query model is considered as one of the simplest and practically oriented models of quantum computation. Different authors intensively investigated it during the last decades.     
%In the paper, we assume that we have a quantum circuit that implements query operator $Q$. The operator is such that $a\ket{i}$ is converted to $(-1)^{f(i)}a\ket{i}$ for some Boolean value function $f$. We call such a circuit an Oracle. Typically researchers consider oracles that can implements functions with time complexity $O(\log d)$, where $d$ is the number of quantum states.

\subsection{Quantum Branching Program or Quantum Data Stream Processing Algorithms} \label{sec:qbp}

Quantum Branching Program (QBP) is a known model of computations -- a generalization of the classical Branching Program (BP) model of computations. We refer to papers \cite{agk01, agkmp2005,av2009, aakv2018,agky16,agky14} for more information  on QBPs.   Quantum Branching Programs and Quantum Query Algorithms are closely related. They can be considered as a specific variant of each other depending on the point of view. For example,  in the content of this work, QBP can be considered as a special variant of the Quantum Query Model, which can test only one input variable on a  level of computation.  Here we define QBP following to \cite{aakv2018}. Another point of view to the model is data stream processing algorithms.

%%%%%%%%%%%%%%%%%%%%%%%%%%%%%%%%%
A QBP  $A$ over the Hilbert space $\Hilbert^d$ is defined as
\begin{equation}
A=( Trans, \ket{\psi_0} %\textup{Accept}
),
\end{equation}
where $Trans$ is a sequence of $l$ instructions: $Trans_j=\left(x_{i_j}, U_j(0),U_j(1)\right)$ is
determined by the variable $x_{i_j}$ tested on the step $j$, and $U_j(0)$, $U_j(1)$ are unitary
transformations in $\Hilbert^d$, $d=dim(A)$. Here $l$ is some positive integer.

Vectors $\ket{\psi}\in \Hilbert^d$ are called states (state vectors) of $A$, $\ket{\psi_0}\in
\Hilbert^d$ is the initial state of $A$. 
%and $\textup{Accept}\subseteq\{1,2,\ldots,d\}$ is the set of indices of accepting basis states.

We define a computation of $A$ on an input instance $\sigma = (\sigma_1
\ldots \sigma_n) \in \{0,1\}^n$  as follows:
\begin{enumerate}
\item A computation of $A$ starts from the initial state $\ket{\psi_0}$.
\item The $j$-th instruction of $A$ reads the input symbol $\sigma_{i_j}$ (the
value of $x_{i_j}$) and applies the transition matrix $U_j = U_j(\sigma_{i_j})$ to the current
      state $\ket{\psi}$ for obtaining the state $\ket{\psi'}=U_j(\sigma_{i_j})\ket{\psi}$.
\item The final state is
\begin{equation}
 \ket{\psi_\sigma}= \left(\prod_{j=l}^1 U_j(\sigma_{i_j})\right)
\ket{\psi_0}\enspace.
\end{equation}
%\item After the $l$-th (last) step of quantum transformation the configuration $\ket{\psi_\sigma}=(\alpha_1,\ldots, \alpha_d)^T$ of $Q$ is measured, and the input $\sigma$ is accepted if and only if the result is in $\textup{Accept}$, which happens with probability
%\begin{equation}
%\textup{Pr}_{\textup{accept}}(\sigma)=\sum\limits_{i\in %\textup{Accept}}|\alpha_i|^2\enspace.
%\end{equation}

\end{enumerate}
%%%%%%%%%%%%%%%%%%%%%%%%%%%%%%%%%%%

%The definition of the algorithm $A$ is similar. It is an initial state $|\psi_{start}\rangle$ and transformations $G_1, \dots, G_{\ell}$, but we fix $\ell=N$. We associate an order of variables $\theta=(j_1,\dots,j_N)$ with the quantum branching program.  The $\theta$ is a permutation on $(1,\dots,N)$. We assume that all nodes of $i$-th level test only $x_{j_i}$. An initial state $|\psi_{start}\rangle$ does not depend on the $x_1\dots,x_n$ variables. The transformation $U_i$ is one of two unitary $w\times w$-matrices $U_i^0$ if $x_j=0$ and $U_i^1$  if $x_j=1$, where $w=dim(A)$. Typically, we assume that the algorithm knows the order $\theta$ for free.

%We can interpret a quantum branching program as a quantum circuit that has access to the classical input variables.

\begin{figure}[h]
    \centering
    \includegraphics[height=5cm]{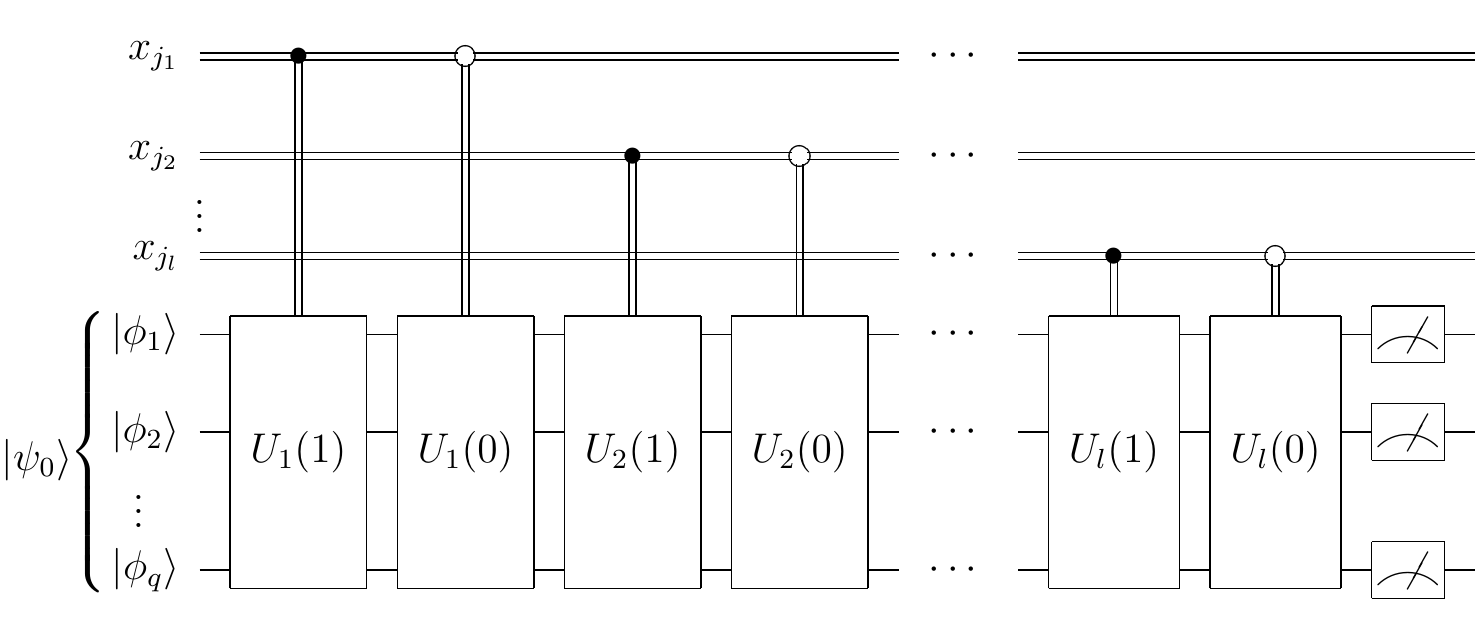}%{qobdd}
    \caption{Quantum branching program as a quantum circuit.}
    \label{fig:hgate}
\end{figure}
BPs and QBPs are convenient computational models in complexity theory. It is easy and natural to define various restricted models for this computational model. One we use here is a read-once model, which has the following restriction: ``each input variable is tested exactly once''. In this case, we have $\ell=N$, where $N$ is the number of input variables. In other words, the number of computational steps equals the number of input variables.

Let us discuss the programming-oriented definition of read-once QBP.
 We modify read-once QBP $P$  to the following  QBP $A$ by modifying its register.  
%
%Another interpretation is a query-model-like interpretation. Let us define the algorithm in a little bit different way. 
%
We equip basis states $S=\{\ket{1},\dots,\ket{d}\}$ by ancillary qubits as follows. A state $\ket{a}\in S$ is modified by adding state $\ket{k}$ and qubit $\ket{\phi}$, where $k$ is the index of variable $x\in X$ tested in the state $\ket{a}$  and $\ket{\phi}$ presents a  Boolean value of the input $x$ tested. The new basis states $S'$ are  
\[ S'=\{\ket{k}\ket{a}\ket{\phi}: k\in\{1,\dots, N\}, a\in\{1,\dots, d\}, \phi\in\{0,1\} \}.
\]

The initial state is  $\ket{j_1}\ket{\psi_{0}}\ket{0}$, where $\ket{\psi_{0}}$ is a starting state of $Q$. 
The transformations of $A$  are the following sequence of matrices  
\[ 
 Q, U_1, Q, T ,Q, U_2, Q, T ,Q,\dots , Q,  U_{N}
 \]
% with  $\ell=N$.
 
% We use a quantum register $\ket{\lambda}\ket{\psi}\ket{\varphi}$. $\ket{\varphi}$ is a single qubit, we enumerate all states of $\ket{\lambda}\ket{\psi}$ as $\ket{k}\ket{u}$, where $k\in\{1,\dots,N\}$,  $u\in\{1,\dots,w\}$, $w=dim(A)$. Here $k$ is an index of a testing variable, $u$ is an index of a node in the level, $\varphi$ is a value of a testing variable. 

 The matrix $T$ defines  a transition that changes the testing variable's index on the current level to the variable that is tested on the next level
 \[
 T:\ket{j_z}\ket{a}\ket{\phi}\to\ket{j_{z+1}}\ket{a}\ket{\phi}.
\] 
% for $z\in\{1,\dots,N-1\}$.

 The matrix $Q$ is a query such that  $Q:\ket{k}\ket{a}\ket{\phi}\to\ket{k}\ket{a}\ket{\phi\oplus x_{k}}$. 

 The $U_i$ is an operator from QBP $P$ that is oriented for acting on basis states of a quantum state. $U_i$
 applies $U_{i}(0)$ to $\ket{a}$ if %$\ket{\lambda}\ket{\varphi}=\ket{j_i}\ket{0}$
 $\ket{\phi}=\ket{0}$
 and applies $U_{i}(1)$ to $\ket{a}$ if %$\ket{\lambda}\ket{\phi}=\ket{j_i}\ket{1}$.
 $\ket{\phi}=\ket{1}$.
 We have to apply the query $Q$ twice for one step of the algorithm to obtain the state  $\ket{\phi}=\ket{0}$ before testing the next input variable. 

More precisely: before applying $Q$, the qubit $\ket{\phi}$ is in the state $\ket{0}$. The first applying of $Q$ converts $\ket{k}\ket{a}\ket{0}$ to state $\ket{k}\ket{a}\ket{x_{k}}$.  The second applying of $Q$ converts $\ket{k}\ket{a'}\ket{x_{k}}$ to $\ket{k}\ket{a'}\ket{0}$. 

\section{Grover's Search Algorithm and Amplitude Amplification}
In this section, we present an algorithm in the Query model for the Search problem.
Let us present the Search problem. 

    {\bf Search Problem} Given a set of numbers $(x_1,\dots,x_n)\in \{0,1\}^n$ one  want to find the $k\in[n]$ such that $x_k=1$.
    
Let $ONES=\{i:x_i=1\}$. Then, we can consider three types of the problem
\begin{itemize}
    \item {\bf single-solution search problem} if $|ONES|=1$ and it is known apriori;
    \item {\bf $t$-solutions search problem} if $|ONES|=t$ and it is known apriori;
    \item {\bf unknown-number-of-solutions search problem} if $|ONES|=t$ and $t$ is unknown.
\end{itemize}

Let us consider the solution of the first of these three problems.

\subsection{Single-solution Search Problem}
So, it is known that there is one and only one $k\in [n]$ such that $x_k=1$ and $x_i=0$ for any $i\in [n]\backslash \{k\}$.

\paragraph{Deterministic solution.}
We can check all elements $i\in[n]$, i.e. from $0$ to $n-1$. If we found $i$ such that $x_i=1$, then it is the solution. The idea is presented in Algorithm \ref{alg:dsearchsingle}.

\begin{algorithm}
\caption{Deterministic solution of Single-solution Search Problem}\label{alg:dsearchsingle}
\begin{algorithmic}
\State $k\gets NULL$
\For{$i\in\{0,\dots,n-1\}$}
\If{$x_i=1$}
\State $k\gets i$
\EndIf
\EndFor
\State \Return $k$
\end{algorithmic}
\end{algorithm}

The query complexity of the such solution is $O(n)$.

\paragraph{Randomized (probabilistic) solution.}
We pick any element $i$ uniformly. If $x_i=1$, then we win, otherwise we fail. The query complexity of the such algorithm is $1$. At the same time, it has a very small success probability that is $Pr_{success}=\frac{1}{n}$.
We call this algorithm {\bf random sampling}.

We can fix the issue with the small success probability using the ``Boosting success probability'' technique. That is the repetition of the algorithm several times. We repeat the algorithm $r$ times. If at least one invocation wins, then we find the solution. If each invocation fails, then we fail.  The idea is presented in Algorithm \ref{alg:rsearchsingle}.

\begin{algorithm}
\caption{Randomized solution of Single-solution Search Problem}\label{alg:rsearchsingle}
\begin{algorithmic}
\State $k\gets NULL$
\State $step\gets 0$
\While{$step<r$ and $k=NULL$}
\State $i\in_R [n]$
\If{$x_i=1$}
\State $k\gets i$
\EndIf
\State $step\gets step+1$
\EndWhile
\State \Return $k$
\end{algorithmic}
\end{algorithm}

Let us compute error probability after $r$ repetitions.
\begin{lemma}\label{lm:boost1}
The error probability of Random Sampling algorithm's $r$ repetitions is $\left(1-\frac{1}{n}\right)^r$.
\end{lemma}
\Beginproof
Each event of an invocation's error is independent of another. The error probability for one invocation is $P_{error}=\left(1-\frac{1}{n}\right)$. $r$ repetitions have an error if all invocations have errors. Therefore, the total error probability is  
$\left(1-\frac{1}{n}\right)^r.$
\Endproof

 If we want to have constant (at most $\frac{1}{3}$) error probability, then $t=O(n)$.

\begin{lemma}
The error probability of the Random Sampling algorithm's $O(n)$ repetitions is at most $\frac{1}{3}$.
\end{lemma}
\Beginproof
Let $r=3n$
Due to Lemma \ref{lm:boost1}, the error probability is 

\[P_{error}^{RSB}=\left(1-\frac{1}{n}\right)^{3n}\]

It is known that 
\[\lim_{n\to\infty}P_{error}^{RSB}=\lim_{n\to\infty}\left(1-\frac{1}{n}\right)^{3n}=\frac{1}{e^3}\approx0.05=\frac{1}{20}.\]

So, even for small $n>2$ we have $P_{error}^{RSB}<1/3$.
\Endproof

So, the total query complexity of the algorithm is $O(n)$.  It is known that $\Omega(n)$ is also lower bound for a randomized algorithm \cite{bbbv1997}. The same is true for deterministic algorithms because deterministic algorithms are a particular case of randomized algorithms.

Let us discuss the quantum algorithm for the problem called Grover's Search Algorithm\cite{g96,bbht98}. In fact, the algorithm is a quantum version of the described randomized algorithm.

%%%%%%%%%%%%%%%%%%%%%%%%%%%%%%%%%%%%%%%%%%%%%%%
%%%   Grover single one %%%%%%%%%%%%%%%%%%%%%%%
%%%%%%%%%%%%%%%%%%%%%%%%%%%%%%%%%%%%%%%

\subsubsection{Grover's Search Algorithm for Single-solution Search Problem}\label{sec:grover}
Let us consider two quantum registers. The quantum register $\ket{\psi}$ has $\log_2 n$ qubits. It stores an argument. The quantum register $\ket{\phi}$ has one qubit. It is an auxiliary qubit.

The access to oracle is presented by the following transformation:
\[Q:|i\rangle|\phi\rangle \to |i\rangle|\phi\oplus x_i\rangle\]
The transformation is implemented by the control NOT gate, which has $\ket{\psi}$ and $x_i$ as the control qubits and $|\phi\rangle$ as the target qubit.

\paragraph{``Inverting the sign of an amplitude'' Procedure.} The procedure is also known as a partial case of the phase kickback trick.
Firstly, we prepare $\ket{\phi}=\frac{\ket{0}-\ket{1}}{\sqrt{2}}$ state. 
Initially, the $\ket{\phi}$ qubit in $\ket{0}$ state. Then we apply $NOT$ or $X$ gate 
\[NOT\ket{0}=\ket{1}\]
After that we apply  Hadamard transformation $H$ 
\[H\cdot NOT\ket{0}= H\cdot \ket{1}=\frac{\ket{0}-\ket{1}}{\sqrt{2}}\]

\begin{property}
Suppose $\ket{\phi}=\frac{\ket{0}-\ket{1}}{\sqrt{2}}$.  Then the operator $Q$ performs the ``Inverting the sign of an amplitude'' procedure. That  is,  
\[
Q: \ket{i}\ket{\phi} \mapsto (-1)^{x_i}\ket{i}\ket{\phi}. 
\]   
\end{property}
\begin{proof}
If $x_i=0$, then  $\ket{i}\ket{\phi\oplus x_i}= \ket{i}\ket{\phi}= (-1)^0\ket{i}\ket{\phi}$.
If $x_i=1$, then 
\[
\ket{i}\ket{\phi\oplus x_i}=\ket{i}\ket{\phi\oplus 1}= \ket{i}\frac{\ket{0\oplus 1}-\ket{1\oplus 1}}{\sqrt{2}}=\]\[= \ket{i}\frac{\ket{1}-\ket{0}}{\sqrt{2}}=-\ket{i}\frac{\ket{0}-\ket{1}}{\sqrt{2}}=(-1)^1\ket{i}\ket{\phi}.
\]
Therefore, $\ket{i}\ket{\phi\oplus x_i}=(-1)^{x_i}\ket{i}\ket{\phi}$
\end{proof}

\paragraph{Description of the Algorithm}
The initial step of the algorithm is the following.
\begin{itemize}
\item The initial state is $\ket{\psi}\ket{\phi}=\ket{0}\ket{0}$.
\item We apply  $NOT$ and $H$ gates to $\ket{\phi}$ for preparing $\ket{\phi}=\frac{\ket{0}-\ket{1}}{\sqrt{2}}$ state.
\item We apply $H^{\otimes \log_2 n}$ to $\ket{\psi}$ for  preparing $\ket{\phi}=\frac{1}{\sqrt{n}}\sum_{i=1}^n |i\rangle$ state. It means we apply $H$ transformation to each qubit of $\ket{\psi}$.
\end{itemize} 

The main step of the algorithm is the following:
\begin{itemize}
\item We do the query to oracle $Q$ that inverts the amplitude of the target element:
\[Q:a_k\ket{k}\to -a_k\ket{k},\]
and does nothing to other elements.

We can say that the matrix $Q$ is the following one

$Q=\begin{pmatrix} 
1&0&\dots&0&0&0&\dots&0\\
0&1&\dots&0&0&0&\dots&0\\
\dots&\dots&\dots&\dots&\dots&\dots&\dots&\dots\\
0&0&\dots&1&0&0&\dots&0\\
0&0&\dots&0&-1&0&\dots&0\\
0&0&\dots&0&0&1&\dots&0\\
\dots&\dots&\dots&\dots&\dots&\dots&\dots&\dots\\
0&0&\dots&0&0&0&\dots&1\\
\end{pmatrix}$

At the same time, we remember that it is implemented using the ``Inverting the sign of an amplitude'' procedure.
\item We apply Grover's Diffusion transformation $D$ that rotates all amplitudes near a mean. Let us describe the operation in details.
Let $\ket{\psi}=\sum_{i=0}^{n-1}a_i\ket{i}$ and $m = \frac{1}{n}\sum_{i=0}^{n-1}a_i$. Then
\[D:a_i\ket{i} \to (2m-a_i)\ket{i}\]

We can say that the matrix $D$ is the following one

$D=\begin{pmatrix} 
\frac{2}{n}-1&\frac{2}{n}&\dots&\frac{2}{n}\\
\frac{2}{n}&\frac{2}{n}-1&\dots&\frac{2}{n}\\
\dots&\dots&\dots&\dots\\
\frac{2}{n}&\frac{2}{n}&\dots&\frac{2}{n}-1
\end{pmatrix}$

We will discuss this operation in details later in this section.
\end{itemize} 

If we measure $\ket{\psi}$ after the initial step, then we can obtain the solution $\ket{\psi}=\ket{k}$ with a very small probability $1/n$.
After one main step, you can see that $a_k>a_i$ for any $i\in[n]\backslash\{k\}$, but the probability is still very small.

For obtaining a good success probability, we should repeat the main step $O(\sqrt{n})$ times. After that, if we measure $\ket{\psi}$, then we obtain the solution $\ket{\psi}=\ket{k}$ with a probability close to $1$.

Let us discuss the complexity of the algorithm and explain the above claims in details.

\paragraph{Query Complexity and Error Probability} 
Let us consider the state of $\ket{\psi}$ after the initial step
\[|\psi\rangle=\frac{1}{\sqrt{n}}\sum\limits_{i\in[n]}|i\rangle =\frac{1}{\sqrt{n}}\sum\limits_{i\in[n]\backslash \{k\}}|i\rangle+ \frac{1}{\sqrt{n}}|k\rangle\]

After $j$ main steps we have the state

\[|\psi\rangle=B_{(j)}\sum\limits_{i\in\{1,\dots,n\}\backslash \{k\}}|i\rangle+ G_{(j)}|k\rangle\]

Here $B_{(j)}$ is the amplitude for non-target states and $G_{(j)}$ is the amplitude for the target state. You can verify that if we apply $Q$ and $D$, then the amplitudes of all non-target states are equal.

We know that 
\[1=B_{(j)}^2+\dots +B_{(j)}^2+G_{(j)}^2=(n-1)B_{(j)}^2 + G_{(j)}^2=\left(\sqrt{n-1}B_{(j)}\right)^2 + G_{(j)}^2\]
In other words:
\[ \left(\sqrt{n-1}B_{(j)}\right)^2 + G_{(j)}^2= 1\]

Remember that all initial amplitudes and transformations are real valued. that is why we can say about numbers themselves, but not about absolute values of complex numbers.

We can choose an angle $\alpha$ such that $\sin\alpha = G_{(j)}$ and $\cos\alpha = \sqrt{n-1}B_{(j)}$.
This angle represents the state $\ket{\psi}$.

After the initial step $G_{0}=\frac{1}{\sqrt{n}}$. Therefore, $\sin \alpha = \frac{1}{\sqrt{n}}$ and $\alpha=\arcsin{\frac{1}{\sqrt{n}}}$. Let us define $\theta=\arcsin{\frac{1}{\sqrt{n}}}$.

\begin{figure}[h]
\begin{center}
\includegraphics[height=5cm]{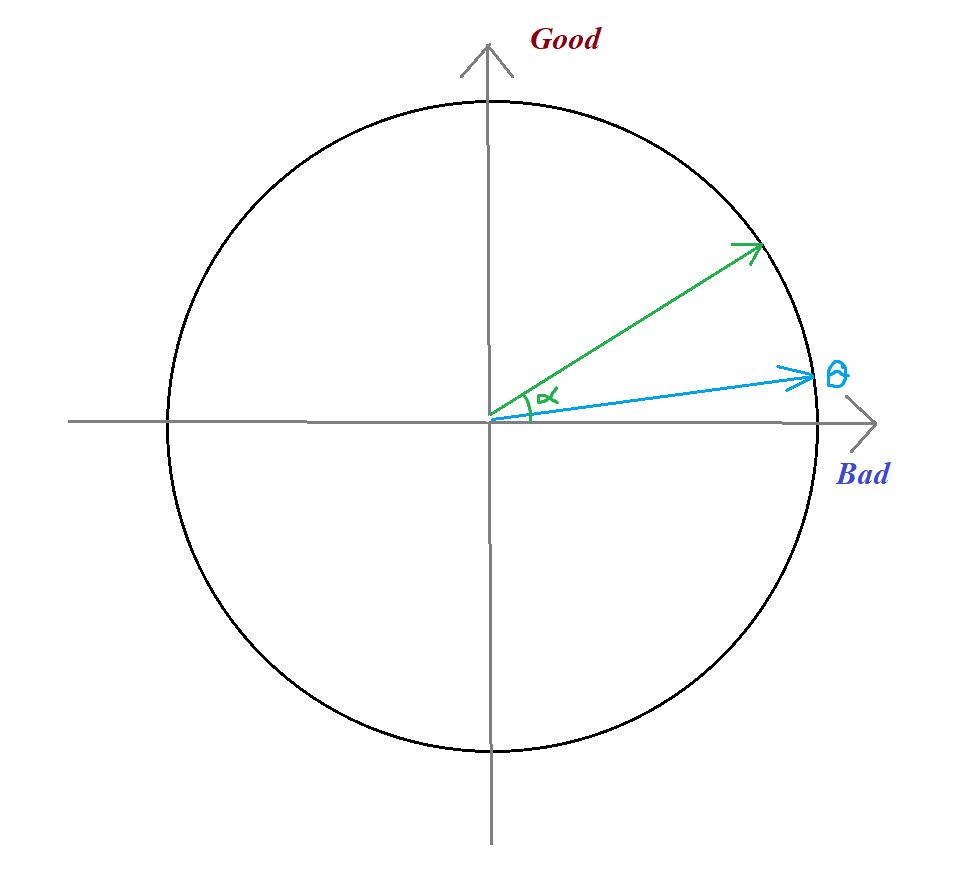}
\caption{The angle $\alpha$ that represents $\ket{\psi}$}
\end{center}
\end{figure}

What happens with $\alpha$ if we apply $Q$ or $D$ transformations?

Let us start with the $Q$ transformation.

\[Q:B_{(j)}\sum\limits_{i\in[n]\backslash \{k\}}|i\rangle+ G_{(j)}|k\rangle \to B_{(j)}\sum\limits_{i\in[n]\backslash \{k\}}|i\rangle - G_{(j)}|k\rangle \]

So, only the amplitude $G_{(j)}$ is changed. We can say that the transformation is 
\[Q:\sin\alpha \to -\sin\alpha = \sin(-\alpha)\]

Due to $\cos\alpha=\cos(-\alpha)$, it is equivalent to transformation
\[Q:\alpha \to -\alpha.\]
This is a reflection near the $0$ angle.

\begin{figure}[h]
\begin{center}
\includegraphics[height=5cm]{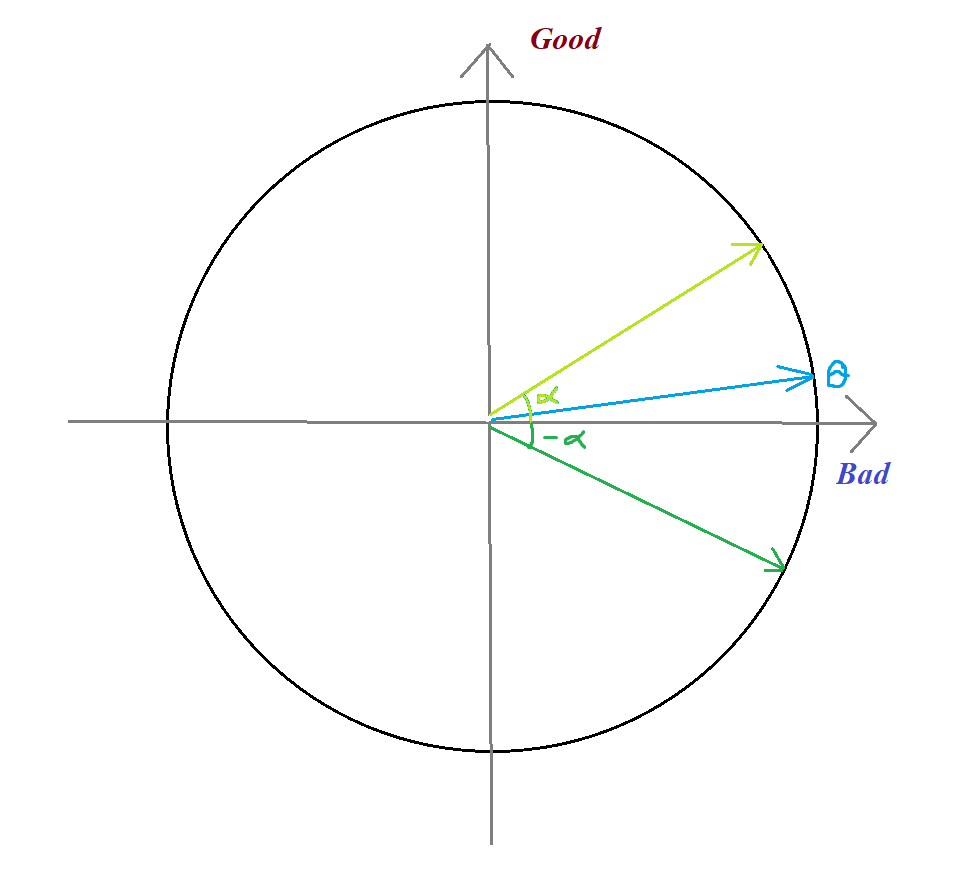}
\caption{The effect of a query to the angle $\alpha$ }
\end{center}
\end{figure}

Let us continue with the $D$ transformation. The matrix of $D$ is following

$D=\begin{pmatrix} 
\frac{2}{n}-1&\frac{2}{n}&\dots&\frac{2}{n}\\
\frac{2}{n}&\frac{2}{n}-1&\dots&\frac{2}{n}\\
\dots&\dots&\dots&\dots\\
\frac{2}{n}&\frac{2}{n}&\dots&\frac{2}{n}-1
\end{pmatrix}=2\begin{pmatrix} 
\frac{1}{n}&\frac{1}{n}&\dots&\frac{1}{n}\\
\frac{1}{n}&\frac{1}{n}&\dots&\frac{1}{n}\\
\dots&\dots&\dots&\dots\\
\frac{1}{n}&\frac{1}{n}&\dots&\frac{1}{n}
\end{pmatrix}-I^{\otimes \log n}$.

Let us denote $|\Psi\rangle=\begin{pmatrix} 
\frac{1}{\sqrt{n}}\\
\dots\\
\frac{1}{\sqrt{n}}
\end{pmatrix}$ and $\langle\Psi|=(|\Psi\rangle)^*=(\frac{1}{\sqrt{n}},\dots,\frac{1}{\sqrt{n}})$.
We can see that 

\[D=2 |\Psi\rangle\langle\Psi| - I^{\otimes \log n}\]

Let $|\Psi'\rangle$ be the orthogonal vector for $|\Psi\rangle$, i.e. $\langle\Psi||\Psi'\rangle=0$. Then we can represent any vector $\ket{\varphi}$ as a linear combination of $|\Psi\rangle$ and $|\Psi'\rangle$.
\[|\varphi\rangle=a|\Psi\rangle + a'|\Psi'\rangle\]

If we apply $D$ matrix to $\ket{\varphi}$, then we obtain the following result:
\[D|\varphi\rangle=(2 |\Psi\rangle\langle\Psi| - I)(a|\Psi\rangle + a'|\Psi'\rangle)=
2a|\Psi\rangle\langle\Psi||\Psi\rangle + 2a'|\Psi\rangle\langle\Psi||\Psi'\rangle-a|\Psi\rangle - a'|\Psi'\rangle=\]
Remember that $\langle\Psi||\Psi'\rangle=0$ by definition of $|\Psi'\rangle$ and $\langle\Psi||\Psi\rangle=\sum_{i=0}^{n-1}\frac{1}{\sqrt{n}}\cdot\frac{1}{\sqrt{n}}=n\cdot \frac{1}{n}=1$. Therefore,
\[=2a|\Psi\rangle-a|\Psi\rangle - a'|\Psi'\rangle=a|\Psi\rangle - a'|\Psi'\rangle\]

In other words, $D$ rotates near $0$ in the coordinate system with $|\Psi\rangle$ and $|\Psi'\rangle$ as directions of the coordinate axis.

Remember that $|\Psi\rangle=\begin{pmatrix} 
\frac{1}{\sqrt{n}}\\
\dots\\
\frac{1}{\sqrt{n}}
\end{pmatrix}$. It means $|\Psi\rangle=H^{\otimes\log n}\ket{0}$ that is the state after the initial step. So, we can see that $|\Psi\rangle$ corresponds to the $\theta$ angle and $|\Psi'\rangle$ corresponds to $\frac{\pi}{2}+\theta$ angle. Therefore, the $D$ transformation is the rotation near the $\theta$ angle. 

\begin{figure}[h]
\begin{center}
\includegraphics[height=5cm]{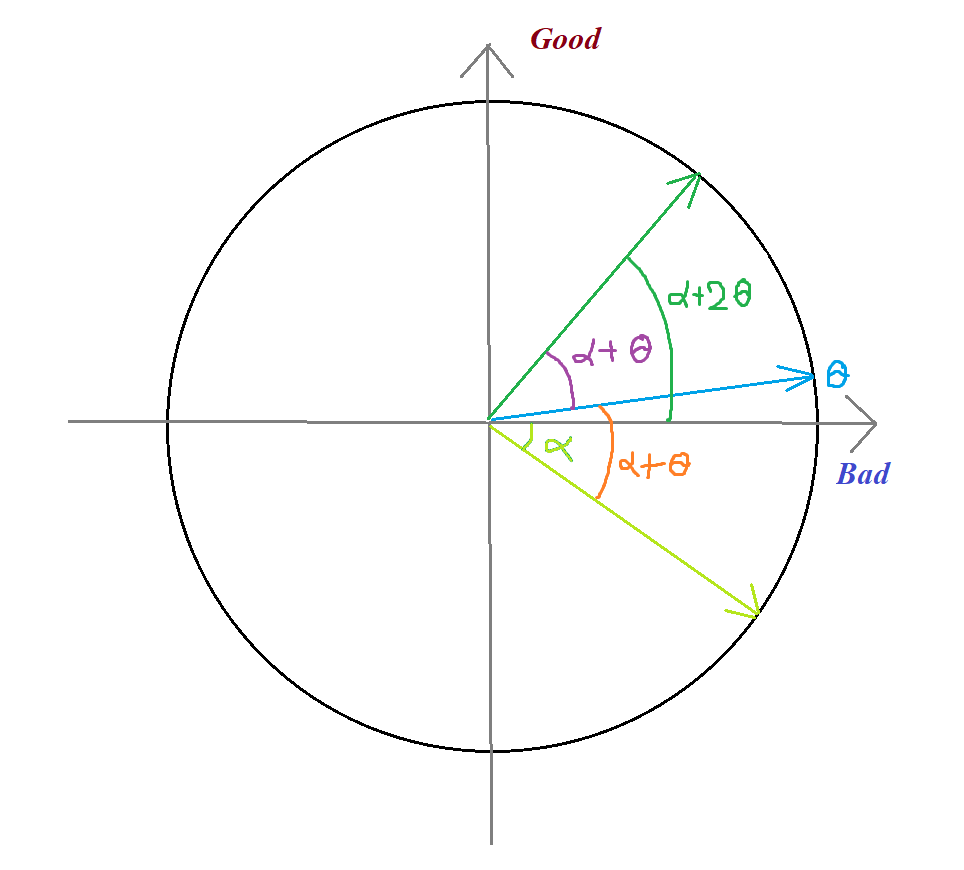}
\caption{The effect of query and diffusion to the angle $\alpha$ }
\end{center}
\end{figure}

If we consider the main step which is applying $Q$ and $D$, then we can see that
\[Q:\alpha\to-\alpha\]
\[D:-\alpha\to \alpha+2\theta\]
for $\alpha\geq \theta$.
Finally, the main step is equivalent to increasing $\alpha$ to $2\theta$.

On the zeroth (initial) step of the algorithm $\alpha\gets\theta$.
On the $i$-th step of the algorithm $\alpha=(2i+1)\theta$.
We stop on the $L$-th step of the algorithm when $\alpha=(2L+1)\theta\approx\frac{\pi}{2}$.

Due to $\lim_{\beta\to 0}\frac{\sin(\beta)}{\beta}=1$, we can say that $\theta \approx sin(\theta)$ for small $\theta$. Therefore,
\[ \frac{\pi}{2}=(2L+1)\theta \approx (2L+1)sin(\theta)=(2L+1)\frac{1}{\sqrt{n}},\]
and $L\approx \frac{\pi}{4}\sqrt{n}$.

In $L$ steps $G_{(L)}\approx\sin(\pi/2)=1$ and $B_{(L)}\approx\cos(\pi/2)=0$. The state is
\[\ket{\psi}\approx\ket{k}\]
If we measure the $\ket{\psi}$ register, then we obtain $k$ with almost $1$ probability.

You can see, that $L$ is an integer. Maybe we cannot find an integer $L$ such that $(2L+1)\theta=\pi/2$. But we can reach the difference between $\alpha$ and $\pi$ at most $\theta$ because one step size is $2\theta$. Therefore, the error probability is $B_{(L)}^2\leq cos(\frac{\pi}{2}-\theta)^2=sin(\theta)^2 =\left(\frac{1}{\sqrt{n}}\right)^2=\frac{1}{n}$.

%%%%%%%%%%%%%%%%%%%%%%%%%%%%%%%%%%%%%%%%%%%%
%           $t$-solutions Search Problem%%%%
%%%%%%%%%%%%%%%%%%%%%%%%%%%%%%%%%%%%%%%%%%%%
\subsection{$t$-solutions Search Problem}\label{sec:grover-fixed-t}

\paragraph{Deterministic solution.}
We can use the same deterministic algorithm (linear search) as for {\em single-solution search problem}. In the worst case, all solutions (indexes of ones) are at the end of the search space. The query complexity of such a solution is $O(n-t)=O(n)$ for a reasonable $t$.

\paragraph{Randomized (probabilistic) solution.}
We can use the same randomized algorithm (random sampling) as for {\em single-solution search problem} that is picking an element $i$ uniformly. Here we have a better success probability that is $Pr_{success}=\frac{t}{n}$. Using the ``Boosting success probability'' technique, we obtain $O(n/t)$ query complexity.

%%%%%%%%%%%%%%%%%%%%%%%%%%%%%%%%%%%%%%%%%%%%%%%
%%%   Grover t ones %%%%%%%%%%%%%%%%%%%%%%%
%%%%%%%%%%%%%%%%%%%%%%%%%%%%%%%%%%%%%%%

\subsubsection{Grover's Search Algorithm for $t$-solutions Search Problem}
Let $K=(k_1,\dots,k_t)$ be indexes of solutions, i.e. $\{i:x_i=1, i\in[n]\}=K$.
We use the same algorithm as for {\em single-solution search problem}. At the same time, you can see that the access to Oracle is changed. In fact, it is the same, but the ``Inverting the sign of an amplitude'' procedure  can be represented by another matrix $O$ that is 

\[Q=\begin{pmatrix} 
1&0&\dots&0&0&0&\dots&0\\
0&1&\dots&0&0&0&\dots&0\\
\dots&\dots&\dots&\dots&\dots&\dots&\dots&\dots\\
0&0&\dots&1&0&0&\dots&0\\
0&0&\dots&0&-1&0&\dots&0\\
0&0&\dots&0&0&1&\dots&0\\
\dots&\dots&\dots&\dots&\dots&\dots&\dots&\dots\\
0&0&\dots&0&0&0&\dots&1\\
\end{pmatrix},\]

 where $-1$s are situated in lines with indexes from $K$.
Grover's Diffusion transformation $D$ is not changed.

For obtaining good success probability, we should repeat the main step $O(\sqrt{\frac{n}{t}})$ times that is the difference with the ``single-solution'' case.

Let us discuss the complexity of the algorithm and explain the above claims in details.

\paragraph{Query Complexity and Error Probability} 
We have an analysis that is similar to the ``single-solution'' case.

Let us consider the state of $\ket{\psi}$ after the initial step
\[|\psi\rangle=\frac{1}{\sqrt{n}}\sum\limits_{i\in[n]}|i\rangle =\frac{1}{\sqrt{n}}\sum\limits_{i\in[n]\backslash K}|i\rangle+ \frac{1}{\sqrt{n}}\sum\limits_{i\in K}|i\rangle\]

After $j$ main steps we have the state

\[|\psi\rangle=B_{(j)}\sum\limits_{i\in\{1,\dots,n\}\backslash K}|i\rangle+ G_{(j)} \sum\limits_{i\in K}|i\rangle\]

Here $B_{(j)}$ is the amplitude for non-target states and $G_{(j)}$ is the amplitude for target states. You can verify that if we apply $Q$ and $D$, then amplitudes of all non-target states are equal and for all target states are equal also.
	
We know that 
\[1=B_{(j)}^2+\dots +B_{(j)}^2+G_{(j)}^2+\dots+G_{(j)}^2=(n-t)B_{(j)}^2 + tG_{(j)}^2=\left(\sqrt{n-t}B_{(j)}\right)^2 + \left(\sqrt{t}G_{(j)}\right)^2\]
In other words:
\[ \left(\sqrt{n-t}B_{(j)}\right)^2 + \left(\sqrt{t}G_{(j)}\right)^2= 1\]

We can choose an angle $\alpha$ such that $\sin\alpha = \sqrt{t}G_{(j)}$ and $\cos\alpha = \sqrt{n-t}B_{(j)}$.
This angle represents the state $\ket{\psi}$.

After the initial step $G_{(0)}=\sqrt{\frac{1}{n}}$. Therefore, $\sin \alpha = \sqrt{\frac{t}{n}}$ and $\alpha=\arcsin{\sqrt{\frac{t}{n}}}$. Let us define $\theta=\arcsin{\sqrt{\frac{t}{n}}}$.

What happens with $\alpha$ if we apply $Q$ or $D$ transformations?

Let us start with the $Q$ transformation.

\[Q:B_{(j)}\sum\limits_{i\in[n]\backslash K}|i\rangle+ G_{(j)}\sum\limits_{i\in K}|i\rangle \to B_{(j)}\sum\limits_{i\in[n]\backslash \{k\}}|i\rangle - G_{(j)}\sum\limits_{i\in K}|i\rangle \]

So, only the amplitudes $G_{(j)}$ of states with indexes from $K$ were changed. Similarly to the single solution case, we can say that the transformation is a reflection near the $0$ angle.
\[Q:\alpha \to -\alpha.\]

Let us continue with the $D$ transformation. It is still the rotation near  $|\Psi\rangle=\begin{pmatrix} 
\frac{1}{\sqrt{n}}\\
\dots\\
\frac{1}{\sqrt{n}}
\end{pmatrix}$ vector. Remember that  $|\Psi\rangle=H^{\otimes\log n}\ket{0}$, and in the new settings $|\Psi\rangle$ corresponds to $\arcsin{\sqrt{\frac{t}{n}}}=\theta$. So, the $D$ transformation is rotation near the new $\theta$ angle. 

We have the same actions for $Q$ and $D$:
$Q:\alpha\to-\alpha$ and 
$D:-\alpha\to \alpha+2\theta$ for $\alpha\geq \theta$.
Hence, the main step is equivalent to increasing $\alpha$ to $2\theta$.

We stop on the $L$-th step of the algorithm when $\alpha=(2L+1)\theta\approx\frac{\pi}{2}$. In that case

\[ \frac{\pi}{2}=(2L+1)\theta \approx (2L+1)sin(\theta)=(2L+1)\sqrt{\frac{t}{n}},\]
and $L\approx \frac{\pi}{4}\sqrt{\frac{n}{t}}$.

In $L$ steps $\sqrt{t}G_{(L)}\approx\sin(\pi/2)=1$ and $\sqrt{n-t}B_{(L)}\approx\cos(\pi/2)=0$. The state is
\[\ket{\psi}\approx\frac{1}{\sqrt{t}}\sum_{i\in K}\ket{i}\]
If we measure the $\ket{\psi}$ register, then we obtain elements from $K$ with almost $1$ probability. Each element from $K$ can be obtained with equal probability.

Similarly to the single-solution case,  $L$ is an integer. So, the error probability is $B_{(L)}^2\leq cos(\frac{\pi}{2}-\theta)^2=sin(\theta)^2 =\left(\sqrt{\frac{t}{n}}\right)^2=\frac{t}{n}$.

%%%%%%%%%%%%%%%%%%%%%%%%%%%%%%%%%%%%%%%%%%%%
%          unknown-number-of-solutions Search Problem%%%%
%%%%%%%%%%%%%%%%%%%%%%%%%%%%%%%%%%%%%%%%%%%%
\subsection{Grover's Search Algorithm for unknown-number-of-solutions Search Problem}\label{sec:grover-t}
If we do not know $t$, then we do not know $\theta$. So, we cannot compute $L$. If we do too many steps, then the angle can become $\pi$ which corresponds to a very small probability of obtaining a target state.
Therefore, we should know the moment for stopping. 

We use a classical technique that is close to the Binary search algorithm or Binary lifting technique.

\begin{itemize}
\item {\bf Step 1}. We invoke Grover's Search Algorithm with $1=2^0$ main step. Then, we measure the $\ket{\psi}$ register and obtain some state $\ket{i}$. If $x_i=1$, then we win. If $x_i=0$, then we continue. 
\item {\bf Step 2}. We invoke Grover's Search Algorithm with $2=2^1$ main steps. Then, we measure the $\ket{\psi}$ register and obtain some state $\ket{i}$. If $x_i=1$, then we win. If $x_i=0$, then we continue. 
\item {\bf Step 3}. We invoke Grover's Search Algorithm with $4=2^2$ main steps. Then, we measure the $\ket{\psi}$ register and obtain some state $\ket{i}$. If $x_i=1$, then we win. If $x_i=0$, then we continue. 
\item {\bf Step $j$}. We invoke Grover's Search Algorithm with $2^{j-1}$ main steps. Then, we measure the $\ket{\psi}$ register and obtain some state $\ket{i}$. If $x_i=1$, then we win. If $x_i=0$, then we continue. 
\item {\bf Step $\log_2\left(\frac{\pi}{4}\sqrt{n}\right)+1$}. We invoke Grover's Search Algorithm with $\frac{\pi}{4}\sqrt{n}$ main steps. Then, we measure the $\ket{\psi}$ register and obtain some state $\ket{i}$. If $x_i=1$, then we win. If $x_i=0$, then there are no target elements. In fact, we take the closest power of $2$ to $\frac{\pi}{4}\sqrt{n}$ for number of main steps, that is $2^{\lceil\log_2 (\frac{\pi}{4}\sqrt{n}) \rceil}$. 
\end{itemize}

Assume that if we know $t$, then we should do $L=O(\sqrt{\frac{n}{t}})$ steps. Let us look at the step $j$ such that $2^{j-1}\leq L < 2^j$. Therefore, the angle $\alpha_j=(2\cdot 2^{j-1}+1)\theta$ such that $(2\cdot \frac{L}{2}+1)\theta<\alpha_j \leq (2\cdot L+1)\theta$.

\begin{figure}[h]
\begin{center}
\includegraphics[height=5cm]{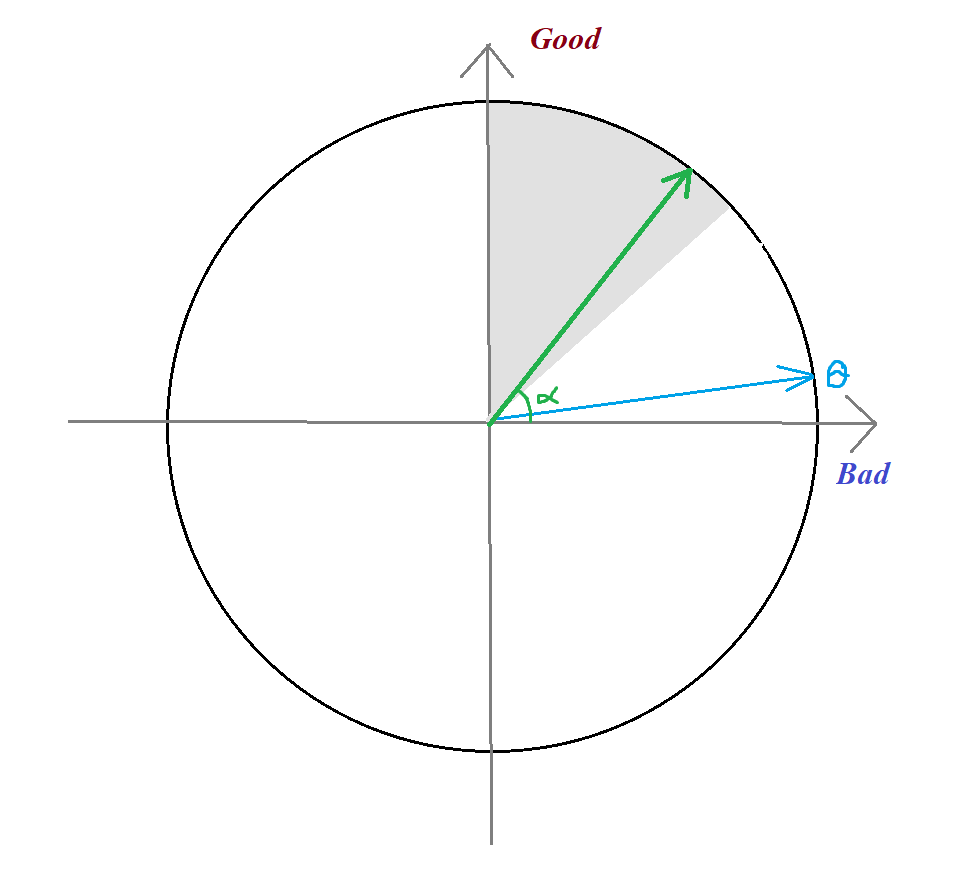}
%\caption{Qubit}
\end{center}
\end{figure}

So for this $j$ we can see that $\frac{1}{2}<(\sin(\alpha_j))^2\leq 1$. Therefore, with at least $0.5$ success probability we obtain the target element in Step $j=\log_2(\sqrt{\frac{n}{t}})+1$. If there is no solution, then we do all steps.

Let us compute the query complexity of steps from $1$ to $j$. It is $1+2+\dots+2^{j-1}=2^j-1$.
So we can say, that if there are $t$ solutions, then the expected number of queries is $O(\sqrt{\frac{n}{t}})$. If there are no solutions, then the number of queries is $O(\sqrt{n})$. 
\subsection{``Oracle generalization'' of Grover's Search Algorithm}
The main step of the algorithm consists of two matrices - $D$ and $O$. We discuss generalizations for both matrices one by one. Let us start with the generalization of the oracle.

Let us consider a more general search problem.

 {\bf General Search Problem} Given a function $f:\{0,\dots,n-1\}\to \{0,1\}$ one want to find an argument $x\in[n]=\{0,\dots,n-1\}$ such that $f(x)=1$.
 
 We want to use Grover's Search algorithm for the problem. Let us look at two quantum registers $\ket{\psi}$ and $\ket{\phi}$. The quantum register $\ket{\psi}$ of $\log_2 n$ qubits holds an index of an element in the original algorithm. Here we assume that $x_i=f(i)$ and $\ket{\psi}$ holds a possible argument of the function $f$. As in the original algorithm  $\ket{\phi}=\frac{\ket{0}-\ket{1}}{\sqrt{2}}$. If the query $Q$ has the following form, then we can use other parts of Grover's Search Algorithm as is.
\[Q:|i\rangle|\phi\rangle \to |i\rangle|\phi\oplus f(i)\rangle.\]
At the same time, in the original algorithm, we had the ``excluding or'' ($\oplus$) operation with a variable that can be implemented by the controlled-NOT gate. Here we should compute $f(i)$ first.
Assume that we have two additional registers that are $\ket{\varphi}$ of one qubit and $\ket{\zeta}$.

Let us have a unitary $U_{f(i)}$ for computing the function $f$. The register $\ket{\varphi}$ holds result of the function $f(i)$ and $\ket{\zeta}$ is additional memory for computing $f(i)$
So $U_{f(i)}$ is such that $U_{f(i)}:\ket{0}\ket{0}\to \ket{f(i)}\ket{\zeta'_i}$ for some specific state $\ket{\zeta'_i}$.
Let $U_{f(i)}^*=U_{f(i)}^{-1}$; in other words $U_{f(i)}:\ket{f(i)}\ket{\zeta'_i}\to\ket{0}\ket{0}$.

Let $U_f$ be the unitary that apply controlled operator $U_{f(i)}$ to $\ket{\varphi}\ket{\zeta}$ depending on the control register $\ket{\psi}$. So,
\[U_f:|i\rangle\ket{0}\ket{0} \to |i\rangle\ket{f(i)}\ket{\zeta'_i}.\]
\[U_f^{-1}:|i\rangle\ket{f(i)}\ket{\zeta'_i}\to|i\rangle\ket{0}\ket{0}.\]

For computing $Q$ we do the following actions:
\begin{itemize}
\item{\bf Step 1}  We apply $U_f$ to $\ket{\varphi}\ket{\zeta}$  as target registers and $\ket{\psi}$ as a control  register. After that $\ket{i}\ket{\phi}\ket{\varphi}\ket{\zeta}=|i\rangle\ket{\phi}\ket{f(i)}\ket{\zeta'_i}$.
\item{\bf Step 2}  We apply controlled NOT operator to $\ket{\phi} $ as a target and $\ket{\psi}\ket{\varphi}$ as a control register:
\[\ket{i}\ket{\phi}\ket{f(i)}\ket{\zeta'_i}\to \ket{i}\ket{\phi\oplus f(i)}\ket{f(i)}\ket{\zeta'_i}\]
\item{\bf Step 3}  We apply  $U_f^{-1}$ to $\ket{\varphi}\ket{\zeta}$  as target registers and $\psi$ as a control  register. After that $\ket{i}\ket{\phi}\ket{\varphi}\ket{\zeta}=|i\rangle\ket{\phi\oplus f(i)}\ket{0)}\ket{0}$.
\end{itemize}

The total sequence of actions is:
\[\ket{i}\ket{\phi}\ket{0}\ket{0}\xrightarrow{U_f}|i\rangle\ket{\phi}\ket{f(i)}\ket{\zeta'_i}\xrightarrow{\text{Controlled }NOT}\ket{i}\ket{\phi\oplus f(i)}\ket{f(i)}\ket{\zeta'_i}\xrightarrow{U_f^{-1}}|i\rangle\ket{\phi\oplus f(i)}\ket{0)}\ket{0}\]

So, we can use Grover's Search Algorithm but replace access to $x_i$ with access to $f(i)$.

Let us discuss the property of $f$ that is required for the algorithm.
\begin{itemize}
\item We can compute $f(i)$ with probability $1$ (with no errors) for each $i\in\{0,\dots,n-1\}$.
\item The algorithm for computing $f$ is reversible. In other words, $U_f$ and $U_f^{-1}$ exist.
\end{itemize}

Both restrictions can be removed using specific techniques. We discuss it in the next sections. Let us start with the second of two properties.
\subsubsection{Representing an Algorithm as a Unitary Transformation}\label{sec:dettoquant}
Assume that the algorithm has a non-reversible operator $g(x)$, where $x$ is the argument of the operator. Then, we can replace it with the operator $x\to(g(x),x)$. In other words, we can store not only the result of the operator but additionally the argument of the operator too. So, using the information we can implement the reverse operator that is $(g(x),x)\to x$. There is a more complex technique that requires less additional memory and is resented in the paper \cite{btv2001}.

Let us discuss representing a reversible probabilistic (randomized) algorithm as a quantum operator. Note, that any probabilistic (randomized) algorithm can be represented as a sequence of stochastic matrices applied to a vector of a probability distribution for states of memory. 
If we have a reversible probabilistic algorithm, then we can quantumly emulate it using the following steps:
\begin{itemize}
 \item The probabilistic distribution $p=(p_1,\dots,p_d)$ is replaced by a quantum state $|\psi\rangle=(\sqrt{p_1},\dots,\sqrt{p_d})$.
\item a stochastic transition matrix $M$ is replace by  a unitary $U$, where $U_{i,j}=\sqrt{M_{i,j}}$. 
\end{itemize}

If we have a quantum algorithm for computing a function, sometimes it uses intermediate measurement. We can move all measurements to the end using the ``principle of safe storage'' or the ``principle of deferred measurement''. Let us explain it.
We show it in the one-qubit case, but it can be easily generalized to the multiple-qubits case.
The intermediate measurement is useful if we apply $U$ depending on the results of a measurement. Assume that we have one control qubit $\ket{\psi}=a_0\ket{0}+a_1\ket{1}$ and one target qubit $\ket{\varphi}=b_0\ket{0}+b_1\ket{1}$.  Let us look at two cases. The first case is applying measurement to $\ket{\psi}$, then transformation $U=\begin{pmatrix}u_{00},&u_{01}\\u_{10},&u_{11}\end{pmatrix}$ to $\ket{\varphi}$ if the result of measurement is $\ket{1}$. The second case is applying control-$U$ to $\ket{\varphi}$ with control bit $\ket{\psi}$. Then we measure $\ket{\psi}$.

Let us look at the first case. After a measurement we obtain $\ket{0}$ with probability $|a_0|^2$ and $\ket{1}$ with probability $|a_1|^2$. After applying $U$, we obtain \[\ket{\varphi}\ket{\psi}=b_0\ket{00} +b_1\ket{10}\mbox{ with probability $|a_0|^2$}\]
 and 
\[\ket{\varphi}\ket{\psi}=(b_0u_{00} + b_1u_{01})\ket{01} +(b_0u_{10} + b_1u_{11})\ket{11}\mbox{ with probability $|a_1|^2$.}\] 
%So, if we measure the state after that, then we obtain $\ket{0}$ with probability $(a_0u_{00})^2+(a_1u_{01})^2$ and we obtain $\ket{1}$ with probability $(a_0u_{10})^2+(a_1u_{11})^2$.

Let us look at the second case. After applying control-$U$, we obtain 
\[\ket{\varphi}\ket{\psi}=a_0b_0\ket{00} +a_0b_1\ket{10} + a_1(b_0u_{00} + b_1u_{01})\ket{01} +a_1(b_0u_{10} + b_1u_{11})\ket{11}\] 
Then we measure $\psi$ and obtain
\[\ket{\varphi}\ket{\psi}=b_0\ket{00} +b_1\ket{10}\mbox{ with probability $|a_0|^2$}\]
 and 
\[\ket{\varphi}\ket{\psi}=(b_0u_{00} + b_1u_{01})\ket{01} +(b_0u_{10} + b_1u_{11})\ket{11}\mbox{ with probability $|a_1|^2$.}\] 
The situation is the same as in the first case.

The difference is the following. In the first case, we can reuse the qubit $\ket{\psi}$ after measurement. In the second case, we should keep the qubit $\ket{\psi}$ before measurement.
So, the moving measurement to the end increases the size of memory. At the same time, there is a new result that provides a method for moving measurement to the end with small additional memory \cite{fr2021}.
%If we have a quantum algorithm that uses inner measurement, then we can do all measurements at the end

\subsubsection{Complexity of the algorithm}

Let us discuss the complexity of the algorithm. The complexity of one query is the complexity of $U_f$ and $U_f^{-1}$. Let the complexity of computing $f$ is $T(f)$. So, the complexity of one query is $2T(f)$ and the total complexity is $O(T(f)\cdot \sqrt{n})$. 

\subsection{Amplitude Amplification or ``Diffusion generalization'' of Grover's Search Algorithm}\label{sec:amplampl}
As we have discussed in ``Query Complexity and Error Probability'' of Section \ref{sec:grover}, the diffusion is a rotation near $\theta$ that represented by $|\Psi\rangle=\begin{pmatrix} 
\frac{1}{\sqrt{n}}\\
\dots\\
\frac{1}{\sqrt{n}}
\end{pmatrix}$. We have two questions. The first one is ``Why should we rotate near exactly this angle $\theta$?''. Another question is ``Can we implement a quantum version of boosting for not Random sampling algorithm, but for some other algorithm?''. For answering these questions we consider any other quantum algorithm that does not have an intermediate measurement. So, this algorithm can be represented as a unitary matrix $A$. Assume that the algorithm finds the target element after one invocation with probability $p$. In other words, if we apply $A$ to $\ket{0}$ and measure it after that, then the target element can be obtained with probability $p$. Formally, if target elements correspond to elements from a set $K=\{k_1,\dots,k_t\}\subset [n]$, then
\[A\ket{0}=\sum_{i=1}^n a_i\ket{i}\mbox{ where }|a_{k_1}|^2+\dots+|a_{k_t}|^2=p\]
So we can say that $\sum\limits_{i\in K}|a_i|^2 + \sum\limits_{i\in[n]\backslash K}|a_i|^2 = 1$.
Let $\theta_A$ be such that 
\[\sin \theta_A = \sqrt{p}=\sqrt{\sum_{i\in K}|a_i|^2}\mbox{ and }\cos \theta_A = \sqrt{1-p}= \sqrt{\sum_{i\in[n]\backslash K}|a_i|^2}.\]  
Let $\ket{\Psi_A}=A\ket{0}$, and $D_A = 2\ket{\Psi_A}\bra{\Psi_A} - I^{\otimes \log n}$ like it was for Grover's search. Let us look at the angle $\alpha$ on the unit circle, where $\ket{1}$ corresponds to the target elements and $\ket{0}$ to non-target elements. So, we can see that we have a situation similar to the standard version of Grover's search algorithm. So, $\ket{\Psi_A}$ vector corresponds to the angle $\theta_A$, $D_A$ matrix rotates $\alpha$ near $\theta_A$, and $O$ converts $\alpha$ to $-\alpha$.

Therefore, the main step (the application of $Q$ and $D_A$) is equivalent to increasing $\alpha$ to $2\theta_A$. On the $i$-th step of the algorithm $\alpha=(2i+1)\theta_A$.
We stop on the $L$-th step of the algorithm when $\alpha=(2L+1)\theta_A\approx\frac{\pi}{2}$.
\[ \frac{\pi}{2}=(2L+1)\theta_A \approx (2L+1)sin(\theta_A)=(2L+1)\sqrt{p},\]
and $L\approx \frac{\pi}{4}\frac{1}{\sqrt{p}}$.

In $L$ steps $G_{(L)}\approx\sin(\pi/2)=1$ and $B_{(L)}\approx\cos(\pi/2)=0$. The state is
\[\ket{\psi}\approx\ket{k}\]
If we measure the $\ket{\psi}$ register, then we obtain $k$ with almost $1$ probability.

You can see, that $L$ is an integer. Maybe we cannot find an integer $L$ such that $(2L+1)\theta_A=\pi/2$. But we can reach the difference between $\alpha$ and $\pi$ at most $\theta_A$ similar to Grover's search case.Therefore, the error probability is $B_{(L)}^2\leq cos(\frac{\pi}{2}-\theta_A)^2=sin(\theta_A)^2 =(\sqrt{p})^2=p$.

Let us discuss $D_A$.
\[D_A=2\ket{\Psi_A}\bra{\Psi_A}-I^{\otimes\log n}=A 2\ket{0}\bra{0}A^* - AI^{\otimes\log n}A^{*}=\]\[A(2\ket{0}\bra{0}A^* - I^{\otimes\log n}A^{*})=A(2\ket{0}\bra{0} - I^{\otimes\log n})A^{*}=ARA^{*}=ARA^{-1}\]
where \[R=2\ket{0}\bra{0} - I^{\otimes\log n}=
\begin{pmatrix} 
2&0&0&\dots&0&0\\
0&0&0&\dots&0&0\\
0&0&0&\dots&0&0\\
\dots&\dots&\dots&\dots&\dots&\dots\\
0&0&0&\dots&0&0\\
0&0&0&\dots&0&0
\end{pmatrix}+\begin{pmatrix} 
-1&0&0&\dots&0&0\\
0&-1&0&\dots&0&0\\
0&0&-1&\dots&0&0\\
\dots&\dots&\dots&\dots&\dots&\dots\\
0&0&0&\dots&-1&0\\
0&0&0&\dots&0&-1
\end{pmatrix}
\]\[R=\begin{pmatrix} 
1&0&0&\dots&0&0\\
0&-1&0&\dots&0&0\\
0&0&-1&\dots&0&0\\
\dots&\dots&\dots&\dots&\dots&\dots\\
0&0&0&\dots&-1&0\\
0&0&0&\dots&0&-1
\end{pmatrix}\]

So the main step is the application of $Q$ and $D_A=ARA^{*}=ARA^{-1}$. As we discussed in Section \ref{sec:dettoquant} we can convert any deterministic, probabilistic (randomized), or quantum (even with intermediate measurement) algorithm to the unitary matrix maybe with increasing memory size.

The presented algorithm is called the Amplitude amplification algorithm \cite{bhmt2002}. Note, that we can combine the generalization of diffusion and the generalization of oracle. Finally, we obtain an algorithm that can amplify or boost the success probability of an algorithm $A$ with success probability $p$ and complexity $T(A)$. 

If we search an argument of a function $f$ with query complexity $O(T(f))$, then the complexity of oracle access is $O(T(f))$ as was discussed in the previous section. The complexity of diffusion is $T(A)+T(A^*)=O(T(A))$. So, the complexity of the main step is $O(T(A)+T(f))$ and the number of steps is $O(\frac{1}{\sqrt{p}})$. The total complexity is $O((T(A)+T(f))\cdot\frac{1}{\sqrt{p}})$.
\subsection{Randomized Oracle or Bounded-error Input for Grover's Search Algorithm}\label{sec:random-oracle}
Let us consider the Bounded-error Input case. Let us present the problem formally.

Given a function $f:\{0,\dots,n-1\}\to \{0,1\}$ one want to find the $x\in[n]=\{0,\dots,n-1\}$ such that $f(x)=1$. Access to the function $f$ is provided via a function $g:\{0,\dots,n-1\}\to \{0,1\}$. If $f(x)=0$ then $g(x)=0$. If $f(x)=1$, then $g(x)=1$ with probability $\varepsilon$ and $g(x)=0$ with probability $1-\varepsilon$ for $0<\varepsilon<1$.

The one way for solving the problem is repeating a query to $g(x)$. If we query it $3\log_2 n$ times and return $0$ iff all of queries return $0$, then the probability of error is at most $(1-\varepsilon)^{3\log_2 n}=O(\frac{1}{n^3})$. If we consider the whole input, that contains $n$ values, then we can see that error events are independent, and the probability of error in at least one element is $O(\frac{1}{n^2})$. The probability of error in at least one element in at least one of $O(\sqrt{n})$ steps is $O(\frac{1}{n^{1.5}})$. The probability error of Grover's search algorithm itself is constant at most $1/3$. Finally, the total error probability is $1/3+O(\frac{1}{n^{1.5}})\leq 0.4$. 

So, we obtain the algorithm with good error probability. The total complexity of the algorithm is $O(T(g)\sqrt{n}\log n)$, where $T(g)$ is the complexity of computing $g$. In other words, we increase the complexity in $\log n$ times.

At the same time, we can obtain similar results without an additional $\log n$ factor.

Assume that we consider an algorithm $A$ that randomly picks an element $i\in\{0,\dots,n-1\}$ and returns the successful answer if $g(i)=1$. What is the probability of success for the such algorithm? It is $p=\frac{\varepsilon}{n}$. Therefore, we can amplify it and obtain a target element with constant error probability and complexity $O(T(g)\cdot \frac{\sqrt{n}}{\varepsilon})$ that is $O(T(g)\cdot \sqrt{n})$ if $\varepsilon=const$.

Similarly, we can act in the case of two-side errors. In other words, in the case of probability $f(x)\neq g(x)$ is less or equal to $\varepsilon$.

The results of this section are presented in \cite{hmw2003, abikkpssv2020}. 
\section{Basic Applications of Grover's Search Algorithm. First One Search, Minimum Search, LCP Problems}
Let us discuss several problems where we can apply Grover's Search algorithm.
\subsection{String Equality Problem}
{\bf Strings Equality Problem}
 \begin{itemize}
 \item We have two strings $s=s_1,\dots,s_n$ and $t=t_1,\dots,t_m$, where $t_i,s_i\in\{0,1\}$.
 \item Check whether these two strings are equal
 \end{itemize}
The algorithm is based on \cite{ki2019} results. At first, we check whether $n=m$. If it is true, then we continue.
Let us consider a search function $f:[n]\to\{0,1\}$ where $[n]=\{0,\dots,n-1\}$. Let $f(i)=1$ iff $t_i\neq s_i$. Our goal is to search any $i_0$ such that $f(i_0)=1$. If we find such $i_0$, then $s\neq t$. If we cannot find such $i_0$, then $s=t$. We do not know how many unequal symbols these two strings have. Therefore, we can apply Grover's Search algorithm with an unknown number of solutions. The complexity of computing $f$ is constant. Therefore, the complexity of the total algorithm is $O(\sqrt{n})$, and the error probability is constant. 
\subsection{Palindrome Checking Problem}
{\bf Polyndrom Checking Problem}
 \begin{itemize}
 \item We have a string $s=s_1,\dots,s_n$, where $s_i\in\{0,1\}$.
 \item We want to check whether $s_1,\dots,s_n=s_n,\dots,s_1$, i.e. $s=s^R$.
 \end{itemize}
 
Let us consider a search function $f:[n/2]\to\{0,1\}$ where $[n/2]=\{0,\dots,n/2-1\}$. Let $f(i)=1$ iff $s_i\neq s_{n-i-1}$. Our goal is to search any $i_0$ such that $f(i_0)=1$. If we find such $i_0$, then $s\neq s^R$. If we cannot find such $i_0$, then $s=s^R$. We do not know how many unequal symbols we have. Therefore, we can apply Grover's Search algorithm with an unknown number of solutions. The complexity of computing $f$ is constant. Therefore, the complexity of the total algorithm is $O(\sqrt{n})$, and the error probability is constant.
\subsection{Minimum and Maximum Search problem}\label{sec:max}
{\bf Minimum Search Problem}
 \begin{itemize}
 \item A vector of integers $a=a_1,\dots,a_n$ is given.
 \item We search $MIN=\min\{a_1,\dots,a_n\}$.
 \end{itemize}
 Here we present the algorithm from \cite{dh96,dhhm2004}.
 The algorithm is based on iterative invocations of Grover's Search algorithm. Each iteration improves the answer from the previous iteration.
 \begin{itemize}
 \item Iteration 0. We uniformly randomly choose an index $j_0\in_R[n]$.
 \item Iteration 1. Let use define the function $f_1:[n]\to\{0,1\}$ such that $f_1(i)=1$ iff $a_i<a_{j_0}$. In other words, it marks elements that improve the answer from Iteration 0. We invoke Grover's Search for $f_1$. The result is $j_1$. Due to behavior of Grover's Search algorithm, $j_1\in_R\{i:f_1(i)=1, i\in[n]\}$. It is chosen from the set with equal probability.
\item Iteration 2. Let use define the function $f_2:[n]\to\{0,1\}$ such that $f_2(i)=1$ iff $a_i<a_{j_1}$. We invoke Grover's Search for $f_2$. The result is $j_2$. We can say that $j_2\in_R\{i:f_2(i)=1, i\in[n]\}$.
\item Iteration 3. Let use define the function $f_3:[n]\to\{0,1\}$ such that $f_3(i)=1$ iff $a_i<a_{j_2}$. We invoke Grover's Search for $f_3$. The result is $j_3$. We can say that $j_3\in_R\{i:f_3(i)=1, i\in[n]\}$.
\item ...
\item If on an iteration $m$ Grover's Search algorithm does not find an element, then it means we cannot improve the existing solution $j_{m-1}$. So, the solution is $a_{j_{m-1}}$. 
 \end{itemize}
 
We can see that the algorithm can be used for the Maximum search problem. 
Assume that we want to search $MAX=\max\{a_1,\dots,a_n\}$.  It is easy to see that $MAX= - \min\{-a_1,\dots,-a_n\}$. In other words, if we find the minimum of $-a_i$ (we invert the sign of the elements), then after inverting the sign of the result back we obtain the result.

In fact, the algorithm returns not only $MIN$ or $MAX$ itself, but the index of the target element also. If there are several elements with the minimal (maximal) value, then it returns one of them with equal probability.
 
 \subsubsection{Complexity and Error Probability of Minimum Search Algorithm }
 Let us consider the complexity of the algorithm.
 Let $m$ be the number of iterations until reaching the minimum.
 Let $t_i$ be the running time of $i$-th iteration. Note that $t_i$ is a random variable.
 Let $T=t_0+\dots+t_m$ be the total running time. $T$ is a random variable.
 We want to compute the expected value of the total running time $\mathbb{E}T$.
Let us present $T$ in another form
 \[T\leq q_1\cdot c\sqrt{\frac{n}{1}}+ q_2\cdot c\sqrt{\frac{n}{2}}+\dots+ q_n\cdot c\sqrt{\frac{n}{n}}.\] Here $q_j$ is an indicator of corresponding term with a value $0$ or $1$, $c=const$. 
 Let us explain this sum. Due to $t_i$ is complexity of Grover's search of some iteration, $t_i\in\{c\sqrt{\frac{n}{v}}, v\in\{1,\dots,n\}\}$. Additionally, we can say that each search function $f_i$ marks strictly smaller than $a_{j_{i-1}}$ elements, it means number of solutions are such that $|\{j:f_i(j)=1\}|<|\{j:f_{i-1}(j)=1\}|$. Therefore, the sum $T$ contains only terms from $\{c\sqrt{\frac{n}{v}}, v\in\{1,\dots,n\}\}$, end each term can be met at most once.
 
Due to the linearity of the expected value, we can say that 
\[\mathbb{E}T=\sqrt{\frac{n}{1}}\mathbb{E}q_1 +\dots + \sqrt{\frac{n}{n}}\mathbb{E}q_n\]
\[\mathbb{E}q_j=1\cdot p_j + 0\cdot(1-p_j)=p_j,\]
where $p_j$ is a probability of occurrence of $\sqrt{\frac{n}{j}}$ term. In other words, probability of choosing $a_i$ such that $|\{z:a_z<a_i\}|=j$. It is because $\sqrt{\frac{n}{j}}$ means that after choosing such an element, we invoke Grover search and the number of marked elements is $j$.

Hence, we have
\[\mathbb{E}T=\sqrt{\frac{n}{1}}p_1 +\dots + \sqrt{\frac{n}{n}}p_n\]
The probability $p_j\leq \frac{1}{j}$. We shall show it later. Therefore,
\[\mathbb{E}T\leq\sum\limits_{j=1}^{n}\frac{1}{j}\cdot c\sqrt{\frac{n}{j}}=c\sqrt{n}+c\sqrt{n}\sum\limits_{j=2}^{n}j^{-1.5}\leq \]
\[c\sqrt{n}+c\sqrt{n}\int\limits_{2}^nj^{-1.5}dj=c\sqrt{n}+c\sqrt{n}(-\frac{1}{2\sqrt{j}})\Big|_2^n=\]
\[=c\sqrt{n}+c\sqrt{n}(\frac{1}{2\sqrt{2}}-\frac{1}{2\sqrt{n}})\leq c\sqrt{n}+c\sqrt{n}\frac{1}{2\sqrt{2}}=O(\sqrt{n})\]
Due to Markov inequality, we have $T=O(\sqrt{n})$ with constant probability, for example, $0.9$.

Let us show that $p_j\leq \frac{1}{j}$.
Assume that we already have proved the claim for all $z\geq j$ and $p_z\leq\frac{1}{z}$.
Let us consider the event $B_j$ that is an occurrence of $c\sqrt{\frac{n}{j}}$. In other words, probability of choosing $a_i$ such that $rank(a_i)=j$, where $rank(a_i)=|\{z:a_z<a_i\}|$.
This event can happen if exactly before $i$-th element we have chosen some element $i'$ with $rank(a_{i'})>j$ because in that case $a_{i'}>a_i$.
The event that $i'$ was chosen and after that $i$ was chosen is $D_{i'}$.  Let us compute the probability  $p(D_{i'})$ of the event $D_{i'}$. The probability is the product of probabilities for two independent events:
\begin{itemize}
\item $a_{i'}$ was chosen. Let $rank(a_{i'})=w$. Then, the probability of this event is the same as the probability of $c\sqrt{\frac{n}{w}}$ occurrence, that is $p(B_w)\leq\frac{1}{w}$.
\item $a_{i}$ was chosen exactly after $a_{i'}$. After choosing $a_{i'}$, the Grover's search returns any element from $\{z:a_z<a_{i'}\}$ with equal probability. Probability of choosing exactly $i\in\{z:a_z<a_{i'}\}$ is $\frac{1}{rank(a_{i'})}=\frac{1}{w}$.
\end{itemize} 

So, the probability of $D_{i'}$ is

\[p(D_{i'})=p(B_{w})\cdot\frac{1}{w}\leq \frac{1}{w^2}\mbox{, where }rank(a_{i'})=w.\]
Let us compute the probability  $p(B_j)$ of the event $B_j$. Any $i'$ such that $rank(a_{i'})>j$ can be before $i$. So, the probability $p(B_j)$ is sum of corresponding $p(D_{i'})$
\[p(B_j)=\sum\limits_{rank(a_{i'})>j}p(D_{i'})\leq\sum\limits_{w=j+1}^{n}\frac{1}{w^2}\leq \int\limits_{j+1}^{n}w^{-2}dz=-\frac{1}{w}\Big|_{j+1}^n=\frac{1}{j+1}-\frac{1}{n}\leq \frac{1}{j}.\]

 \subsubsection{Boosting the Success Probability}
We have proved, that the success probability is at least $0.9$ and the error probability is at most $0.1$. Let us obtain the success probability $1-O(\frac{1}{n})$ and error probability $O(\frac{1}{n})$.

The error of the algorithm means returning $a_i$ that is not the minimum. Let us repeat the algorithm $k$ times and choose the minimum of the answers. In that case, we have an error iff all invocations of the algorithm have an error. Therefore error probability of $k$ invocations is $0.1^k$. Let $k=(-\log_20.1)\cdot\alpha\cdot \log_2n$ for some constant $\alpha$. Then, the error probability is $0.1^k=0.1^{(-\log_20.1)\cdot\alpha\log_2n}=2^{-\alpha\log_2n}=n^{-\alpha}=O(\frac{1}{n^{\alpha}})$.
We can choose $\alpha=1$ or any other positive constant or function that we need. 
 
\subsection{First One Search Problem}\label{sec:first-one}
{\bf The First One Search Problem}
 \begin{itemize}
 \item We have a function $f(i)$ where $f:\{0,\dots,n-1\}\to\{0,1\}$.
 \item Find the minimal $x$ such that $f(x)=1$. 
 \end{itemize}
 
 We have two algorithms for the problem. The first one is based on the minimum search algorithm and was presented in \cite{dhhm2004,k2014,ll2015,ll2016}.
Let us consider a function $g:[n]\to[2]\times[n]$ that is $g(i)=(1-f(i),i)$. Then, the target element $g(x)=(1-f(x),x)$ is the lexicographical minimum of $g$.
We can apply the minimum search algorithm for the problem. Complexity is $O(\sqrt{n})$.

The second one is based on Grover's search algorithm and the Binary search algorithm. It was presented in \cite{kkmsy2020}.

Assume we have $i\in\{l,\dots,r\}$ such that that $f(i)=1$ for some $l$ and $r$ where $0\leq l\leq r\leq n-1$. At the same time, we assume that $f(i)=0$ for any $0\leq i<l$.
Let $mid = (l+r)/2$. If $Grover(l,mid,f)=1$, then we can be sure that the first one in $\{l,\dots,mid\}$ and we update $r\gets mid$, else $l\gets mid+1$ because the target element in $\{mid+1,\dots,r\}$.
So we start from $l=0$ and $r=n-1$. After each step that was presented before, we always have $x\in\{l,\dots,r\}$. We stop when $l=r$.

The Complexity of such solution is \[O\left(\sqrt{n}+\sqrt{\frac{n}{2}}+\sqrt{\frac{n}{2^2}}+\dots+\sqrt{\frac{n}{2^{\lceil\log_2 n\rceil}}}\right)\leq O\left(\sqrt{n}\sum\limits_{i=0}^{\lceil\log_2 n\rceil}\frac{1}{2^{0.5i}}	\right)=O\left(\sqrt{n}\right)\]

 Success probability is $(\frac{1}{2})^{\log_2 n}=\frac{1}{n}$. It is a very small success probability! Let us boost it. For this reason, we repeat the invocation of Grover's Search algorithm on $i$-th step of Binary search $2i$ times. The complexity of such solution is \[\bigo{\sqrt{n}+2\cdot 1\cdot\sqrt{\frac{n}{2}}+2\cdot2\cdot\sqrt{\frac{n}{2^2}}+\dots+2\lceil\log_2 n\rceil\sqrt{\frac{n}{2^{\lceil\log_2 n\rceil}}}}\leq \bigo{\sqrt{n}\sum\limits_{i=0}^{\lceil\log_2 n\rceil}\frac{i}{2^{0.5i}}}=\bigo{\sqrt{n}}\]
Error probability on the $i$-the step of Binary search is $0.25^i$, where $p=0.25$ is an error of Grover's Search.
The total error probability is \[\sum_{i=1}^n0.25^i\leq\sum_{i=1}^\infty 0.25^i=\frac{0.25}{1-0.25}=\frac{1}{3}.\]
\subsection{Searching with Running Time that Depends on the Answer}\label{sec:first-one2}
Let us present the algorithm with $O(\sqrt{x})$ running time where $x=min\{i\in[n]: f(i)=1\}$ the target element. The algorithm is presented in \cite{k2014,ll2015,ll2016,kkmsy2020}.
Let us consider a search border $b$. The algorithm is following.
\begin{itemize}
\item Iteration 1. We set $b\gets 2^1$. We invoke the first one search algorithm for $[0,b-1]$.  If the solution is found, then stop, otherwise, continue.
\item Iteration 2.  We set $b\gets 2^2$. Invoke the first one search algorithm for $[0,b-1]$. If the solution is found, then stop, otherwise, continue.
\item Iteration 2.  We set $b\gets 2^3$. Invoke the first one search algorithm for $[0,b-1]$. If the solution is found, then stop, otherwise, continue.
\item ...
\item Iteration $z$.  We set $b\gets 2^z$. Invoke the first one search algorithm for $[0,b-1]$. If the solution is found then stop, otherwise continue.
\end{itemize}

The process is stopped when $2^{z-1}<x\leq 2^z$. So, we can say that $z=\lceil \log_2 x\rceil$.
The complexity is $O(\sqrt{2^1}+\dots+\sqrt{2^z})=O(\sqrt{2^z})=O(\sqrt{x})$.

The error probability of this method is big, at the same time, it can be reduced without increasing complexity using the technique described in  \cite{k2014,ll2015,ll2016} or the method similar to the Binary search's way to decrease the error probability (repeating invocation several times).

\subsection{All Ones Search Problem}\label{sec:all-ones}
{\bf All Ones Search Problem}
 \begin{itemize}
 \item We have a function $f(i)$ where $f:[n]\to\{0,1\}$.
 \item Our goal is to find all  $x_1,\dots,x_t$ such that $f(x_i)=1$ for $i\in[n]$. 
 \end{itemize}
 
 For the problem, we can use the first one search algorithm. 
 The algorithm is following
\begin{itemize}
\item Iteration 1. We search the first solution $x_1$ in the segment $[0,n-1]$.
\item Iteration 2. We search the first solution $x_2$ in the segment $[x_1+1,n-1]$.

\item ...
\item Iteration $t$.  We search the first solution $x_t$ in the segment $[x_{t-1}+1,n-1]$.
\item Iteration $t+1$. We search for the first solution in the segment $[x_{t}+1,n-1]$. On this iteration, we do not find any solution and stop the process.
\end{itemize}
Complexity of the iteration 1 is $O(\sqrt{x_1})$, complexity of an iteration $i$ is $O(\sqrt{x_i-x_{i-1}})$ for $i\in\{2,\dots,t\}$. Complexity of the iteration $t+1$ is $O(\sqrt{n-x_t})$.
The total complexity is \[\bigo{\sqrt{x_1}+\sqrt{x_2-x_{1}}+\dots+\sqrt{x_t-x_{t-1}}+\sqrt{n-x_t}}.\] Due to Cauchy–Bunyakovsky-Schwarz inequality the complexity is less or equal to \[\bigo{\sqrt{t\cdot(x_1+x_2-x_{1}+\dots+x_t-x_{t-1}+n-x_t)}}=\bigo{\sqrt{tn}}.\]

The algorithm and the way of reducing error probability are presented in \cite{k2014}.

\subsection{The Longest common prefix (LCP) Problem}
{\bf The Longest common prefix (LCP) Problem}
 \begin{itemize}
 \item We have two strings $s=s_0,\dots,s_{n-1}$ and $t=t_0,\dots,t_{m-1}$, where $t_i,s_i\in\{0,1\}$.
 \item The goal is to find the maximum $j$ such that $s_1=t_1,\dots,t_j=s_j$ and $t_{j+1}\neq s_{j+1}$. If there is no such $j$, then $j=\min\{n,m\}$. It happens when one string is a prefix of the second one.
 \end{itemize}
 The algorithm is presented in \cite{kkmsy2020}.  
 Let $k=\min\{n,m\}$.
Let us consider a search function $f:[k]\to\{0,1\}$ where $[k]=\{0,\dots,k-1\}$. Let $f(i)=1$ iff $t_i\neq s_i$. Our goal is searching {\bf the minimal} $i_0$ such that $f(i_0)=1$. If we find such $i_0$, then the length of LCP of $s$ and $t$ is $j=i_0-1$. If we cannot find such $i_0$, then one string is a prefix of the second one and $j=k$. We do not know how many unequal symbols these two strings have. Therefore, we can apply Grover's Search algorithm with an unknown number of solutions. The complexity of computing $f$ is constant. Therefore, the complexity of total algorithm is $O(\sqrt{j}\cdot O(1))=O(\sqrt{j})=O(\sqrt{\min\{n,m\}})$ and error probability is constant. 

\subsection{Comparing Two Strings in Lexicographical Order}\label{sec:str-cmp}
{\bf Comparing Two Strings in Lexicographical Order}
 \begin{itemize}
 \item We have two strings $s=s_0,\dots,s_{n-1}$ and $t=t_0,\dots,t_{m-1}$, where $t_i,s_i\in\{0,1\}$.
 \item Our goal is one of three following options:
 \begin{itemize}
 \item return $-1$ if $s<t$ in lexicographical order;
 \item return $+1$ if $s>t$ in lexicographical order;
 \item return $0$ if $s=t$;
 \end{itemize}
 \end{itemize}
 The algorithm is presented in \cite{ki2019,kkmsy2020,k2021,kl2020}.  
 Let $j$ be LCP of $s$ and $t$. The algorithm for this problem was presented in the previous section.  If $j=\min\{n,m\}$, then we can return the answer depending on the lengths of the strings
 \begin{itemize}
 \item if $n=m$, then return $0$;
 \item if $n\neq m$ and $n=\min\{n,m\}$, then return $-1$;
 \item if $n\neq m$ and $m=\min\{n,m\}$, then return $+1$.
 \end{itemize}
 If $j=\min\{n,m\}$ and $s_{j+1}<t_{j+1}$, then we return $-1$, and $+1$ otherwise.
 The complexity of the total algorithm is the same as for LCP that is $O(\sqrt{j}\cdot O(1))=O(\sqrt{j})=O(\sqrt{\min\{n,m\}})$ and error probability is constant. 
\section{Applications of Grover's Search Algorithm for Graph Problems. DFS, BFS, Dynamic Programming on DAGs, TSP}
Let us discuss several important algorithms on Graphs. Classical versions of most of these algorithms are presented in \cite{cormen2001,l2017guide}.
\subsection{Depth-first search}\label{sec:dfs}
Let us discuss the well-known DFS algorithm that can be found in Section 22 of \cite{cormen2001} or Section 7.2 of \cite{l2017guide}.
 \begin{itemize}
 \item We have a graph $G=(V,E)$, where $V$ is the set of vertexes, $E$ is the set of edges, $n=|V|$, and $m=|E|$. 
 \item We should invoke the DFS algorithm.
 \end{itemize}

Assume that we have a $visited$ array such that  $visited[v]=true$ if our search has visited a vertex $v$ and $visited[v]=false$ otherwise. Let $\textsc{Neighbors}(v)$ be a list of all neighbors of a vertex $v$. In other words, $\textsc{Neighbors}(v)=(x_1,\dots,x_{L_v})$, where $(v,x_i)\in E$ and $L_v$ is the length of the neighbors list for the vertex $v$.

The standard form of the recursive procedure for the DFS algorithm is presented in Algorithm \ref{alg:dfs1}.

\begin{algorithm}
\caption{Standard version of DFS algorithm. It is a $dfs$ procedure that accepts an observing vertex $v$ as a parameter.}\label{alg:dfs1}
\begin{algorithmic}
\State $visited[v] \gets True$
\For{$x \in \textsc{Neighbors}(v)$}
\If{$visited[x]=false$}
\State $dfs(x)$
\EndIf
\EndFor
\end{algorithmic}
\end{algorithm}
 
We can modify the code, by replacing checking {\em all} neighbors of $v$ by checking {\em only not visited} neighbors of $v$. See Algorithm \ref{alg:dfs2}.

\begin{algorithm}
\caption{Modified version of DFS algorithm. It is a $dfs$ procedure that accepts an observing vertex $v$ as a parameter.}\label{alg:dfs2}
\begin{algorithmic}
\State $visited[v] \gets True$
\For{$x \in \textsc{NOT\!-\!VISITED\_Neighbors}(v)$}
\State $dfs(x)$
\EndFor
\end{algorithmic}
\end{algorithm}
Here  \textsc{NOT$\_$VISITED$\_$Neighbors}$(v)$ is the sequence of elements $x$ form \textsc{Neighbors}$(v)$ such that $visited[x]=false$. Note that this list can be changed during invocations of the function for other vertexes. That is why we can use another function that returns an index of a not-visited neighbor of $v$ with an index bigger than $i$ that is  \textsc{NEXT$\_$NOT$\_$VISITED$\_$Neighbor}$(v,i)$. It returns $i'>i$ such that  $visited[x_{i'}]=false$ for  \textsc{Neighbors}$(v)=(x_1,\dots,x_{L_v})$. Assume that if there is no such $i'$, then the procedure returns $NULL$ constant. To search the first not-visited neighbor of $v$, we should invoke \textsc{NEXT$\_$NOT$\_$VISITED$\_$Neighbor}$(v,-1)$.  See Algorithm \ref{alg:dfs3}.  

\begin{algorithm}
\caption{The third version of DFS algorithm. It is a $dfs$ procedure that accepts an observing vertex $v$ as a parameter.}\label{alg:dfs3}
\begin{algorithmic}
\State $visited[v] \gets True$
\State $i \gets\textsc{NEXT\_NOT\_VISITED\_Neighbor}(v,-1)$
\While{$i \neq NULL$}
\State $dfs(x_i)$
\State $i \gets\textsc{NEXT\_NOT\_VISITED\_Neighbor}(v,i)$
\EndWhile
\end{algorithmic}
\end{algorithm}

Let us implement $\textsc{NEXT\_NOT\_VISITED\_Neighbor}(v,i)$ as a quantum algorithm for the first one search (Section \ref{sec:first-one2}) for a search function $f:\{i+1,\dots,L_v\}\to \{0,1\}$, where $f(i')=1$ iff $visited[x_{i'}]=false$.

The quantum algorithm is presented in \cite{f2008}.

\subsubsection{Complexity}
Let us discuss the complexity of the algorithm.

Complexity of processing a vertex $v$ is $O(\sqrt{L_v\cdot N_v}\log L_v)=O(\sqrt{L_v\cdot N_v}\log n)$, where $L_v$ is a number of neighbors for $v$ and $N_v$ is a number of processed by $dfs(v)$ not visited neighbors for $v$.
The total complexity is
\[O\left(\sum_{v\in V}\sqrt{L_v\cdot N_v}\log n\right)\leq\]
due to Cauchy–Bunyakovsky-Schwarz inequality
\[ O\left(\sqrt{\sum_{v\in V}L_v}\cdot\sqrt{\sum_{v\in V} N_v}\log n\right)\leq\]
Note that $\sum_{v\in V}L_v=O(|E|)=O(m)$ because each directed edge gives us a neighbor in one of the lists and an undirected edge gives us neighbors in two lists.  $\sum_{v\in V} N_v\leq |V|=n$ because each vertex is visited only once. This property is provided by the $visited$ array.
\[\leq \bigo{\sqrt{m}\cdot\sqrt{n}\log n}=\bigo{\sqrt{nm}\log n}.\]

In the case of the Adjacency matrix, the complexity is $O(n^{1.5}\log n)$.

Note that in the classical case, the complexity is $O(m)$ in the case of the list of neighbors and $O(n^2)$ in the case of the Adjacency matrix. So, we obtain a quantum speed-up if $m>n(\log n)^2$ in the case of the list of neighbors. If we use the Adjacency matrix, then we have a quantum speed-up anyway. 

In the case of the Adjacency matrix, the algorithm with $O(n^{1.5})$ complexity is presented in \cite{bt2020}.

\subsection{Topological sort}\label{sec:topsort}
Let us discuss the well-known Topological sort algorithm for sorting vertexes of an acyclic-directed graph. The problem and a classical algorithm can be found in Section 22 of \cite{cormen2001} or Section 7.2 of \cite{l2017guide}.
\begin{itemize}
 \item We have a directed graph $G$ with $n$ vertexes and $m$ edges. 
 \item We should sort vertexes and obtain an order $order=(v_1,\dots,v_n)$ such that for any edge $(v_i,v_j)$ we have $i<j$.
 \end{itemize}
 
The standard implementation of the topological sort algorithm is based on the DFS algorithm and is presented in Algorithm \ref{alg:topSort}.

\begin{algorithm}
\caption{Topological sort algorithm. It is a $TopSort$ procedure that accepts an observing vertex $v$ as a parameter.}\label{alg:topSort}
\begin{algorithmic}
\State $visited[v] \gets True$
\State $i \gets\textsc{NEXT\_NOT\_VISITED\_Neighbor}(v,-1)$
\While{$i \neq NULL$}
\State $dfs(x_i)$
\State $i \gets\textsc{NEXT\_NOT\_VISITED\_Neighbor}(v,i)$
\EndWhile
\State $order\gets v \cup order $ 
\end{algorithmic}
\end{algorithm}

We have a small modification of the DFS algorithm. So, the complexity of the algorithm equals to the complexity of the DFS algorithm. That is $O(\sqrt{nm}\log n)$ in the case of the list of neighbors and $O(n^{1.5}\log n)$ in the case of the adjacency matrix. In the case of the Adjacency matrix, the algorithm with $O(n^{1.5})$ complexity is presented in \cite{bt2020}.

\subsection{Breadth-first search}
Let us discuss the well-known BFS algorithm that can be found in Section 22 of \cite{cormen2001} or Section 7.2 of \cite{l2017guide}.
\begin{itemize}
 \item We have a graph $G=(V,E)$, where $V$ is the set of vertexes, $E$ is the set of edges, $n=|V|$, and $m=|E|$. 
 \item We should invoke the BFS algorithm.
 \end{itemize}
We use the same notations as for the DFS algorithm. Additionally, we use the queue data structure (See  Section 10 of \cite{cormen2001}). It is a collection (sequence) of elements such that we can add an element at the end of the queue and get an element from the beginning of the queue. It is a so-called FIFO (first-in-first-out) data structure. Assume that we have 
\begin{itemize}
\item $\texttt{ADD}(queue,v)$ procedure for adding a vertex $v$ to the queue;
\item $\texttt{GET\_TOP}(queue)$ procedure returns a vertex $v$ from the top of the queue;
\item $\texttt{EMPTY}(queue)$ procedure returns $true$ if the queue is empty and $false$ otherwise.
\end{itemize}

The standard implementation of BFS is presented in Algorithm \ref{alg:bfs1}. 

\begin{algorithm}
\caption{Standard version of BFS algorithm. It is a $bfs$ procedure that accepts the starting vertex $v$ as a parameter.}\label{alg:bfs1}
\begin{algorithmic}
\State $visited[v] \gets True$
\State $ADD(queue,v)$
\While{$EMPTY(queue)=False$}
\State $v\gets GET\_TOP(queue)$
\For{$x \in \textsc{Neighbors}(v)$}
\If{$visited[x]=False$}
\State $visited[x] \gets True$
\State $ADD(queue,x)$
\EndIf
\EndFor
\EndWhile
\end{algorithmic}
\end{algorithm}

As for DFS we use  $\texttt{NOT\_ VISITED\_ Neighbors}(v)$ in Algorithm \ref{alg:bfs2}.

\begin{algorithm}
\caption{The modified version of BFS algorithm. It is a $bfs$ procedure that accepts the starting vertex $v$ as a parameter.}\label{alg:bfs2}
\begin{algorithmic}
\State $visited[v] \gets True$
\State $ADD(queue,v)$
\While{$EMPTY(queue)=False$}
\State $v\gets GET\_TOP(queue)$
\For{$x \in \textsc{NOT\_VISITED\_Neighbors}(v)$}
\If{$visited[x]=False$}
\State $visited[x] \gets True$
\State $ADD(queue,x)$
\EndIf
\EndFor
\EndWhile
\end{algorithmic}
\end{algorithm}

So, we can implement $\texttt{NOT\_ VISITED\_ Neighbors}(v)$ function using the quantum procedure for searching all ones (Section \ref{sec:all-ones}) for a function $f_v:\{1,\dots,L_v\}\to \{0,1\}$. See Algorithm \ref{alg:bfs3}.

\begin{algorithm}
\caption{The modified version of BFS algorithm. It is a $bfs$ procedure that accepts the starting vertex $v$ as a parameter.}\label{alg:bfs3}
\begin{algorithmic}
\State $visited[v] \gets True$
\State $ADD(queue,v)$
\While{$EMPTY(queue)=False$}
\State $v\gets GET\_TOP(queue)$
\For{$x \in \textsc{ALL\_ONES}(f_v)$}
\If{$visited[x]=False$}
\State $visited[x] \gets True$
\State $ADD(queue,x)$
\EndIf
\EndFor
\EndWhile
\end{algorithmic}
\end{algorithm}
\subsubsection{Complexity}
Let us discuss the complexity of the algorithm.

Complexity of processing a vertex $v$ is $O(\sqrt{L_v\cdot N_v}\log L_v)=O(\sqrt{L_v\cdot N_v}\log n)$, where $L_v$ is a number of neighbors for $v$ and $N_v$ is a number of processed by $bfs(v)$ not visited neighbors for $v$.
The total complexity is
\[O\left(\sum_{v\in V}\sqrt{L_v\cdot N_v}\log n\right)\leq\]
due to Cauchy–Bunyakovsky-Schwarz inequality
\[ O\left(\sqrt{\sum_{v\in V}L_v}\cdot\sqrt{\sum_{v\in V} N_v}\log n\right)=\]
Note that $\sum_{v\in V}L_v=O(|E|)=O(m)$ because each directed edge gives us a neighbor in one of the lists and the directed edge give us neighbors in two lists.  $\sum_{v\in V} N_v=|V|=n$ because each vertex is visited only once. This property is provided by the $visited$ array.
\[=O(\sqrt{m}\cdot\sqrt{n}\log n)=O(\sqrt{nm}\log n)\]

In the case of the Adjacency matrix, the complexity is $O(n^{1.5}\log n)$.

Note that in the classical case, the complexity is $O(m)$ in the case of the list of neighbors and $O(n^2)$ in the case of the Adjacency matrix. So, we obtain a quantum speed-up if $m>n(\log n)^2$ in the case of the list of neighbors. If we use the Adjacency matrix, then we have a quantum speed-up anyway.

 In the case of the Adjacency matrix, the algorithm with $O(n^{1.5})$ complexity is presented in \cite{ll2015,ll2016}.
\subsection{Dynamic Programming on Directed Acyclic Graphs. Games on DAGs}
Let us discuss a problem from the Game theory called games on Graphs or games on DAGs. An example of such a game can be a Subtraction game or NIM game \cite{emaxxGame}.

 \begin{itemize}
 \item We have a Directed Acyclic Graph(DAG) $G(V,E)$ with $n$ vertexes and $m$ edges. 
 \item There is a stone on the starting vertex $A$.
 \item Two players can move it by directed edges. Players moves turn by turn, one after another.
 \item The player who cannot move the stone loses.
 \item We want to compute the number of player (the first or the second) who wins in the game. Additionally, we want to have a winning strategy for this player.
 \end{itemize}
 Many games can be represented in this way. A vertex can be a game situation or game configuration. A stone is the current game configuration. An edge is a possible move to another game configuration. In chess, positions of chessmen. At the same time, many complex games like chess and others have too many configurations or too big graph for such analysis.
 
For solving the problem we mark each vertex $v$ as a winning vertex or as a loss vertex. We use a $Win:V\to\{true,false\}$ function. $Win(v)=true$ iff we can start from this vertex $v$ and win. $Win(v)=false$ if there is no way to start from this vertex $v$ and win.

Assume, that the graph is topologically sorted. Otherwise, topological sorting is the first step. We can say that each vertex $v$ without outgoing edges is such that $Win(v)=false$.  Then, for each vertex from $n$ to $1$ (in topological order), we can say that $Win(v)=true$ iff there is a vertex $x$ such that there is an edge $(v,x)\in E$ and $Win(x)=false$. In that case, we can move the stone to $x$ and the opposite player will be in the losing game configuration.  Because of the topological sort, we can be sure that $x$ and $v$ satisfy $x>v$ and it is already computed in previous steps. So we have Algorithm \ref{alg:game1}.

\begin{algorithm}
\caption{The algorithm for solving a game on a graph.}\label{alg:game1}
\begin{algorithmic}
\State $G\gets ClassicalTopSort(G)$
\For $v\in\{n,\dots,1\}$
\If{any neighbor element $x$ is such that $Win(x)=false$}
\State $Win(v)\gets True$
\Else
\State $Win(v)\gets False$
\EndIf
\EndFor
\end{algorithmic}
\end{algorithm}

Finally, if $Win(1)=true$, then the first player wins, and the second player wins otherwise.

For checking ``all neighbor elements $x$ are such that $Win(x)=true$'' condition, we can use Grover's Search algorithm. We search an $x$ for $f_v:\{1,\dots,L_v\}\to\{0,1\}$ function such that $f(i)=1$ iff $Win(x_i)=false$ for $Neighbors(v)=(x_1,\dots, x_{L_v})$. Additionally, we use the quantum implementation of the Topological sort (Section \ref{sec:topsort}). Finally, the algorithm is following.

\begin{algorithm}
\caption{The quantum algorithm for solving a game on a graph.}\label{alg:game2}
\begin{algorithmic}
\State $G\gets QuantumTopSort(G)$
\For $v\in\{n,\dots,1\}$
\If{$GROVER(f_v)\neq NULL$}
\State $Win(v)\gets True$
\Else
\State $Win(v)\gets False$
\EndIf
\EndFor
\end{algorithmic}
\end{algorithm}

The algorithm is presented in \cite{kks2019} and \cite{kksk2020} for a multi-player game.

\subsubsection{Complexity}\label{sec:game-cmpl}
Let us discuss the complexity of the algorithm.

The complexity of processing a vertex $v$ is $O(\sqrt{L_v})$, where $L_v$ is the number of neighbors for $v$. We repeat each Grover's search invocation $2\log n$ times for getting $O(\frac{1}{n^2})$ error probability of processing a vertex. In that case, the error probability of processing $n$ vertexes is $O(\frac{1}{n})$. Finally, the total complexity of processing a vertex is $O(\sqrt{L_v}\log n)$.
The total complexity is
\[O\left(\sum_{v\in V}\sqrt{L_v}\log n\right)\leq\]
due to Cauchy–Bunyakovsky-Schwarz inequality
\[ O\left(\sqrt{n\sum_{v\in V}L_v}\log n\right)=\]
Note that $\sum_{v\in V}L_v=O(|E|)=O(m)$ because each directed edge gives us a neighbor in one of the lists and a directed edge give us neighbors in two lists. 
\[=O(\sqrt{nm}\log n)\]

In the case of the Adjacency matrix, the complexity is $O(n^{1.5}\log n)$.

Note that in the classical case, the complexity is $\Theta(m)$ in the case of the list of neighbors and $\Theta(n^2)$ in the case of the Adjacency matrix. These are lower and upper bounds \cite{kks2019}. So, we obtain a quantum speed-up if $m>n(\log n)^2$ in the case of the list of neighbors. If we use the Adjacency matrix, then we have a quantum speed-up anyway.
\subsection{Dynamic Programming on Directed Acyclic Graphs. The Longest Path Problem}
It is known, that the longest path problem is NP-hard for a general graph. At the same time, for a DAG we can solve it in polynomial time.

 \begin{itemize}
 \item We have a DAG $G(V,E)$ with $n$ vertexes and $m$ edges. 
 \item Find the longest path from the starting vertex $A$.
 \end{itemize}
 
 For solving the problem we compute the longest path from each vertex. We use a $TheLongestPath:V\to\{0,\dots,n\}$ function. $TheLongestPath(v)$ is the length of the longest path that starts from a vertex $v$.

Assume, that the graph is topologically sorted. Otherwise, it is the first step. Additionally, we are interested only in vertexes that can be reached from $A$. Therefore, after the Topological sort, the index of $A$ is $1$. Then, for each vertex $i$ from $n$ to $1$, we can say that $TheLongestPath(i)=MAX\{TheLongestPath(v):(i,v)\in E\}$ is the maximal length of a path that started from a vertex $v$. Assume that $MAX$ of an empty set is $0$. Because of the topological sort, we can be sure that $i$ and $v$ satisfy $v>i$, and $TheLongestPath(v)$ is already computed. So, we have Algorithm \ref{alg:longest-dag}.

\begin{algorithm}
\caption{The quantum algorithm for The Longest Path Problem on a DAG.}\label{alg:longest-dag}
\begin{algorithmic}
\State $G\gets TopSort(G)$
\For{$i\in\{n,\dots,1\}$}
\State $TheLongestPath(i)\gets MAX\{TheLongestPath(v):(i,v)\in E\}$
\EndFor
\State\Return $TheLongestPath(A)$
\end{algorithmic}
\end{algorithm}

For searching a maximum we can use the Quantum maximum/minimum search algorithm (Section \ref{sec:max}). Additionally, we use the quantum implementation of the Topological sort (Section \ref{sec:topsort}).

Complexity of processing a vertex $v$ is $O(\sqrt{L_v}\log n)$ if we have $O(\frac{1}{n})$ error probability, where $L_v$ is a number of neighbors for $v$. So, the complexity is similar to the complexity of the previous problem (Section \ref{sec:game-cmpl}). That is $O(\sqrt{nm}\log n)$ in the case of a list of neighbors and $O(n^{1.5}\log n)$ in the case of an adjacency matrix.

Note that the last two problems are an example of the Quantum version of the Dynamic Programming approach for DAGs that was discussed in details in \cite{ks2019,ks2019short}.

\subsection{Hamiltonian Path Problem}\label{sec:hamilt}
Let us consider the following well-known problem.
 \begin{itemize}
 \item We have a graph $G=(V,E)$ with $n$ vertexes and $m$ edges.
 \item We should find a path that 
 \begin{itemize}
     \item contains all vertexes and
     \item each vertex occurs only once in the path. 
 \end{itemize}
 \end{itemize}
The name of the presented problem is ``Hamiltonian path Problem'', and the path that visits all vertexes exactly once is called ``Hamiltonian path''. If the starting and finishing vertexes of the path are the same, then it is a ``Hamiltonian cycle''. It is known that the problem is NP-complete. See \cite{cormen2001} for a detailed description.

\subsubsection{Brute Force Solution}
The first possible solution of the problem is a brute force solution.
We can check all possible permutations of numbers from $1$ to $n$. Each permutation can be a Hamiltonian path. Each vertex index in the sequence occurs exactly once because it is a permutation. We should check that it is a path. In other words, for a permutation $\pi=(v^1,\dots,v^n)$, we should check the property $(v^i,v^{i+1})\in )$ for each $i\in\{1,\dots,n-1\}$. Let us present Algorithm \ref{alg:hp-bf1} that uses two functions:
\begin{itemize}
\item $\textsc{GetPermuation}(i)$ that returns $i$-th permutation.
\item $\textsc{CheckPath}(\pi)$ that checks a permutation $\pi$ whether it is a path.
\end{itemize}

There are $n!$ different permutations.

\begin{algorithm}
\caption{Classical brute force algorithm for the Hamiltonian Path Problem.}\label{alg:hp-bf1}
\begin{algorithmic}
\State $resultPath\gets NULL$\Comment{Returns $NULL$ if the path is not found}
\For{$i\in\{1,\dots,n!\}$}
\State $\pi \gets \textsc{GetPermuation}(i)$
\If{$\textsc{CheckPath}(\pi)=True$}
\State $resultPath\gets \pi$\Comment{Stop the search}
\EndIf
\EndFor
\Return $resultPath$\Comment{Returns $NULL$ if the path is not found}
\end{algorithmic}
\end{algorithm}
The procedures $\textsc{GetPermuation}$ and $\textsc{CheckPath}$ can be implemented with $O(n)$ running time. The total running time is 
\[O(n!\cdot n)=O\left(\sqrt{n}\frac{n^n}{e^n}\cdot n\right)=O\left(n^{1.5}\cdot 2^{n(\log_2 n - \log_2 e)}\right)\] due to Stirling's approximation. 
 
We can consider the problem as a search function $f:\{1,\dots,n!\}\to\{0,1\}$, where $f(i)=1$ iff $\textsc{CheckPath}(\textsc{GetPermuation}(i))=True$.
So, we invoke Grover's search algorithm for the function. The search space size is $n!$ and the complexity of one element checking is $O(n)$. The total complexity is $O(\sqrt{n!}\cdot n)= O\left(n^{1.25}\cdot 2^{0.5n(\log_2 n - \log_2 e)}\right)$.

\subsubsection{Dynamic Programming Solution}
There is another solution that is based on the Dynamic Programming approach \cite{b62,hk62}. Let us describe the solution.

Let  $mask\in\{0,\dots,2^n-1\}$ be a bitmask of a subset of vertexes, and let $V(mask)$ be the corresponding subset. Formally, if $bin(mask)=(mask_1,\dots,mask_n)\in\{0,1\}^n$ is a binary representation of $mask$, then $v_i\in V(mask)$ iff $mask_i=1$.

 Let us consider a function $h:\{0,\dots,2^n-1\}\times \{1,\dots,n\}\times \{1,\dots,n\}\to\{0,1\}$.
 $h(mask,v,u)=1$ iff there is a Hamiltonian path in the subset $V(mask)$ that starts from $v$ and finishes in $u$, where $v,u\in V(mask)$.
 
There are several properties of the function $h$.
\begin{itemize}
\item The single vertex property. $h(mask,v,v)=1$ for any $v\in V$ and $mask$ such that $V(mask)=\{v\}$.
\item The cutting one edge property. $h(mask,u,v)=1$ iff there is $t\in V(mask)$ such that $h(mask',u,t)=1$, $V(mask')=V(mask)\backslash \{v\}$, and there is an edge between $t$ and $v$.
\item The splitting of a subset to two parts property. $h(mask,u,v)=1$ iff there are 
  \begin{itemize}
    \item $mask', mask''\in\{0,\dots, 2^n-1\}$ such that $V(mask')\cup V(mask'')=V(mask)$ and $V(mask')\cap V(mask'')=\emptyset$;
  \item $t\in V(mask'), z\in V(mask'')$ such that $h(mask',u,t)=1$, $h(mask'',z,v)=1$, and there is an edge between $t$ and $z$.  
  \end{itemize} 
\end{itemize}

Let us describe the classical algorithm. Our goal is computing $h(FullSet,v,u)$ for each $v,u\in V$ and a bit mask $FullSet=2^{n}-1$ such that $V(FullSet)=V$. If we can find a pair of vertexes $(v,u)$ such that $h(FullSet,v,u)=1$, then it is a solution.
 In that case, we can solve the decision problem which is checking the existence of a Hamiltonian path. Let us focus on this problem and then we solve the original one.
 
We find $h(mask,v,u)$ recursively using ``The cutting one edge property''. For this reason we check all  neighbors $t$ of $u$ and search $h(mask',v,t)=1$ where $V(mask')=V(mask)\backslash \{u\}$.
At the same time, we can store values of the function $h$ in an array and avoid recomputing the same value several times. In fact, this ``caching'' trick converts the ``brute force'' solution to the ``dynamic programming'' solution.

Let us present the recursive procedure in Algorithm \ref{alg:ham-dp1}. Assume that $\hat{h}$ is an array that stores value of $h$, i.e., $\hat{h}[mask,v,u]=h(mask,v,u)$.

\begin{algorithm}
\caption{The recursive procedure for computing $h$ function. Input parameters are $mask\in\{0,\dots,2^n-1\}, v,u\in V$. Output parameter is $h(mask,v,u)$}\label{alg:ham-dp1}
\begin{algorithmic}
\If{$\hat{h}[mask,v,u]$ is not assigned}
\If{$v=u$ and $V(mask)=\{v\}$}
\State $\hat{h}[mask,v,u]\gets 1$
\Else
\State $mask'=mask-2^u$\Comment{$V(mask')=V(mask)\backslash \{u\}$}
\For{$t$ is a neighbor of $u$}
\If{$h(mask',v,t)=1$}\Comment{a recursive invocation of the current procedure}
\State $\hat{h}[mask,v,u]\gets 1$\Comment{Stop the for-loop}
\EndIf
\EndFor
\EndIf
\EndIf
\State\Return $\hat{h}[mask,v,u]$
\end{algorithmic}
\end{algorithm}

Using this procedure we can solve the decision problem (Algorithm \ref{alg:ham-dp2}).

\begin{algorithm}
\caption{The Dynamic programming solution of the decision version of Hamiltonian path problem.}\label{alg:ham-dp2}
\begin{algorithmic}
\State $result \gets 0$
\For{$v\in V$}
\For{$u\in V$}
\If{$h(2^n-1,v,u)=1$}\Comment{$V(2^n-1)=V$}
\State $result\gets 1$\Comment{There is a Hamiltonian path between $v$ and $u$. We stop the for-loop.}
\EndIf
\EndFor
\EndFor
\State \Return $result$
\end{algorithmic}
\end{algorithm}

For solving the main problem we can store additional information in the array $\hat{F}(mask,v,u)$. Let $\hat{F}(mask,v,u)=t$ if the vertex $t$ precedes $u$ in a Hamiltonian path between $v$ and $u$ in $V(mask)$. For $v\in V, mask'$ such that $V(mask')=v$, we assume $\hat{F}(mask',v,v)=NULL$.
We can change Algorithm \ref{alg:ham-dp1} and obtain Algorithm \ref{alg:ham-dp3}.

\begin{algorithm}
\caption{The recursive procedure for computing $h$ function. Input parameters are $mask\in\{0,\dots,2^n-1\}, v,u\in V$. Output parameter is $h(mask,v,u)$}\label{alg:ham-dp3}
\begin{algorithmic}
\If{$\hat{h}[mask,v,u]$ is not assigned}
\If{$v=u$ and $V(mask)=\{v\}$}
\State $\hat{h}[mask,v,v]\gets 1$
\State $\hat{F}(mask',v,v)\gets NULL$
\Else
\State $mask'=mask-2^u$\Comment{$V(mask')=V(mask)\backslash \{u\}$}
\For{$t$ is a neighbor of $u$}
\If{$h(mask',v,t)=1$}\Comment{a recursive invocation of the current procedure}
\State $\hat{F}(mask,v,v)\gets t$
\State $\hat{h}[mask,v,u]\gets 1$\Comment{Stop the for-loop}
\EndIf
\EndFor
\EndIf
\EndIf
\State\Return $\hat{h}[mask,v,u]$
\end{algorithmic}
\end{algorithm}

We can change Algorithm \ref{alg:ham-dp2} and obtain Algorithm \ref{alg:ham-dp4}.

\begin{algorithm}
\caption{The Dynamic programming solution of the decision version of Hamiltonian path problem.}\label{alg:ham-dp4}
\begin{algorithmic}
\State $result \gets 0$
\For{$v\in V$}
\For{$u\in V$}
\If{$h(2^n-1,v,u)=1$}\Comment{$V(2^n-1)=V$}
\State $resultV\gets v$
\State $resultU \gets u$
\State $result\gets 1$\Comment{There is a Hamiltonian path between $v$ and $u$. We stop the for-loop.}
\EndIf
\EndFor
\EndFor
\If{$result= 1$}
\State $HamiltonianPath\gets(resultU)$
\State $t \gets resultU$
\State $mask = 2^n-1$\Comment{$V(2^n-1)=V(FullSet)=V$}
\While{$t\neq result V$}
\State $t'\gets \hat{F}(mask,v,t)$
\State $mask = mask-2^t$\Comment{Excluding $t$ from the set $V(mask)$}
\State $t\gets t'$
\State $HamiltonianPath\gets t\circ  HamiltonianPath$ \Comment{Adding $t$ to the begin of the path}
\EndWhile
\Else 
\State $HamiltonianPath\gets NULL$
\EndIf
\State \Return $result, HamiltonianPath$
\end{algorithmic}
\end{algorithm}

Let us discuss the quantum algorithm for the problem that was presented in \cite{abikpv2019}.
Let us consider the following sets
 \begin{itemize}
\item Let $M_2\subset \{0,\dots, 2^n-1\}$ be such that 
$mask\in M_2$ satisfies $|V(mask)|=\lceil n/2\rceil$;
\item Let $M_2'\subset \{0,\dots, 2^n-1\}$ be such that $mask\in M_2'$ satisfies $|V(mask)|=\lfloor n/2\rfloor$;
 \item Let $M_4\subset \{0,\dots, 2^n-1\}$ be such that $mask\in M_4$ satisfies $|V(mask)|=\lceil n/4\rceil$. 
 \item Let $M_4'\subset \{0,\dots, 2^n-1\}$ be such that $mask\in M_4'$ satisfies $|V(mask)|=\lfloor n/4\rfloor$.
 \item For $mask_2\in M_2$ or $mask_2\in M_2$', let $M_4(mask_2)\subset \{0,\dots, 2^n-1\}$ be such that $mask\in M_4(mask_2)$ satisfies 
 \begin{itemize}
 \item $|V(mask)|=\lceil n/4\rceil$ and
 \item $V(mask)\subset V(mask_2)$.
 \end{itemize} 
 \item Similarly,  For $mask_2\in M_2$ or $mask_2\in M_2$', let $M_4'(mask_2)\subset \{0,\dots, 2^n-1\}$ be such that $mask\in M_4'(mask_2)$ satisfies 
 \begin{itemize}
 \item $|V(mask)|=\lfloor n/4\rfloor$ and
 \item $V(mask)\subset V(mask_2)$.
 \end{itemize}  
 \end{itemize} 

The algorithm is following.
 As a first step, we classically compute $h(mask,u,v)$ for all $mask$ such that $|V(mask)|\leq n/4+1$ and $v,u\in V$. We do it using an invocation of Algorithm \ref{alg:ham-dp3} for each $mask\in M_4'$ and each $v,u\in V$.

For $mask\in M_2, u,v\in V(mask)$, let $f_{mask,u,v}:M_4(mask)\times V(mask)\to \{0,1\}$ be the first search function. The value $f_{mask,u,v}(mask_4,t)=1$ iff 
 \begin{itemize}
 \item $u,t\in V(mask_4)$, $h(mask_4,u,t)=1$,
 \item for $mask': V(mask')= V(mask)\backslash V(mask_4)$, there is $z\in V(mask)\backslash V(mask_4)$ such that $h(mask',z,v)=1$ and
 \item an edge between $t$ and $z$ exists.
 \end{itemize}     
 Note, that we can compute $mask'=mask - mask_4$.
 The values $h(mask_4,u,t)$ and $h(mask',z,v)$ are already computed and stored in the array $\hat{h}$. The computing $f_{mask,u,v}(mask_4,t)$ process is searching a required $z$. It can be done using Grover's search. The running time of computing $f_{mask,u,v}(mask_4,t)$ is $O(\sqrt{n})$. 

We want to compute 	$h(mask,u,v)$ for $mask\in M_2,u,v\in V$.  Note that $h(mask,u,v)=1$ iff there is a $mask_4'$ and $t'$ such that $f_{mask,u,v}(mask_4',t')=1$ due to ``The splitting a subset to two parts property'' of $h$. We can use Grover's Search for searching the argument $(mask_4',t')$.

 We define $f'_{mask,u,v}$ similar to $f_{mask,u,v}$, but $mask\in M_2'$.
 
For $u,v\in V$, let $f_{u,v}:M_2\times V\to \{0,1\}$ be the second search function.
The value $f_{u,v}(mask_2,t)=1$ iff
 \begin{itemize}
 \item for $u,t\in V(mask_2)$, we have $h(mask_2,u,t)=1$.
 \item for $mask': V(mask')= V\backslash V(mask_2)$, there is $z\in V\backslash V(mask_2)$ such that $h(mask',z,v)=1$.
 \item An edge between $t$ and $z$ exists.
\end{itemize}    

Computing $f_{u,v}(mask_2,t)$ is searching $z$ such that corresponding values of $h$ are $1$, in other words, $h(mask',z,v)=1$ and $h(mask_2,u,t)=1$. Searching $z$ can be done using Grover's Search in $O(\sqrt{n})$ running time. At the same time, computing $h(mask_2,u,t)$ and $h(mask',z,v)$ itself requires Grover's Search as it was discussed before. Computing $h(mask_2,u,t)$ requires Grover's Search for $f_{mask_2,u,t}$  function; and computing $h(mask',z,v)$ requires Grover's Search for $f'_{mask',z,v}$  function, where $V(mask_2')=V\backslash V(mask_2)$, i.e. $mask'=(2^n-1) - mask_2$.

Finally, for any two vertexes $u,v\in V$, we have $h(FullSet,u,v)=1$ iff there is $mask_2'$ and $t'$ such that $f_{u,v}(mask_2',t')=1$. So, we can compute it using Grover's Search by $(mask_2,t)$ argument.

 Then, we search $u,v\in V$ such that $h(FullSet,u,v)=1$ for a Hamiltonian path. We can use Grover's search by $(u,v)$ argument here. 
 In the case of a Hamiltonian cycle, we can check whether  $h(FullSet,v_1,v_1)=1$ for the fixed constant vertex $v_1$ because the cycle should visit all the vertexes including the vertex $v_1$.

 Let us discuss the complexity of the Algorithm.
\begin{lemma}
The running time of the presented algorithm is $O(1.755^n\cdot n^3)$ for a Hamiltonian path, and $O(1.755^n\cdot n^2)$ for a Hamiltonian cycle. The error probability is constant.
\end{lemma}
\begin{proof}
Firstly, let us note that we have several nested Grover's search algorithms. It means, all except the deepest Grover's search has bounded-error input; and we should use the random-oracle version of Grover's Search (Section \ref{sec:random-oracle}).

Let us compute the complexity of the classical part. There are $\binom{n}{n/4}$ different elements of $M_4$. For computing all functions $h(mask,u,v)$ for $mask\in M'_4, u,v\in V$ we should compute all elements $\hat{h}(mask,u,v)$ for  $V(mask)\leq n/4, u,v\in V$. The total number of elements is
\[n^2\cdot\left(\binom{n}{1}+\binom{n}{2}+\dots+\binom{n}{n/4}\right)\]
The complexity is 
\[O\left(n^3\cdot\left(\binom{n}{1}+\binom{n}{2}+\dots+\binom{n}{n/4}\right)\right)\]

The complexity of the quantum part is the following
\begin{itemize}
\item Computing $f_{mask,u,v}$ takes $O(\sqrt{n})$ running time.
\item For $mask\in M_2$ or $mask\in M_2'$, $u,v\in V(mask)$, computing $h(mask,u,v)$ takes $O \left(\sqrt{n\binom{n}{n/4}}\cdot\sqrt{n}\right)=O \left(n\sqrt{\binom{n}{n/4}}\right)$. The search space size is $n\binom{n}{n/4}$ for the pair of arguments $(mask_4,t)$ and complexity of computing $f_{mask,u,v}$ is $O\left(\sqrt{n}\right)$.
\item Computing $f_{u,v}$ takes $O\left(\sqrt{n}\cdot n\sqrt{\binom{n}{n/4}}\right)=O\left(n^{1.5}\sqrt{\binom{n}{n/4}}\right)$ running time. The search space size is $n$ for the argument $z$ and complexity of computing $h(mask,u,v)$ is $O \left(n\sqrt{\binom{n}{n/4}}\right)$.
\item Computing  $h(FullSet,u,v)$ takes $O \left(\sqrt{n\binom{n}{n/2}}\cdot n^{1.5}\sqrt{\binom{n}{n/4}}\right)=O \left(n^{2}\sqrt{\binom{n}{n/2}\binom{n}{n/4}}\right)$. The search space size is $n\binom{n}{n/2}$ for the pair of arguments $(mask_2,t)$ and complexity of computing $f_{u,v}$ is $O\left(n^{1.5}\sqrt{\binom{n}{n/4}}\right)$.
\item The quantum part for a Hamiltonian path takes $O \left(\sqrt{n^2}\cdot n^{2}\sqrt{\binom{n}{n/2}\binom{n}{n/4}}\right)=O\left(n^{3}\sqrt{\binom{n}{n/2}\binom{n}{n/4}}\right)$ running time. The search space size is $n^2$ for the pair of arguments $(u,v)$ and complexity of computing $h(FullSet,u,v)$ is $O\left(n^{2}\sqrt{\binom{n}{n/2}\binom{n}{n/4}}\right)$. 
\item The quantum part for a Hamiltonian cycle takes $O \left(n^{2}\sqrt{\binom{n}{n/2}\binom{n}{n/4}}\right)$ running time. It is exactly the running time for computing $h(FullSet,v_1,v_1)$. 
\end{itemize}
The error probability is constant according to Section \ref{sec:random-oracle}.

The total complexity is
\begin{itemize}
\item for a Hamiltonian path:
\[O\left(n^3\cdot\left(\binom{n}{1}+\binom{n}{2}+\dots+\binom{n}{n/4}\right)+n^{3}\sqrt{\binom{n}{n/2}\binom{n}{n/4}}\right)\approx O(1.755^n\cdot n^3)=O^*(1.755^n),\]
\item for a Hamiltonian cycle:
\[O\left(n^3\cdot\left(\binom{n}{1}+\binom{n}{2}+\dots+\binom{n}{n/4}\right)+n^{2}\sqrt{\binom{n}{n/2}\binom{n}{n/4}}\right)\approx  O(1.755^n\cdot n^2)=O^*(1.755^n).\]
\end{itemize}
Here $O^*$ hides a log factor.
\end{proof} 
 
This algorithm is splinting the path into two parts of size $n/2$ and is splitting each of these parts into two sub-parts of size $n/4$. Then we use classically precomputed information about paths of size $n/4$. At the same time, we can split the path once more time.

Classically we compute $h(mask,u,v)$ for all $mask:V(mask)\leq (1-\alpha )n/4$ where $\alpha$ is some constant. Then, we split sets $V(mask_4)$ of size $n/4$ to two parts: $V(mask_\alpha)\subset V(mask_4)$ and $ V(mask_4)\backslash V(mask_\alpha)$ and search the result similar to $f_{mask,u,v}$ by Grover's Search.

The total complexity in that case is 
\begin{itemize}
\item for a Hamiltonian path:
\[O\left(n^3\cdot\left(\binom{n}{1}+\binom{n}{2}+\dots+\binom{n}{(1-\alpha)n/4}\right)+n^{4}\sqrt{\binom{n}{n/2}\binom{n}{n/4}\binom{n}{(1-\alpha)n/4}}\right),\]
\item for a Hamiltonian cycle:
\[O\left(n^3\cdot\left(\binom{n}{1}+\binom{n}{2}+\dots+\binom{n}{(1-\alpha)n/4}\right)+n^{3}\sqrt{\binom{n}{n/2}\binom{n}{n/4}\binom{n}{(1-\alpha)n/4}}\right).\]
\end{itemize}
$\alpha=0.055362$ minimizes these functions and the complexity is
$O^*(1.728^n)$.

We can compute the path itself using $\hat{F}(mask,u,v)$ similar to the classical case.
\subsection{Traveling Salesman Problem}
Let us consider the following well-known problem.
 \begin{itemize}
 \item We have a weighted graph $G=(V,E)$ with $n$ vertexes and $m$ edges. Let $w(u,v)$ be a weight of the edge between vertexes $u$ and $v$. If there is no edge between $u$ and $v$, then $w(u,v)=+\infty$.
 \item We should find a Hamiltonian path with minimal weight. The weight of a path is the sum of weights of edges that belong to the path.
\end{itemize}
Like the Hamiltonian Path Problem, the Traveling Salesman Problem(TSP) is NP-complete. 

We can use a similar solution as for  Hamiltonian Path Problem. We consider a function $h:\{0,\dots,2^n-1\}\times V\times V\to \mathbb{R}_{+0}$, where $\mathbb{R}_{+0}$ is the set of non-negative real numbers. The value  $h(mask,v,u)$ is the minimal weight of a Hamiltonian path in the subset $V(mask)$ that starts from $v$ and finishes in $u$, where  $v,u\in V(mask)$. We assume that $h(mask,v,u)=+\infty$ if such a Hamiltonian path does not exist. The function has the following properties:
  
 \begin{itemize}
\item The single vertex property. $h(mask,v,v)=0$ for any $v\in V$ and $mask$ such that $V(mask)=\{v\}$.
\item The cutting one edge property. If $h(mask,u,v)\neq +\infty$, then there is $t\in V(mask)$ such that $h(mask,u,v)=h(mask',u,t)+w(t,v)$.
\item The splitting of a subset to two parts property. if $h(mask,u,v)\neq +\infty$, then there are 
  \begin{itemize}
    \item $mask', mask''\in\{0,\dots, 2^n-1\}$ such that $V(mask')\cup V(mask'')=V(mask)$ and $V(mask')\cap V(mask'')=\emptyset$;
  \item $t\in V(mask'), z\in V(mask'')$ such that $h(mask,u,v) =h(mask',u,t)+ w(t,z)+ h(mask'',z,v)$.  
  \end{itemize} 
\end{itemize}

We can use exactly the same algorithm as for Hamiltonian Path Problem, but we replace Grover's search algorithm with the Quantum minimum/maximum search algorithm (Section \ref{sec:max}).
The complexity of the algorithm for TSP is $O^*(1.728^n)$. It is presented in \cite{abikpv2019}.

There is another application of this idea for the following problems:
\begin{itemize}
\item Path in the Hypercube \cite{abikpv2019},
\item Vertex Ordering Problems \cite{abikpv2019},
\item Graph Bandwidth \cite{abikpv2019},
\item Feedback Arc Set \cite{abikpv2019},
\item Minimum Set Cover \cite{abikpv2019},
\item The Shortest Substring Problem \cite{kb2022},
\item The Longest Trail Problem \cite{kk2021trail},
\item Minimum Steiner Tree Problem \cite{mikl2020} and others.
\end{itemize} 
\section{Applications of Grover's Search Algorithm and Meet-in-the-middle Like Ideas.}
%Subsetsum Problem and Collision Problem
\subsection{Subsetsum Problem (0-1 Knapsack Problem)}
Problem.
 \begin{itemize}
 \item Let us have $n$ positive integers $a=(a_0,\dots,a_{n-1})$. 
 \item Let us have a positive integer $k$. 
 \item We want to find a subset $\{i_1,\dots,i_m\}\subset\{0,\dots,n-1\}$ such that
 $a_{i_1}+\dots+a_{i_m}=k$
 \end{itemize}
 An example: $n=5$ $a=(3,7,4,9,12)$, $k=19$.
A possible solutions: $(3,7,9)$ or $(7,12)$. We can return any of them.

The simple version is the following. We should answer ``Yes'' if we can find such a set, and ``No'' otherwise.

We know, that if $a_i$ and $k$ are in order of $2^n$, then
the problem is NP-complete;
and only exponential-time solutions are known!
 If $k=poly(n)$, then 
  there is a dynamic programming solution in $O(nk)$.
\subsubsection{Brute Force Solution}
 Let us consider all possible subsets $2^{[n]}$ of set $[n]=\{0,\dots,n-1\}$. For each of them, we can compute a sum of the set and compare it with $k$.
Let us consider a search function $f:\{0,\dots,2^n-1\}\to \{0,1\}$.
  The function is such that $f(i)=1$ iff the $i$-th subset has a sum $k$.
 We can say that the $i$-th subset is a subset with a bit mask $i$ or a characteristic vector that is the binary representation of $i$.
  
  For example:  $a=(3,7,4,9,12)$. If $i=10$, then  $bin(i)=(0,1,0,1,0)$ and the  subset is $\{7,9\}$ and $a_1+a_3=7+9=16$.
  
Complexity of computing $f(i)$ is $V(f)=O(n)$.

The classical solution is the linear search. We check all possible $i$ and search $f(i)=1$. Complexity is $O(2^n\cdot V(f))=O(n2^n)$.

The quantum solution is Grover's search. Complexity is $O(\sqrt{2^n}\cdot V(f))=O(n2^{n/2})$.
Note, that memory complexity is $O(n)$.

\subsubsection{Meet-in-the-Middle Solution}
Let us discuss the classical version of the Meet-in-the-Middle Solution.
We split the sequence $a$ to two parts $(a_1,\dots,a_t)$ and $(a_{t+1},\dots,a_n)$ for $t=\lfloor n/2 \rfloor$.

As the first step, we check all subsets of the first part and store all possible sums of these subsets in a set $Sums$. We implement the set using Self-Balanced Binary Search Tree \cite{cormen2001}.

As the second step, we check all subsets of the second part. We consider a search function $f:\{0,\dots,2^{n-t}-1\}\to \{0,1\}$. The value $f(mask)=1$ iff the subset $\{j_1,\dots,j_{m''}\}\subset\{t+1,\dots,n\}$ from the second part, which characteristic vector is the binary representation of $mask$, has sum $a_{j_1}+\dots+a_{j_{m''}}=s$  and there is the value  $k-s$ in the set $Sums$.
If we find a subset $a_{j_1},\dots,a_{j_{m''}}$ from the second part with a sum $s$ and the corresponding subset $a_{i_1},\dots,a_{i_{m'}}$ from the first part with a sum $k-s$, then we can construct a subset $a_{i_1},\dots,a_{i_{m'}},a_{j_1},\dots,a_{j_{m''}}$ with sum $k$ using elements from the whole sequence. In other words, we search for $mask: f(mask)=1$. We can do it using the linear search by all arguments $mask\in\{0,\dots 2^{n-t}-1\}$.

Let us compute the complexity of the algorithm. There are at most $2^t$ different sums for the first part. The complexity of adding an element to the set (Self-Balanced Binary Search Tree) is $O(\log(2^t))=O(t)$. The complexity of the first step (adding all sums to $Sums$) is $O(t\dot 2^t)$. 
There are $2^{n-t}$ different subsets for the second part. The complexity of computing one of them is $O(n-t)$ for computing a sum $s$ and $O(t)$ for searching $k-s$ in the set $Sums$. So, the complexity is $O(n-t+t)=O(n)$. The complexity of the second step is  $O(2^{n-t}n)$, and the total complexity is  $O(2^t\cdot t + 2^{n-t}\cdot n)=O(2^{n/2}n)$ if $t=n/2$.

Note that here we should store at most $2^t$ different sums in memory. Therefore, memory complexity is $O(2^n/2)$! That is exponentially bigger than in the brute force case. At the same time, we have achieved good running time as for the quantum version of the brute force algorithm.

Let us discuss the quantum version of the algorithm. We do the following changes:
\begin{itemize}
\item We change $t=n/3$.
\item We do the first step classically.
\item We replace the linear search of a $1$-argument of $f$ with Grover's Search.
\end{itemize}

Let us compute the complexity of the algorithm. The complexity of the first step is $O(2^{n/3}\cdot n)$. Complexity of the second step is $O(\sqrt{2^{2n/3}}\cdot n)=O(2^{n/3}\cdot n)$. The total complexity is $O(2^{n/3}\cdot n+2^{n/3}\cdot n)=O(2^{n/3}\cdot n)$.

The algorithm is presented in \cite{as2003}.

\subsection{Collision Problem}\label{sec:collision}
Problem.
\begin{itemize}
 \item We have a positive integer $n$ and positive integers $a_1,\dots,a_n$. There are two cases:
 \begin{enumerate}
 \item all elements are unique
 \item for any $a_i$ there is only one $a_j$ such that $i\neq j$ and $a_i=a_j$
 \end{enumerate}
 \item Our goal is to distinguish these two cases. We should return $1$ if it is the first case, and $0$ if it is the second case. 
 \end{itemize}
  In other words, we should distinguish $1$-to-$1$ and $2$-to-$1$ functions.

Firstly, let us discuss a simple solution with big running time. We start with the classical solution. In the second case, any element has a pair. Let us take the first number $a_1$. It should have a pair also. If it does not have a pair, then it is the first case (a $1$-to-$1$ function). We search $j$ such that $a_j=a_1$. In other words, we consider a search function $f:\{2,\dots,n\}\to\{0,1\}$. The value $f(j)=1$ iff $a_j=a_1$. If we found $f(j)=1$ and $j$ exists, then it is a $2$-to-$1$ function (the second case), else it is a $1$-to-$1$ function (the first case). Classically, we search a $1$-argument using linear search. Complexity is $O(n)$.

The quantum version of the algorithm uses Grover's search algorithm for searching the required $1$-argument for the search function $f$. Complexity is $O(\sqrt{n})$.

Let us use the meet-in-the-middle idea. Let us split the sequence into two parts. The first part is $a_1,\dots, a_t$, and the second one is $a_{t+1},\dots, a_n$.

Step 1. We read all variables from the first part and check whether collisions exist in the first part. If there is a collision, then we return that it is a ``$2$-to-$1$ function'' (the second case).
We can do it, for example, in the following way. We read all of these elements and store them into a set $S$ that is implemented by Self-Balanced Search Tree \cite{cormen2001}. Before storing an element, we check whether an element with the such value exists. If it is true, then it is a collision.

Step 2. We consider a search function for the second part that is $f_2:\{t+1,\dots,n\}\to \{0,1\}$. The value $f_2(j)=1$ iff $a_j$ belongs to the set $S$. We search an argument $j$ such that $f(j)=1$ using Grover's Search algorithm. If such an argument $j$ exists, then it is a ``$2$-to-$1$ function'' (the second case), else it is a ``$1$-to-$1$ function'' (the first case).

Let us discuss the complexity. The running time for Step 1 is $O(t\log t)$. Step 2 uses Grover's search algorithm for the search space of size $n-t$. If it is a ``$2$-to-$1$ function'' (the second case), then each element from the first part has exactly one pair element. Therefore, there are exactly $t$ solutions (arguments with $1$-value of the search function $f_2$), and Grover's search does $O(\sqrt{\frac{n-t}{t}})$ steps. The complexity of computing a value $f_2(j)$ is $
 O(\log t)$. So, the complexity of Step 2 is $O\left(\left(\frac{n-t}{t}\right)^{1/2}\log t\right)$.
Total complexity is $O\left(t\log t + \left(\frac{n-t}{t}\right)^{1/2}\log t\right)$. Let us take $t=n^{1/3}$. Total complexity is $O(n^{1/3}\log n + n^{(2/3)/2}\log n)=O(n^{1/3}\log n)$.
We can use another way for implementing the set $S$. If $a_i\in\{0,\dots,m-1\}$ and $m$ is not very big, then we can use, for example, a boolean array of size $m$ for the set $S$. In that case, the size of memory becomes $O(m)$, but the running time for search/insert operation for the set is $O(1)$ and the total complexity is $O( n^{1/3})$. Another way to obtain constant complexity for search/insert operation is using Hash tables.
At the same time, the query complexity of the algorithm is $O(n^{1/3})$ because it counts queries of input variables.

The algorithm is presented in \cite{bht97collarx}.

\section{Applications of Grover's Search Algorithm for String Problems.}
In this section, we compare strings in lexicographical order.

Firstly, let us discuss the useful data structure, that implements the ``Sorted set of strings'' data structure.

Assume that we have a sequence of strings $s^1,\dots,s^n$ of a length $k$. We fix the length for simplicity of explanation. At the same time, it is not necessary to have a fixed length. We use the Self-balanced Binary Search Tree as the data structure \cite{cormen2001}.

The data structure allows us to add, delete and search elements in $O(\log n)$ running time.
As keys in the nodes of the tree, we store indexes $i\in\{1,\dots,n\}$.
If an index $i$ is stored in a node, then any node in the left subtree has an index $j$ such that $s^j<s^i$; and any node in the right subtree has an index $j$ such that $s^j>s^i$.
As the comparing procedure for strings, we use the quantum comparing procedure from Section \ref{sec:str-cmp}.

The comparing procedure is noisy. If we use it as is, then we obtain a high error probability. We can repeat the comparing procedure $O(\log n)$ times to obtain $O(\frac{1}{n^3})$ error probability; and $O(\frac{\log n}{n^3})=O(\frac{1}{n^2})$ error probability for add, delete, search operations of the tree. In that case, the running time of comparing two strings is $\bigo{\sqrt{k}\log n}$. Add, delete, and search operations do $O(\log n)$ comparing operations. Therefore, the running time for these operations is $\bigo{\sqrt{k}(\log n)^2}$. At the same time, there is an advanced technique  \cite{ke2022,ksz2022} based on the Random Walks algorithm that allows us to obtain the same error probability $O(\frac{1}{n^2})$ with $\bigo{\sqrt{k}\log n}$ running time for add, delete, search operations. 

There are additional operations that the Self-Balanced Binary Search Tree provides. They are minimal end maximal elements search. It can be done in constant running time and correct if the structure of the tree is correct. It is correct if all previous add/delete operations were correct. In other words, if we have done $u$ add/delete operations, then the probability of error for minimum/maximum search is $O(\frac{u}{n^2})$.

The data structure is discussed in \cite{ke2022,ki2019,ksz2022}.

Using the ``Sorted set of strings'' data structure we implement quantum algorithms for several problems that are presented in the following sections.
\subsection{Strings Sorting}
Problem
\begin{itemize}
\item There are $n$ strings of the length $k$ that are $s_1,\dots,s_n$.
\item We want to sort them in the lexicographical order.
\end{itemize}

It is known \cite{hns2001,hns2002} that no quantum algorithm can sort arbitrary comparable objects faster than $O(n\log n)$. At the same time, several researchers tried to improve the hidden constant \cite{oeaa2013,oa2016}. In the case of string sorting, we can obtain quantum speed-up.

Let us discuss the classical algorithm for the problem. Strings can be sorted using Radix sort algorithm \cite{cormen2001}. The running time of the algorithm is $O((k+p)n)$, where $p$ is the size of the alphabet. If the size of the alphabet is constant, then the running time is $O(kn)$. At the same time, it is the lower bound too \cite{kiv2022}. So, the classical complexity is $\Theta(kn)$.
The main idea of sorting is sorting by the $k$-th symbol using the counting sort. Then, by $(k-1)$-th symbol and so on.

The quantum sorting algorithm can use different ideas. One of the presented here. 

Step 1. Add all strings to the ``Sorted set of strings'' data structure that was discussed below.

Step 2. Get and remove from the ``Sorted set of strings'' data structure the index of the minimal string and add it to the end of the result list.

A similar idea is used in HeapSort algorithm \cite{cormen2001}. At the same time, if we use Trie (Prefix tree) data structure \cite{cormen2001} as an implementation of  ``Sorted set of strings'', then we obtain the mentioned complexity $O(nk)$ in a case of the constant size of the alphabet.

Let $T$ be the Self-Balanced Binary Search Tree with strings as keys. We use it as a ``Sorted set of strings'' $T$. The procedure $\textsc{Add}(T,i)$ adds the string $s^i$ to the set. The procedure $\textsc{PopMin}(T)$ returns the index of the minimal string in the set and removes it. Both operations work with $O(\sqrt{k}\log n)$ running time as it was discussed below. The sorting algorithm is presented in Algorithm \ref{alg:sort}.

	\begin{algorithm}
\caption{Quantum algorithm for the String Sorting problem}\label{alg:sort}
\begin{algorithmic}
\For{$i\in\{1,\dots,n\}$}
\State $\textsc{Add}(T,i)$
\EndFor
\State $order\gets()$\Comment{Empty list}
\For{$i\in\{1,\dots,n\}$}
\State $order \gets order \cup \textsc{PopMin}(T)$
\EndFor
\end{algorithmic}
\end{algorithm}
Let us discuss complexity of the algorithm. The main operations with ``Sorted set of strings'' has $O(\sqrt{k}\log n)$ running time and $O(\frac{1}{n^2})$ error probability. If we repeat them $2n$ times, then the total running time is $O(\sqrt{k}\cdot n\log n)$ and $O(\frac{1}{n})$ error probability.

Note, that the lower bound \cite{ksz2022} for quantum sorting of strings is $\Omega(\sqrt{k}\cdot n/\log n)$. So, the presented algorithm is close to the lower bound.  

Different versions of the sorting algorithm and lower bounds for classical and quantum cases are presented in \cite{ki2019,kiv2022,ksz2022}. 
\subsection{The Most Frequent String Search Problem}
Problem
\begin{itemize}
\item There are $n$ strings of the length $k$ that are $s^1,\dots,s^n$.
\item We want to find the minimal index of the most frequent string. Let us define it formally. 
\begin{itemize}
\item Assume $\#(u)=|\{i:s^i=u\}|$ is the number of occurrence of a string $u$.
\item Let $u_{max}=s^j$ for some $j\in\{1,\dots, n\}$ such that $\#(u_{max})=\max\limits_{i\in\{1,\dots, n\}}\#(s^i)$.
\item We search the minimal index $j_0=min\{j: s^j=u_{max}\}$.
\end{itemize}
\end{itemize}
There are different solutions in the classical case. One of them is sorting strings using Radix Sort \cite{cormen2001}, then all equal strings are situated sequentially. Here searching for the largest segment of equal strings is a solution to our problem.
Let us present the algorithm. We check strings one by one, if we found an unequal pair of sequential strings, then it is the end of a segment of equal strings. We return the length of the longest segment of equal strings.
Assume that we have $\textsc{Compare}(u,v)$ procedure from Section \ref{sec:str-cmp} that compares two strings $u$ and $v$ for equality. It returns $0$ iff strings equal.
\begin{algorithm}
\caption{Deterministic solution for the Most Frequent String Search Problem}
\begin{algorithmic}
\State $\textsc{Sort}(s^1,\dots,s^n)$
\State $lastBorder\gets 0$
\For{$i\in\{1,\dots,n-1\}$}
\If{$\textsc{Compare}(s^i,s^{i+1})\neq 0$}
\State $SegmentLength\gets i-lastBorder+1$
\If{$SegmentLength>ans$}
\State $ans\gets SegmentLength$
\EndIf
\State $lastBorder\gets i+1$
\EndIf
\EndFor\State $SegmentLength\gets n-lastBorder+1$
\If{$SegmentLength>ans$}
\State $ans\gets SegmentLength$
\EndIf
\State \Return $ans$
\end{algorithmic}
\end{algorithm}

 The complexity of Radix Sort is $O(nk)$, complexity of searching the largest segment is also $O(nk)$. The total complexity is $O(nk)$. At the same time, the lower bound for the problem is also $\Omega(nk)$. Finally, we obtain the classical complexity of the problem $\Theta(nk)$. (Lower bound is presented in \cite{kiv2022}). 
 
The quantum version of the algorithm uses the quantum sorting algorithm and the quantum string comparing algorithm. So, the complexity of the sorting algorithm is $O(\sqrt{k}n\log n)$. The complexity of comparing algorithm is $O(\sqrt{k}\log n)$ (We repeat it to obtain $O(\frac{1}{n^2})$ error probability). The final complexity is $O(\sqrt{k}n\log n + \sqrt{k}n\log n)=O(\sqrt{k}n\log n)$. 
 
 A detailed description of the algorithm and alternative solutions are presented in \cite{kiv2022, ki2019, ke2022}.
\subsection{Dyck language}
Let us consider a string that contains only parenthesis ``('' or ``)''. We want to check that the parentheses are balanced. The problem is well-known and important. In theoretical computer science, it is known as Dyck language recognition.
Example:
\begin{itemize}
\item the string ``((()()))'' belongs to Dyck language,
\item  the string (())\textbf{)} does not belong to Dyck language.
\end{itemize}

Let us consider the following problem $Dyck_{k,n}$:

\begin{itemize}
\item Input is a boolean string of a length at most $n$. $0$ means opening parenthesis,  $1$ means closing parenthesis.
\item A depth of parentheses is at most $k$.
\end{itemize}

The answer is $1$ if an input string is a well-balanced parentheses sequence and the depth is at most $k$. Otherwise, the result is $0$.  

Example:
\begin{itemize}
\item ((()())) belongs to $Dyck_{3,8}$ or $Dyck_{4,8}$,
\item ((()())) does not belong to $Dyck_{2,8}$ or $Dyck_{1,8}$.
\end{itemize} 

%Assume that we encode the opening parenthesis $($ using $0$ and the closing parenthesis $)$ using $1$.
Let an input string be $x=(x_1,\dots,x_n)$ and $x[l,r]=(x_l,\dots,x_r)$ be a substring of $x$. Let $x[1,r]$ be a prefix of the string $x$. Let $h(x,r)=\#_0(x[1,r])-\#_1(x[1,r])$ be a balance of opening and closing parenthesizes in the prefix $x[1,r]$, where $\#_0(x[1,r])$ and $\#_1(x[1,r])$ are numbers of $0$ and $1$ symbols, respectively. We assume that $h(x,0)=0$. Let $g(x[l,r])=h(x,r)-h(x,l-1)$ be a balance of opening and closing parenthesizes in the substring $x[l,r]$.  

If we check the following three properties, then we can be sure, that an input string is $1$-instance.
\begin{itemize}
\item the balance for $x$ is $0$, i.e. $g(x)=0$;
\item there is no negative-balance prefix, i.e. $g(x[1,r])\geq 0$ for $r\in\{1,\dots,n\}$;
\item there is no $(k+1)$-balance prefix, i.e. $g(x[1,r])\leq k$ for $r\in\{1,\dots,n\}$.
\end{itemize}

For checking these three properties classically, we can use a counter for computing a balance on a prefix and check all three conditions. The running time of the algorithm is $O(n)$. At the same time, the lower bound is also $\Omega(n)$.

Let us consider the quantum algorithm. We want to reformulate the problem as a search problem.
Let us consider the new input $u=1^k x 0^k$, where $1^k$ is the string of $k$ times $1$s (closing parenthesizes) and $0^k$ is the string of $k$ times $0$s (opening parenthesizes). 

Let a substring $x[l,r]$ be $+t$-substring if $g(u[l,r])=t$ and it is minimal, i.e. there is no $(l',r')$ such that $l\leq l'\leq r'\leq r$, $(l,r)\neq (l',r')$ and $g(x[l',r'])=t$.
Let a substring $x[l,r]$ be $-t$-substring if $g(x[l,r])=-t$ and it is minimal, i.e. there is no $(l',r')$ such that $l\leq l'\leq r'\leq r$, $(l,r)\neq (l',r')$ and $g(x[l',r'])=-t$.
Let a substring $x[l,r]$ be $\pm t$-substring if it is  $+t$-substring or $-t$-substring.
Let a sign of a substring $x[l,r]$ be $sign(x[l,r])=1$ if $x[l,r]$ is a $+t$-substring; and let $sign(x[l,r])=-1$ if $x[l,r]$ is a  $-t$-substring.

The computing of $Dyck_{k,n}(x)$ is equivalent to searching any $\pm (k+1)$-substring in $u=1^k x 0^k$. Let $m=|u|=2k+n$ be the length of $u$.
\begin{itemize}
\item If $g(x)>0$, then there is $l$ such that $u[l,m]$ is a $+(k+1)$-substring.
\item If $g(x[1,r])<0$ for some $r\in\{1,\dots,n\}$, then $u[1,k+r]$ is a $-(k+1)$-substring;
\item If $g(x[1,r])= k+1$ for some $r\in\{1,\dots,n\}$, then $u[k+1,k+r]$ is a $+(k+1)$-substring;.
\end{itemize}
If we found a $\pm (k+1)$-substring, then $Dyck_{k,n}(x)=0$, and $Dyck_{k,n}(x)=1$ otherwise.

\subsubsection{$\pm 2$-substring Search}
Let us discuss how to search $\pm t$-substring if $t=2$. 
Let us consider a search function $f:\{1,\dots,n-1\}\to\{0,1\}$ such that $f(i)=1$ iff $x_i= x_{i+1}$. If we find any $i'$ such that $f(i')=1$, then the substring $x[i',i'+1]$ is a $\pm 2$-substring. If $x_{i}=0$, then it is a $+2$-substring, and a $-2$-substring otherwise.
We use Grover's Search for searching $i'$ argument for $f$. The complexity of the search is $O(\sqrt{n})$.

We can search a $\pm 2$-substring inside segment $[L,R]$, where $1\leq L\leq R\leq n$. For this, we invoke the Grover Search algorithm only inside this segment. 

We can search the minimal $j\in[L,R|$ such that $f(j)=1$ gives us the first $\pm 2$-substring in the segment. We can use the First One Search algorithm from Section \ref{sec:first-one2} for this. The complexity is $O(\sqrt{j'-L})$, where $j'$ is the target index. We can search the maximal $j\in[L,R|$ such that $f(j)=1$ gives us the last $\pm 2$-substring in the segment. Here we also can use the First One Search algorithm, but we search in the segments $[R-2^J, R]$ for different powers of two $2^J$. The complexity is $O(\sqrt{R-j''})$, where $j''$ is the target index. 

%%%%%%%%%%%%%%%%%%%%%%%%%%%%%%%%%%%%%%%%%%%%%%%%%%%%%
\subsubsection{$\pm 3$-substring Search}
%%%%%%%%%%%%%%%%%%%%%%%%%%%%%%%%%%%%%%%%%%%%%%%%%%%%%
Let us discuss how to search $\pm t$-substring in a segment $[L,R]$ if $t=3$.

Firstly, we fix a size of $\pm 3$-substring. Assume that it is $d>2$. We should have at least $3$ symbols in the target substring, that are $000$ or $111$.

For an index $q\in[L,R]$, let us consider an algorithm that checks whether the index $q$ is inside a $\pm 3$-substring of size at most $d$. For this reason, we check whether $q$ belongs to a substring of one of the two following patterns:
\begin{itemize}
\item a substring $x[i,j]$ is $+3$-substring if there are $u',u''$ such that $i\leq u'\leq j$, $i\leq u''\leq j$ and
\begin{itemize}
\item $x[i,u']$ is a $+2$-substring,
\item $x[u'',j]$ is a $+2$-substring,
\item if $u'< u''$, then there are no $\pm 2$-substring in $x[u',u'']$. 
\end{itemize}
The last condition is important, because if there is a $-2$-substring, then  $g(x[i,j])\leq 3$ or it is not minimal.  
\item a substring $x[i,j]$ is $-3$-substring if there are $u',u''$ such that $i\leq u'\leq j$, $i\leq u''\leq j$ and
\begin{itemize}
\item $x[i,u']$ is a $-2$-substring,
\item $x[u'',j]$ is a $-2$-substring,
\item if $u'< u''$, then there are no $\pm 2$-substring in $x[u',u'']$. 
\end{itemize}
\end{itemize}

For this reason, we invoke the following algorithm
\begin{itemize}
\item Step 1. We check, whether $q$ is inside a $\pm 2$ substring. We can check $x[q]=x[q+1]$ or $x[q]=x[q-1]$ condition for this reason. Let $x[i',j']$ be the result $\pm 2$ substring in that case.
If the condition is true we go to Step 2. If it is false, then we go to Step 4.
\item Step 2. We check whether it is the left side of the pattern string. We search the minimal $u\in [i'+1,min(i'+d-2,R-1)]$ such that $x[u]=x[u+1]$. If $sign(x[i',j'])=sign(x[u,u+1])$, then we stop the algorithm and return $x[i',u+1]$ as a result substring. Otherwise, we go to Step 3.
\item Step 3. We check, whether $x[i',j']$ is the right side of the pattern string. We search the maximal $u\in [max(L+1,j'-d+1),j'-1]$ such that $x[u]=x[u-1]$. If $sign(x[i',j'])=sign(x[u-1,u])$, then we stop the algorithm and return $x[u-1,j']$ as a result substring. Otherwise, $q$ does not belong to a $\pm 3$-substring of a length at most $d$ and the algorithm fails.
\item Step 4. Assume that $q$ is not inside a $\pm 2$ substring. Then, we search a $\pm 2$-substring to the right. That is the minimal $u\in [q,min(R-1,q+d-2)]$ such that $x[u]=x[u+1]$. If we found it, then it is $x[i',j']$, where $j'=i'+1$. Then, we go to Step 5. Otherwise, the algorithm fails.
\item Step 5. We search a $\pm 2$-substring to the left from $x[i',j']$. That is the maximal $u\in [max(L+1,j'-d+1),j'-1]$ such that $x[u]=x[u-1]$. If there is no such $u$, then the algorithm fails. If we found $u$, then we check whether $sign(x[i',j'])=sign(x[u-1,u])$. If it is true, then we stop the algorithm and return $x[u-1,j']$ as a result substring. Otherwise, the algorithm fails.
\end{itemize} 

Let the procedure that implements this algorithm be called $\textsc{FixedPoindFixedLen}(q,d,t,L,R)$, where $t=3$.

It returns the pair of the indexes $(i,j)$ such that $x[i,j]$ is a $\pm 3$-substring. If the algorithm fails, then it returns $NULL$. 

The complexity of the algorithm is $O(\sqrt{d})$ because of the complexity of Grover's search.

Let us consider a randomized algorithm that searches any $\pm 3$-substring in a segment $[L,R]$ of a length at most $d$. We randomly choose $q\in [L,R]$ and check whether $\textsc{FixedPoindFixedLen}(q,d,t,L,R)\neq NULL$.
If $q$  belongs to the target segment, then the algorithm succeeds, otherwise, it fails. The probability of success is $p_{success}=\frac{len}{R-L}<\frac{d}{R-L}$, where $len<d$ is the length of the target segment. 

We apply the amplitude amplification algorithm (Section \ref{sec:amplampl}) to this randomized algorithm. The complexity is $O\left(\sqrt{\frac{1}{p_{success}}}\cdot \sqrt{d}\right)=\bigo{\sqrt{\frac{R-L}{d}}\cdot \sqrt{d}}=O(\sqrt{R-L})$.

Let the procedure that implements this algorithm be called $\textsc{FixedLen}(d,t,L,R)$, where $t=3$. 

Let us discuss the algorithm for searching the minimal indexes $(i,j)\in [L,R]$ such that $x[i,j]$ is a $\pm 3$-substring. We modify the binary search version of the First One Search algorithm from Section \ref{sec:first-one}. The algorithm is presented in Algorithm \ref{alg:first-segment}.

\begin{algorithm}[ht]
\caption{Quantum algorithm for searching the minimal indexes of $\pm 3$-segment in $[L,R]$ of size $d$.}\label{alg:first-segment}
\begin{algorithmic}
\If{$\textsc{FixedLen}(d,t,L,R)\neq NULL$}
\State $left\gets L$, $right\gets R$.
 \While{$left<right$}
\State $middle\gets \lceil (left + right)/2 \rceil$
\If{$\textsc{FixedLen}(d,t,left,middle-1)\neq NULL$}
\State $right\gets middle-1$
\Else
\If{$\textsc{FixedPointFixedLen}(middle,d,t,left,right)\neq NULL$}
\State $result\gets\textsc{FixedPointFixedLen}(middle,d,t,left,right)$
\State Stop the loop.
\Else
\State $left\gets middle+1$
\EndIf
\EndIf
 \EndWhile
\Else
 \State $result \gets NULL$
\EndIf
\State \Return $result$
\end{algorithmic}
\end{algorithm}

Similar to the proof of complexity from Section \ref{sec:first-one} we can show that complexity is $O(\sqrt{R-L})$.

Let the procedure that implements this algorithm is called $\textsc{FixedLenMinimal}(d,t,L,R)$, where $t=3$. 

Similarly, we can define the procedure  $\textsc{FixedLenMaximal}(d,t,L,R)$ that searches the maximal indexes of a $\pm t$-segment from $[L,R]$, where $t=3$.

We can discuss how to fix the issue with a fixed length of the segment $d$. Let us consider possible lengths of the segment $\{2^0,2^1,2^2,\dots,2^{\lceil \log _2 n\rceil}\}$. There is the minimal $j$ such that $2^j>len'$, where $len'$ is the minimal possible length of a $\pm 3$-substring in $[L,R]$. The algorithm fails for each smaller power of $2$. So, we can find the target $j$ using the First One Search algorithm (Section \ref{sec:first-one2}). The complexity of the algorithm is $O(\sqrt{j}\cdot T)=O(\sqrt{\lceil \log_2 len' \rceil}\cdot T)=O((\log (R-L))^{0.5}\cdot T)$, where $T$ is the complexity of the algorithm for the fixed size of a segment.
Finlally, we can use three procedures with complexity $O((\log (R-L))^{0.5}\sqrt{R-L})$: $\textsc{SegmentMinimal}(t,L,R)$, $\textsc{SegmentMaximal}(t,L,R)$ and $\textsc{SegmentAny}(t,L,R)$ for the minimal indexes of a $\pm 3$-segment, maximal ones and any indexes of a $\pm 3$-segment, respectively.

Note, that on each step we have Grover's Search based algorithms as subroutines. Therefore, we should use the bounded-error-input version of Grover's Search (Section \ref{sec:random-oracle}).

%%%%%%%%%%%%%%%%%%%%%%%%%%%%%%%%%%%%%%%%%%%%%%%%%%%%%%
\subsubsection{$\pm t$-substring Search}
%%%%%%%%%%%%%%%%%%%%%%%%%%%%%%%%%%%%%%%%%%%%%%%%%%%%%%
Let us discuss the way of searching $\pm t$-substring in a segment $[L,R]$. Assume that we already know how to implement $\textsc{SegmentMinimal}(t-1,L,R)$, $\textsc{SegmentMaximal}(t-1,L,R)$ and $\textsc{SegmentAny}(t-1,L,R)$ procedures. The complexity of these procedures is $O(\sqrt{R-L}(\log (R-L))^{0.5(t-3)})$.

We use the idea similar to the idea of the algorithm for  $\pm 3$-substring.
Firstly, we fix a size of $\pm t$-substring. Assume that it is $d>t-1$. We should have at least $t$ symbols in the target substring, that is $00\dots 0$ or $11\dots 1$.

For a given index $q\in[L,R]$, let us consider an algorithm that checks whether an index $q$ inside a $\pm t$-substring of size at most $d$. For this reason, we check whether $q$ belongs to a substring of one of the two following patterns:
\begin{itemize}
\item a substring $x[i,j]$ is a $+t$-substring if there are $u',u''$ such that $i\leq u'\leq j$, $i\leq u''\leq j$ and
\begin{itemize}
\item $x[i,u']$ is a $+(t-1)$-substring,
\item $x[u'',j]$ is a $+(t-1)$-substring,
\item if $u'< u''$, then there are no $\pm (t-1)$-substring in $x[u',u'']$. 
\end{itemize}
\item a substring $x[i,j]$ is a $-t$-substring if there are $u',u''$ such that $i\leq u'\leq j$, $i\leq u''\leq j$ and
\begin{itemize}
\item $x[i,u']$ is a $-(t-1)$-substring,
\item $x[u'',j]$ is a $-(t-1)$-substring,
\item if $u'< u''$, then there are no $\pm (t-1)$-substring in $x[u',u'']$. 
\end{itemize}
\end{itemize}

For this reason, we invoke the following algorithm
\begin{itemize}
\item Step 1. We check, whether $q$ is inside a $\pm (t-1)$ substring. We can invoke $\textsc{FixedPointFixedLen}(q,d,t-1,L,R)$ procedure for checking the condition. If the condition is true, then let $x[i',j']$ be the result $\pm (t-1)$ substring, and we go to Step 2. If it is false, then we go to Step 4.
\item Step 2. We check, whether the found $x[i',j']$ is the left side of the pattern string. We search the minimal $(i'',j'')$ such that $i'+1\leq i''\leq j''\leq min(i'+d-2,R)$ and $x[i'',j'']$ is a $\pm (t-1)$-substring. We use $\textsc{SegmentMinimal}(t-1,L,R)$ for this reason. If $sign(x[i',j'])=sign(x[i'',j''])$, then we stop the algorithm and return $x[i',j'']$ as a result substring. Otherwise, we go to Step 3.
\item Step 3. We check, whether $x[i',j']$ is the right side of the pattern string. We search the maximal $(i'',j'')$ such that $max(L,j'-d+1)\leq i''\leq j''\leq j'-1$ and $x[i'',j'']$ is a $\pm (t-1)$-substring. We use $\textsc{SegmentMaximal}(t-1,L,R)$ for this reason. If $sign(x[i',j'])=sign(x[i'',j''])$, then we stop the algorithm and return $x[i'',j']$ as a result substring. Otherwise, $q$ does not belong to a $\pm t$-substring of a length at most $d$ and the algorithm fails.
\item Step 4. Assume that $q$ is not inside a $\pm (t-1)$ substring. Then, we search for a $\pm (t-1)$-substring to the right. That is the minimal $(i',j')$ such that $q\leq i'\leq j'\leq min(R,q+d-2)$ and $x[i',j']$ is a $\pm (t-1)$-substring. We use $\textsc{SegmentMinimal}(t-1,L,R)$ for this reason. If the procedure finds a $\pm (t-1)$-substring, then we go to Step 5. Otherwise, the algorithm fails.
\item Step 5. We search a $\pm (t-1)$-substring to the left from $x[i',j']$. 
That is the maximal $(i'',j'')$ such that $max(L,j'-d+1)\leq i''\leq j''\leq j'-1$ and $x[i'',j'']$ is a $\pm (t-1)$-substring.  We use $\textsc{SegmentMaximal}(t-1,L,R)$ for this reason. If there is no such $(i'',j'')$, then the algorithm is failed. If we found $(i'',j'')$, then we check $sign(x[i',j'])=sign(x[i'',j''])$. If it is true, then we stop the algorithm and return $x[i'',j']$ as a result substring. Otherwise, the algorithm fails.
\end{itemize} 

We have implemented
the procedure $\textsc{FixedPoindFixedLen}(q,d,t,L,R)$, where $t\geq 3$. It returns a pair of indexes $(i,j)$ such that $x[i,j]$ is a $\pm t$-substring. If the algorithm fails, then it returns $NULL$. 

The complexity of the algorithm is $O(\sqrt{d}\cdot (\log d)^{0.5(t-3)})$ because of complexity of Grover's search and complexity of $\textsc{SegmentMinimal}(t-1,L,R),\textsc{SegmentMaximal}(t-1,L,R)$ and $\textsc{FixedPoindFixedLen}(q,d,t-1,L,R)$ procedures.

We obtain the procedure $\textsc{FixedLen}(d,t,L,R)$ from $\textsc{FixedPoindFixedLen}(q,d,t,L,R)$ using amplitude amplification algorithm (Section \ref{sec:amplampl}) similarly to the $t=3$ case. 

Similarly to the $t=3$ case, we can obtain the algorithm for searching the minimal indexes $(i,j)\in [L,R]$ such that $x[i,j]$ is a $\pm t$-substring. We use a modification of the binary search version of the First One Search algorithm from Section \ref{sec:first-one}.

Using these ideas we obtain $\textsc{FixedLenMinimal}(d,t,L,R)$ procedure and  $\textsc{FixedLenMaximal}(d,t,L,R)$ procedure for searching the maximal indexes of the $\pm t$-segment from $[L,R]$.

Procedures $\textsc{FixedLen}(d,t,L,R)$, $\textsc{FixedLenMinimal}(d,t,L,R)$ and $\textsc{FixedLenMaximal}(d,t,L,R)$ have complexity $O(\sqrt{R-L} \cdot (\log (R-L))^{0.5(t-3)})$.

Similarly to the $t=3$ case, we can resolve the issue with the fixed length $d$ using The First One Search algorithm for powers of two. Finally, we can use three procedures with complexity $O((\log (R-L))^{0.5(t-2)}\sqrt{R-L})$: $\textsc{SegmentMinimal}(t,L,R)$, $\textsc{SegmentMaximal}(t,L,R)$ and $\textsc{SegmentAny}(t,L,R)$ for the minimal indexes of a $\pm t$-segment, maximal ones and any indexes of a $\pm t$-segment, respectively.
%%%%%%%%%%%%%%%%%%%%%%%%%%%%%%%%%%%%%%%%%%%%%%%%%%%%
\subsubsection{Complexity of the Dyck Language Recognition}
We have discussed that the Dyck language recognition is equivalent to $\pm(k+1)$-substring search. So, the complexity of the algorithm is $O(\sqrt{n}\cdot (\log n)^{0.5(k-1)})$. The result was presented in \cite{abikkpssv2020}.

The lower bound for the problem is also known. It is $\Omega(\sqrt{n}\cdot c^k)$. The lower bound was proven by two groups independently \cite{abikkpssv2020,bps2021}.

Let us remind that classical complexity (lower and upper bound) is  $\Theta(n)$.

So, we can see two points:
\begin{itemize}
    \item The lower bound and upper bound are close, but they are not the same. It is an open problem to find a quantum algorithm that ``changes'' logarithm to constant in complexity or to prove a better lower bound.
    \item If $k=\Omega(\log n)$, then we cannot obtain a quantum speed-up.
\end{itemize}

Based on the same idea there is an algorithm for Dyck language recognition in the case of multiple types of brackets that have the same upper and lower bounds as in the case of a single type of brackets (parentheses). The research is presented in \cite{kk2021}. 
\section{Quantum Hashing and Quantum Fingerprinting}
Let us consider the streaming computational model that was discussed in Section \ref{sec:qbp}. The main goal is a minimization of memory.
 
Let us consider the Strings Equality problem for data stream processing algorithms.

   \begin{itemize}
     \item The input is $u2v$, where  $u\in\{0,1\}^n$, $v\in\{0,1\}^m$ are two binary strings.
     \item We should output $1$ if $u=v$, and $0$ otherwise.
     \item Example:
     
   $010100010001\textbf{2}010100010001 \mbox{ are equal}$
   
  $01010\textit{0}010001\textbf{2}01010\textit{1}010001 \mbox{ are inequal}$
   \end{itemize}
Let us discuss a deterministic solution.
We store $u$ in the memory. Then, we read $v_i$ and check equality of $u_i=v_i$ for each $i\in\{1,\dots,n\}$. If $m\neq n$, then we return $0$.
The memory complexity of the algorithm is $O(n)$.

Let us discuss a randomized solution.
The idea is to guess an unequal position and check whether it is the correct guess. We randomly choose $j\in\{1,\dots,n\}$. Then, we skip all indexes except $j$ and keep $v_j$ in memory. After that, we skip $n$ variables and compare them with $u_j$. If $v_j=u_j$, then we return $1$ and $0$ otherwise. Here we assume, that $n$ and $m$ are known in advance, and we already checked $n=m$ equality. If it is a wrong assumption, then we can assume that $n$ is maximal possible on the random choice step, then we update $j$ according to the actual value of $n$.
 The memory complexity of the algorithm is $O(\log n)$.
 Let us compute the success probability. In the worst case, strings differ in only one bit. Probability of picking this bit is $p_{sucess}=\frac{1}{n}$. 
  So, the size of memory is good, but the success probability is bad.

Let us discuss alternative approaches. How do we compare strings in programming languages (deterministically)? We compare hash values first. If hash values are equal, then we compare strings symbol by symbol. Here by a hash value, we mean a result of a hash function. In fact, as a hash function, we can choose any function  $h:X\to Y$ such that
 \begin{itemize}
\item $|X|>|Y|$; and
\item if $h(x)\neq h(x')$, then $x\neq x'$.
\end{itemize}

Let us discuss an example of a hash function.
 Let us consider a binary string $x$, e.g. $x=10110101$. Let $value(x)$ be a number which binary representation is reversed $x$, e.g. 
 
 $value(x)=1\cdot 2^0 + 0\cdot 2^1 + 1\cdot 2^2 + 1\cdot 2^3 + 0\cdot 2^4 + 1\cdot 2^5 + 0\cdot 2^6 + 1\cdot 2^7 =173$

 As a hash function we choose $h_p:\{0,1\}^*\to \{0,\dots,p-1\}$, such that $h_p(x)=value(x)$ mod $p$, for some prime $p$, e.g. for $p=7$, we have $h_7(10110101)=2$.
 
Let us consider only $h_p$ as a hash function. The algorithm for our problem is the following. We compute a hash value for the first string $h_p(u)$ and for the second one $h_p(v)$. If $h_p(u)=h_p(v)$, then we say that strings are equal. If $h_p(u)\neq h_p(v)$, then $u\neq v$.

For this algorithm, we store the hash values and the index of the current symbol. So, memory complexity is $O(\log p+\log n)$. The solution is deterministic and has good memory complexity. Does it work correctly? Sometimes... 

If $u=v$, then $h_p(u)=h_p(v)$ always. At the same time, if $u\neq v$, then we cannot claim that $h(u)\neq h(v)$. On the other hand, if $h(u)\neq h(v)$, then we can be sure that $u\neq v$. At the same time,  if $h(u)=h(v)$, then we cannot be sure that $u=v$. We can have \textit{collisions}. It means we can have two strings $u\neq v$ such that $h_p(u)=h_p(v)$. For example, $h_7(10110101)=173$ mod $7=2$, and $h_7(00000010)=2$ mod $7=2$. Hash values are equal, but strings are different.

How can we fix the issue? Let us choose two primes $p$ and $q$. We compute two hash values for each string. If $h_p(u)=h_p(v)$ and $h_q(u)=h_q(v)$, then we claim that $u=v$.  The memory complexity is $O(\log n+ \log p+\log q)$. It is much better but does not work always correctly anyway. The new idea can have a collision also. Maybe, we should choose more primes? Three, four, five...? How many primes should it be?

Due to Chinese Remainder Theorem (CRT) \cite{ir90}, we should take $p_1,\dots, p_z$ such that \[\prod_{i=1}^z p_i\geq value(u),value(v).\] In that case, if $u\neq v$ and length $n=m$, then there is $j\in\{1,\dots z\}$ such that $value(u)$ mod $p_j\neq value(v)$ mod $p_j$ or $h_{p_j}(u)\neq h_{p_j}(v)$.
So, the size of memory is $\log_2(p_1)+\dots+\log_2(p_z)=\log_2(p_1\cdot \dots\cdot p_z)=O(\max\{n,m\})$ bits. There is no difference in memory size with the naive deterministic algorithm.

Let us convert the idea to a randomized algorithm. We can pick $p\in\{p_1,\dots,p_z\}$ uniformly  randomly. In that case, the success probability is $p_{success}=\frac{1}{z}$ due to the Chinese Remainder Theorem (CRT). We can say that $p_i\geq 2$, therefore $z\leq \log_2 value(u)\leq n$. The memory complexity is $O(\log p_i)=O(\log n)$ for storing the hash value, and $O(\log n)$ for storing the index of an observing symbol. So, the total memory complexity is $O(\log n)$. At the same time, the success probability is $p_{success}\approx \frac{1}{n}$. So, there is no benefit compared to the previous randomized algorithm.

Let us take $z+1$ primes $\{p_1,\dots,p_{z+1}\}$. Assume that $u\neq v$. Due to CRT, there is $j\in\{1,\dots z+1\}$ such that $h_{p_j}(u)\neq h_{p_j}(v)$. Due to CRT, there is $j'\in\{1,\dots z+1\}\backslash\{j\}$ such that $h_{p_{j'}}(u)\neq h_{p_{j'}}(v)$. Therefore, if we pick $p\in\{p_1,\dots,p_{z+1}\}$ uniformly randomly, then $p_{success}\approx \frac{2}{n}$.

If we take $z+t$ primes $\{p_1,\dots,p_{z+t}\}$, then due to the same idea,  $p_{success}\approx \frac{1+t}{n}$ for picking  $p\in\{p_1,\dots,p_{z+t}\}$ uniformly randomly.
If we choose $z+t=n/\varepsilon$, then $p_{success}\approx 1-\varepsilon$, and the error probability is $\varepsilon$. 
Let us compute the memory complexity. We need $O(\log n)$ memory for storing an index of the current symbol. Additionally, we need $O(\log p)=O(\log(p_{z+t}))$ memory for storing a hash value. An $i$-th prime number $p_i\approx i\ln i$ because of Prime number Theorem (Chebyshev Theorem) \cite{d82}.
Therefore, the memory complexity is 
\[\bigo{\log n + \log((z+t)\log (z+t))}=\bigo{\log n +\log(\frac{n}{\varepsilon}\log\frac{n}{\varepsilon}) }=\]\[
\bigo{\log n +\log n +\log\varepsilon^{-1} + \log\log\frac{n}{\varepsilon} }=\bigo{\log n+\log \varepsilon^{-1}}\]

Finally, we obtain the algorithm that have $O(\log n+\log \varepsilon^{-1})$ memory complexity for $\varepsilon$ error probability.  If $\varepsilon$ is constant or $O(1/n)$, then the memory complexity is $O(\log n)$.

Let us present it as Algorithm \ref{alg:fp}.
\begin{algorithm}[ht]
\caption{Fingerprinting technique for comparing two strings}\label{alg:fp}
\begin{algorithmic}
\State $i\gets Random(1,\dots,n/\varepsilon)$
\State $p\gets p_i$
\State $read(x)$

\State $hu \gets 0$

\While{$x\neq 2 $}
\State $hu \gets (hu\cdot 2 + x)\mbox{ }mod\mbox{ }p$
\State $read(x)$
\EndWhile
\State$read(x)$

\State$hv \gets 0$

\While{$x\neq EndOfLine$}

\State$hv \gets (hv\cdot 2 + x)\mbox{ }mod\mbox{ }p$

\State$read(x)$

\EndWhile
\If{$hu=hv$}
\State $result\gets True$
\Else
\State $result\gets False$
\EndIf
\State \Return $result$
\end{algorithmic}
\end{algorithm}

The presented algorithm is called Fingerprinting technique \cite{Fre79}. It has many different applications that can be found, for example, in \cite{agky14,agky16}. We can also say that the Rabin-Karp algorithm \cite{kr87,cormen2001} for searching substring in a text also has a similar idea.

\subsection{Quantum Fingerprinting Technique}
Let us discuss the quantum version of this idea. 

Firstly, we discuss the algorithm on a single qubit.
Let us choose some positive integer $k$. We use one qubit. The hash value of $u$ is the state of the qubit or an angle of the qubit. The angle is
 $\alpha = 2\pi k\cdot value(u)$.
 On reading a symbol $u_i$, if $u_i=1$, then we rotate it to $ \gamma_i = 2\pi\cdot k\cdot 2^{i}$. If $u_i=0$, then we do nothing. We can say, that we rotate to $u_i\cdot\gamma_i$ angle. After reading the string $u$ we obtain the following state of the qubit
 \[|\psi_u\rangle=cos(2\pi k\cdot value(u))|0\rangle+sin(2\pi k\cdot value(u))|1\rangle.\]
 
 In the same way, we can use one more qubit and compute the hash value for $v$. That is 
\[|\psi_v\rangle=cos(2\pi k\cdot value(v))|0\rangle+sin(2\pi k\cdot value(v))|1\rangle\].

Then, we compare the states of these two qubits, if they are equal, then we say, that $u=v$. If the states are different, then we say, that $u\neq v$.
How can we compare two states of qubits for equality? One way is SWAP-test (You can read details about it in Section \ref{sec:swaptest}). Here  Reverse-test idea \cite{aavz2016} is more useful.

Instead of two qubits, we use one qubit. 
\begin{itemize}
\item On reading $u_i$, we rotate the qubit to $u_i\cdot\gamma_i$ angle.
\item On reading $v_i$, we rotate the qubit to $-v_i\cdot\gamma_i$ angle.
\end{itemize}

After reading both strings, the final state is $|\psi\rangle=cos(2\pi\cdot k\cdot (a-b))|0\rangle+sin(2\pi\cdot k\cdot (a-b))|1\rangle$, where $a=value(u)$ and $b=value(v)$. Finally, we measure the qubit.

If $u=v$, then the probability of $0$-result after measurement is $Pr\{|0\rangle\}=1$; and the probability of $1$-result is $Pr\{|1\rangle\}=0$

If $u\neq v$, then $Pr\{|0\rangle\}=(cos(2\pi\cdot k\cdot (a-b)))^2$; and $Pr\{|1\rangle\}=(sin(2\pi\cdot k\cdot (a-b)))^2$.

So, we use $|0\rangle$ state for equality-result and $|1\rangle$ for inequality-result.
In the case of $u=v$, the algorithm is always correct. In the case of $u\neq v$, the probability of success is very small.

Let us boost the success probability.
The first way is similar to the deterministic approach. In other words, we use several hash functions. We choose several coefficients  $k_0,\dots,k_{t-1} $ and use $t$ qubits. The coefficient $k_i$ is used for $i$-th qubit. The state of qubits after reading both strings is 
\[|\psi_0\rangle=cos(2\pi k_0\cdot (a-b))|0\rangle+sin(2\pi k_0\cdot (a-b))|1\rangle,\]
\[\dots,\]
\[|\psi_{t-1}\rangle=cos(2\pi k_{t-1}\cdot (a-b))|0\rangle+sin(2\pi k_{t-1}\cdot (a-b))|1\rangle.\]

Finally, we measure all qubits. We say, that strings are equal if all qubits are in $|0\rangle$; and unequal if at least one of them returns $|1\rangle$.
If $u=v$, then all qubits are $|0\rangle$ and we obtain the equality-result with probability $1$.
If $u\neq v$, then $p_{error}=\prod_{i=0}^{t-1} (cos(2\pi k_i\cdot (a-b)))^2$.
At the same time, the number of qubits can be up to $n$ for obtaining a good success probability due to the Chinese Reminder Theorem.

Let us present the quantum fingerprinting algorithm.
We use two registers $|\psi\rangle$ of $\log t$ qubits and $|\phi\rangle$ of $1$ qubit.
\begin{itemize}
\item \textbf{Step 1.} We apply Hadamard transformation $H^{\otimes\log_2 t}$ to $|\psi\rangle$. The state after the step is 
\[|\psi\rangle|\phi\rangle=\frac{1}{\sqrt{t}}\sum_{i=0}^{t-1}|i\rangle|0\rangle.\]
\item \textbf{Step 2.} On reading $u_j$, if $u_j=1$, then we apply $Q_j$ such that for each $i\in\{0,\dots,t-1\}$ we have

\[|i\rangle(\cos \beta|0\rangle + \sin \beta|1\rangle) \to |i\rangle(\cos (\beta+\gamma_{i,j})|0\rangle + \sin (\beta+\gamma_{i,j})|1\rangle)\mbox{, where }\gamma_{i,j}=2\pi k_i \cdot 2^j.\]

For $a=value(u)$, the state after the step is  

\[|\psi\rangle|\phi\rangle=\frac{1}{\sqrt{t}}\sum_{i=0}^{t-1}|i\rangle(cos(2\pi k_i\cdot a)|0\rangle+sin(2\pi k_i\cdot a)|1\rangle)\]

\item \textbf{Step 3.} On reading $v_j$, if $v_j=1$, then we apply $Q_j'$ such that for each $i\in\{0,\dots,t-1\}$ we have

\[|i\rangle(\cos \beta|0\rangle - \sin \beta|1\rangle) \to |i\rangle(\cos (\beta-\gamma_{i,j})|0\rangle + \sin (\beta+\gamma_{i,j})|1\rangle)\mbox{, where }\gamma_{i,j}=2\pi k_i \cdot 2^j.\]

For $a=value(u),b=value(v)$, the state after the step is  

\[|\psi\rangle|\phi\rangle=\frac{1}{\sqrt{t}}\sum_{i=0}^{t-1}|i\rangle(cos(2\pi k_i\cdot (a-b))|0\rangle+sin(2\pi k_i\cdot (a-b))|1\rangle)\]
\item \textbf{Step 4.} We apply Hadamard transformation $H^{\otimes\log_2 t}$ to $|\psi\rangle$.
The state after the step is  
\[|\psi\rangle|\phi\rangle=\frac{1}{t}\left(\sum_{i=0}^{t-1}cos(2\pi\cdot k_i\cdot (a-b))\right)|0\rangle|0\rangle+\lambda_i\sum_{i=1}^{2t-1}|i\rangle\]

Here we use $|i\rangle$ for $|i'\rangle$ state of $|\psi\rangle$ and $|i''\rangle$ of $|\phi\rangle$, where $i'=\lfloor i/2 \rfloor$ and $i''=i$ mod $2$.
\item \textbf{Step 5.} We measure all the qubits.
\end{itemize}

If $u=v$, then we obtain $\ket{0}\ket{0}$ state with probability $1$. If $u\neq v$, then we obtain $\ket{0}\ket{0}$ state with probability $\frac{1}{t^2}\left(\sum_{i=0}^{t-1}cos(2\pi\cdot k_i\cdot (a-b))\right)^2$.

We associate $\ket{0}\ket{0}$ state with equality-result and all other states with inequality-result. If $u=v$, then we have a correct result with probability $1$. If $u\neq v$, then 
The error probability is $p_{error}=Pr\{\ket{0}\ket{0}\}=\frac{1}{t^2}\left(\sum_{i=0}^{t-1}cos(2\pi\cdot k_i\cdot (a-b))\right)^2$.
Due to \cite{an2008,an2009}, we can choose $t=\frac{1}{\varepsilon}\log_2 n$ coefficients $k_0,\dots, k_{t-1}$ such that $p_{error}<\varepsilon$. The claim is based on Azuma's theorem \cite{mr96}. An alternative procedure for choosing coefficients is presented in \cite{zk2019tucs}.

Note, that the scheme of the algorithm is similar to the Quantum Fourier Transform algorithm \cite{av2020}.

The quantum fingerprinting method was developed as computing method in \cite{af98}, and is used in many papers like \cite{an2008,an2009,l2009,av2013,av2009,av2011,kk2017,aakv2018,kk2019,kk2022,aaksv2022,kkk2022}. At the same time, it was used as a cryptography hashing algorithm in \cite{bcwd2001}, and then was used and extended in many papers \cite{av2013hash,aa2015,aa2015e,z2016,z2016group,aav2016,aavz2016,v2016binary,v2016,vlz2017,z2018,
aav2018,vvl2019,aav2020,av2022} including experimental results \cite{tavak2021}.
\subsection{SWAP-test}\label{sec:swaptest}
Firstly, let us discuss the SWAP gate. 

Assume we have two variables $x$ and $y$, and we want to swap their values. How to do it?

The first way uses a third variable:
   \begin{itemize}
   \item $x\gets a$, $y\gets b$
   \item $z\gets x$ $\quad\quad\quad x=a$, $y= b$ , $z= a$;
   \item $x\gets y$ $\quad\quad\quad x=b$, $y= b$ , $z= a$;
   \item $y\gets z$ $\quad\quad\quad x=b$, $y= a$ , $z= a$.
\end{itemize}    

What if we cannot use the third variable? Then, we can use $+$ and $-$ operations:
      \begin{itemize}
   \item $x\gets a$, $y\gets b$      
   \item $x\gets x+y$ $\quad\quad\quad x=a+b$, $y= b$;
   \item $y\gets x-y$ $\quad\quad\quad x=a+b$, $y= (a+b)-b=a$;
   \item $x\gets x-y$ $\quad\quad\quad x=(a+b)-a=b$, $y= a$.
\end{itemize}

Another way uses excluding or ($\oplus$) operation. It is based on the following properties: $a \oplus a = 0$ and $a \oplus b = b \oplus a$. The way is presented above:
      \begin{itemize}
         \item $x\gets a$, $y\gets b$  
   \item $x\gets x\oplus y$ $\quad\quad\quad x=a\oplus b$, $y= b$;
   \item $y\gets x\oplus y$ $\quad\quad\quad x=a\oplus b$, $y= (a\oplus b)\oplus b=a$;
   \item $x\gets x\oplus y$ $\quad\quad\quad x=(a\oplus b)\oplus a=b$, $y=a$.
\end{itemize}

Let us discuss the quantum version of the problem and solution that uses the idea with $\oplus$ operation.

 In the case of two qubits we have $\ket{\psi}\ket{\phi}=\ket{a}\ket{b}$ and we want to swap their values.
 For implementation  $x\gets x\oplus y$ that is  $\ket{\phi}\gets \ket{\psi\oplus \phi}$ we use $CNOT$ gate with control $\ket{\psi}$ and target $\ket{\phi}$. 
 
   \begin{itemize}
   \item If $\ket{\psi}$ is control and $\ket{\phi}$ is target for $CNOT$-gate, then modification is
   \[\ket{a}\ket{b}\to \ket{a}\ket{b\oplus a}\] 
      \item  If $\ket{\phi}$ is control and $\ket{\psi}$ is target for $CNOT$-gate, then modification is
      \[ \ket{a}\ket{b\oplus a}\to \ket{a\oplus b\oplus a}\ket{b\oplus a}=\ket{b}\ket{b\oplus a}\]
      
            \item If $\ket{\psi}$ is control and $\ket{\phi}$ is target for $CNOT$-gate, then modification is \[ \ket{b}\ket{b\oplus a}\to \ket{b}\ket{b\oplus a \oplus b}=\ket{b}\ket{a}\]

\end{itemize}

The circuit implementation of the gate and gate's notation is presented in Figure \ref{fig:swap}.
\begin{figure}[h]
\begin{center}
\includegraphics[height=3cm]{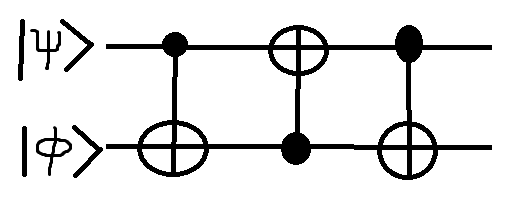} $\quad$ \includegraphics[height=3cm]{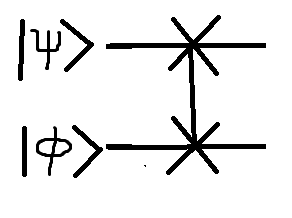}
\caption{The circuit  implementation of the SWAP-gate and gate's notation}\label{fig:swap}
\end{center}
\end{figure}

The matrix representation of the SWAP-gate is following
\[SWAP=\begin{pmatrix} 
1&0&0&0\\
0&0&1&0\\
0&1&0&0\\
0&0&0&1\\
\end{pmatrix}\]

Additionally, we can implement control SWAP-gate. If a qubit $\ket{\xi}$ is a control bit for a SWAP-gate, then we use $CNOT$-gates with the additional control qubit $\ket{\xi}$.

Let us discuss the SWAP-test itself. Assume that we have two qubits $\ket{\psi}$ and $\ket{\phi}$. We want to compare their amplitudes for equality. The first way is to measure them several times and compare probability distributions. The disadvantage of this approach is the requirement of having many copies of these qubits. Let us consider another way that uses the SWAP-gate. We use one additional qubit $\ket{\xi}$. Assume that $\ket{\psi}=\ket{a}$, $\ket{\phi}=\ket{b}$, and $\ket{\xi}=\ket{0}$.      
    
    \begin{itemize}
    
      \item \textbf{Step 1.} We apply  Hadamard transformation to $\ket{\xi}$
      \[\ket{0,a,b}\to\frac{1}{\sqrt{2}}(\ket{0,a,b}+\ket{1,a,b})\]
      \item \textbf{Step 2.} We apply control SWAP-gate to $\ket{\psi}$ and $\ket{\phi}$ as target and $\ket{\xi}$ as control:
      \[\frac{1}{\sqrt{2}}(\ket{0,a,b}+\ket{1,a,b})\to\frac{1}{\sqrt{2}}(\ket{0,a,b}+\ket{1,b,a})\]
      \item \textbf{Step 3.}We apply  Hadamard transformation to $\ket{\xi}$
      \[\frac{1}{\sqrt{2}}(\ket{0,a,b}+\ket{1,b,a})\to\frac{1}{2}(\ket{0,a,b}+\ket{1,a,b}+\ket{0,b,a}-\ket{1,b,a})=\frac{1}{2}\ket{0}(\ket{a,b}+\ket{b,a}) + \frac{1}{2}\ket{1}(\ket{a,b}-\ket{b,a})\]
      \item \textbf{Step 4.} We measure $\ket{\xi}$.    
   \end{itemize}
The circuit implementation of the algorithm is presented in Figure \ref{fig:swaptest}.
\begin{figure}[h]
\begin{center}
\includegraphics[height=3cm]{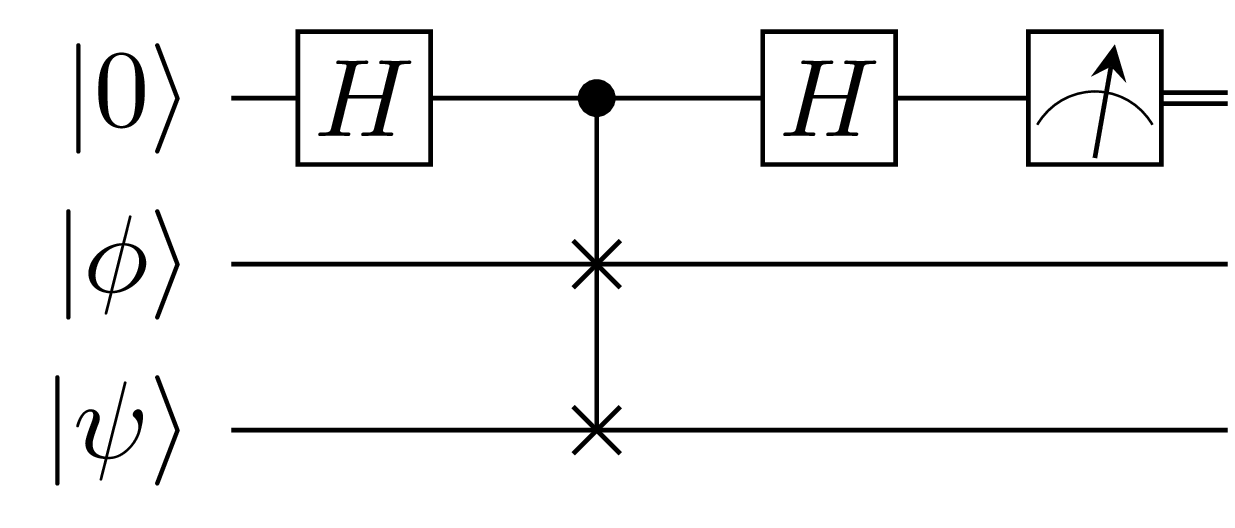} 
\caption{The circuit  implementation of the SWAP-test algorithm}\label{fig:swaptest}
\end{center}
\end{figure}

 The $0$-result probability is $Pr(\ket{0})=\bra{\Psi_0}\ket{\Psi_0}$, where $\ket{\Psi_0}=\frac{1}{2}(\ket{a}\ket{b}+\ket{b}\ket{a})$. 
Let us compute $Pr(\ket{0})$. So, $\bra{\Psi_0}=\frac{1}{2}(\bra{a}\bra{b}+\bra{b}\bra{a})$.
      
      \[ Pr(\ket{0})=\frac{1}{4}(\bra{a}\bra{b}\ket{a}\ket{b} + \bra{a}\bra{b}\ket{b}\ket{a} +\bra{b}\bra{a}\ket{a}\ket{b} + \bra{b}\bra{a}\ket{b}\ket{a})=\]
      \[=\frac{1}{4}(\bra{a}\bra{b}\ket{a}\ket{b} + \bra{a}\ket{a} +\bra{b}\ket{b} + \bra{b}\bra{a}\ket{b}\ket{a})\]
      \[=\frac{1}{4}(\bra{a}\bra{b}\ket{a}\ket{b} + 1 +1 + \bra{b}\bra{a}\ket{b}\ket{a})\]
     \[=\frac{1}{4}(\bra{a}\bra{b}\ket{a}\ket{b} + \bra{b}\bra{a}\ket{b}\ket{a}) + \frac{1}{2}\]     
     \[=\frac{1}{4}(|\bra{b}\ket{a}|\cdot \bra{a}\ket{b} + |\bra{a}\ket{b}|\cdot \bra{b}\ket{a}) + \frac{1}{2}\]
      
    \[=\frac{1}{4}(2|\bra{b}\ket{a}|^2) + \frac{1}{2} = \frac{1}{2} + \frac{1}{2}|\bra{b}\ket{a}|^2\]
      if $\ket{b}=\ket{a}$, then $Pr(\ket{0})=1$. If $\ket{b}$ and $\ket{a}$ are orthogonal, then $Pr(\ket{0})=0.5$. So, if we have several copies, then we can do the SWAP-test several times. If all invocation return $0$-result, then we can say that $\ket{b}=\ket{a}$. If there is at least one $1$-result, then $\ket{b}\neq\ket{a}$.
\section{Quantum Walks. Coined Quantum Walks}   
\subsection{Discussion. Random Walk, One-dimensional Walk.}
Firstly, let us discuss the classical base algorithm which is the Random walk. The general algorithm is the following. Assume we have a graph $G=(V,E)$ where $n=|V|$ is the number of vertices, and $m$ is the number of edges. Let us have an agent, that can be in one of the vertices.
Assume that the agent in a vertex $i$, then it chooses one of the neighbor vertices randomly. Let us define the probability of moving from a vertex $i$ to a vertex $j$ by $M[i,j]$. If there are no edges between $i$ and $j$, then $M[i,j]=0$. So, $M$ is a transition matrix for our random walk. We can say that $p=(p_1,\dots,p_n)$ is a probability destruction for vertexes. We can compute $p$ on $t$-th step by $p'$ from $(t-1)$-th step as $p=Mp'$. We assume that the agent cannot just ``jump'' from a vertex to another vertex, but can only move to neighbor vertices. We call this property ``locality''. We can use the random walk as a traversal algorithm (like Depth-first search or Breadth-first search, Section \ref{sec:dfs}). As an example, it can be a search algorithm that searches a vertex with a specific property. There are other applications of the algorithm.

 Why the algorithm is interesting? One reasonable advantage is the small size of memory. Because of the locality property, in any step, we should store only the current vertex in memory that requires $O(\log n)$ bits of memory. At the same time, if we implement one of standard traversal algorithms (Depth-first search or Breadth-first search), then you should store $O(n)$ vertices in a stack or a queue data structures for DFS or BFS respectively. That requires $O(n\log n)$ bits of memory. So, we have an exponentially smaller size of memory for a random walk. We can say \cite{b2000} that we use a more restricted computational model that is close to Branching programs (Section \ref{sec:qbp}) and 2-way probabilistic automata in which an input head can be in different positions of an input tape with some probabilities. Here the agent is an analog of the input head, moving to neighbor vertices is an analog of moving to neighbor cells of the input tape for automata and the algorithm can be in different vertices with some probabilities.
The difference is in the non-linearity of input data. Automata like computational models assume that input has a linear form (like a list or string on a tape), but random walk moves by a more general data structure that is a graph. You can read more about automata in \cite{k2019,sy2014} and about branching programs in \cite{Weg00} and in Section \ref{sec:qbp}.
 
There are many non-trivial applications of the algorithm that allow us to obtain an advantage in time or memory complexity, for example, \cite{frpu94,ksz2022,ke2022}.

As an example, we explore a simple version of a graph with a regular structure. 
Let us consider a number line with vertexes in integer points. Neighbor points are connected. We have an agent that is situated in the $0$-vertex. The agent moves with probability $q_{right}$ to the right direction and with probability $q_{left}=1-q_{right}$ to the left direction. As an example, we can take $q_{left}=q_{right}=0.5$. What happens in $t$ steps of the process? The trace of the first $3$ steps is following. Here $p_i$ is a probability of being in a vertex $i$.

\begin{itemize}
\item Step 0. $p_0=1$
\item Step 1. $p_{-1}=\frac{1}{2}$ $p_0=0$ $p_{-1}=\frac{1}{2}$
\item Step 2. $p_{-2}=\frac{1}{4}$ $p_{-1}=0$ $p_{0}=\frac{1}{2}$ $p_{1}=0$ $p_{2}=\frac{1}{4}$
\item Step 3. $p_{-3}=\frac{1}{8}$ $p_{-2}=0$ $p_{-1}=\frac{3}{8}$ $p_{0}=0$ $p_{1}=\frac{3}{8}$ $p_{2}=0$ $p_{3}=\frac{1}{8}$ 
\end{itemize}      

You can check that we obtain a normal distribution by vertices.

We have another interesting situation in the case of the number circle. 
Let us have a circle of size $s = 5$ with vertices in integer points with indexes from $0$ to $4$. Neighbor points are connected. The probability distribution of the first $99$ steps is presented in Figure \ref{fig:rw}.

\begin{figure}[ht]
\begin{center}
  \begin{tabular}{ |c | c |c|c|c|c| }
    \hline
     Step & 3  &  4&0 &1 & 2   \\ \hline
   0 &0.0  &  0.0&1.0 &0.0 & 0.0   \\ \hline 
  1 &0.0  &  0.5&0.0 &0.5 & 0.0   \\
\hline 
2 &0.25  &  0.0&0.5 &0.0 & 0.25   \\
\hline 
3 &0.125  &  0.375&0.0 &0.375 & 0.125   \\
\hline 
4 &0.25  &  0.0625&0.375 &0.0625 & 0.25   \\
\hline 
$\dots$ & &  &&&   \\
\hline 
98 &0.2  &  0.2&0.2 &0.2 & 0.2   \\
\hline 
99 &0.2  &  0.2&0.2 &0.2 & 0.2   \\
\hline 
  \end{tabular}
\end{center}
     \caption{Probability distribution on a circle of size $5$ for 0-99 steps.}  \label{fig:rw}
       \end{figure}
       
You can see that after some steps the distribution does not change. We call it a ``stationary state''. Let us have a transition matrix $M$ such that an element $M[i,j]$ is the probability of moving from the point $i$ to the point $j$. Then, the stationary state is the probability distribution $p=(p_0,\dots,p_4)$ such that $ Mp=p$. Therefore, $p$ is an eigenvector of $M$ with eigenvalue $1$. We can see that any eigenvector of $M$ with eigenvalue $1$ is a ``stationary state''.
 
\subsubsection{One-dimensional Quantum Walk}
Let us construct Quantum walk as a direct ``translation'' of Random walk to the quantum world. One quantum state is for one point. We start from the $\ket{0}$ point. Like in the Random walk algorithm, in the quantum version we want to keep the ``locality'' property. Therefore, the transformation $U$ such that
       $\ket{i}\to a\ket{i-1} + b\ket{i}+c\ket{i+1}$ for each $i$.
       
    The only ways \cite{m1996} to have a unitary transition are
      \begin{itemize}
      \item $|a|=1,b=0,c=0$
      \item $a=0,|b|=1,c=0$
      \item $a=0,b=0,|c|=1$
      \end{itemize}
      All of them are trivial. That is why we cannot directly ``translate'' the algorithm to the quantum world and should change something in the idea.
      
     Let us consider each edge of the graph as a directed edge. If an edge $(i,j)$ is not directed, then we consider it as a pair of two directed edges $i\to j$ and $i \leftarrow j$. Let us use the following notation for an edge in the case of a number line. An edge can be resented as a pair $(vertex, direction)$, where $vertex$ is an integer point, and $direction\in\{0,1\}$ is a direction of the edge, $0$ for the left, and $1$ for the right directions.
    
     We associate a quantum state with a directed edge. They are $\ket{i,0}$ and $\ket{i,1}$ for an integer point $i$.  Let us start from the $\ket{0,0}$ edge. One step consists of two transformations:
           \begin{itemize}
      \item ``Coin flip'', $C$ such that
      
        $\ket{i,0}\to a\ket{i,0} + b\ket{i,1}$,
        
        $\ket{i,1}\to c\ket{i,0} + d\ket{i,1}$.
       
       The transformation means redistribution of amplitudes (analog of probability redistribution for Random walks). 
      \item``Shift'', $S$ such that
      
      $\ket{i,0}\to \ket{i-1,0}$,
      
      $\ket{i,1}\to \ket{i+1,1}$.
      
       The transformation means moving one step.
      \end{itemize}

      Due to the ``locality'' property, both transformations affect only close states (states of the same vertex or neighbor vertices).
      After some number of steps, we measure the quantum state of the system. As a probability of a point $i$ we mean the sum of probabilities of obtaining $\ket{i,0}$ and $\ket{i,1}$. In other words, if we assume that a quantum register $\ket{\phi}$ holds an integer point and a quantum register $\ket{\psi}$ holds direction, then we take probability for obtaining $\ket{i}$ for $\ket{\phi}$. 
      
            Researchers explore different ``Coin flip'' transformations. Let, for example,  
            $C=H=\begin{pmatrix} 
\frac{1}{\sqrt{2}}&\frac{1}{\sqrt{2}}\\
\frac{1}{\sqrt{2}}&-\frac{1}{\sqrt{2}}\\
\end{pmatrix}$ 
that is a quantum analog of uniformly random choice. Another example can be 
$C=\begin{pmatrix} 
\frac{1}{\sqrt{2}}&\frac{i}{\sqrt{2}}\\
\frac{i}{\sqrt{2}}&\frac{1}{\sqrt{2}}\\
\end{pmatrix}$.
  The probability distribution for the result of measurement after $20$ steps for these coins are presented in Figure \ref{fig:1dqw} (the left side for the first coin and the right side for the second coin).

\begin{figure}[ht]
\begin{center}
\includegraphics[scale=0.6]{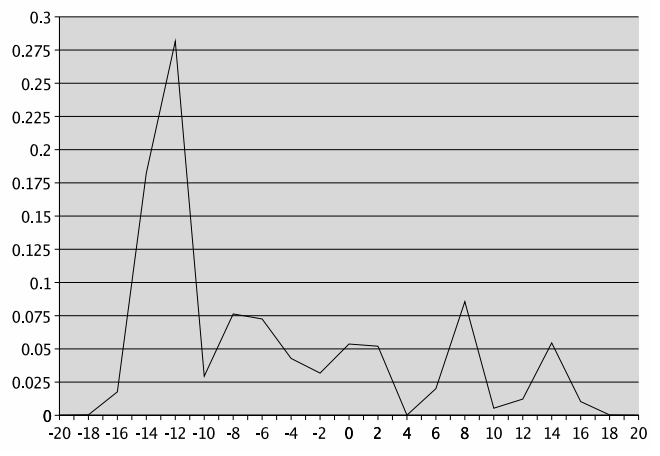}$\quad\quad$\includegraphics[scale=0.6]{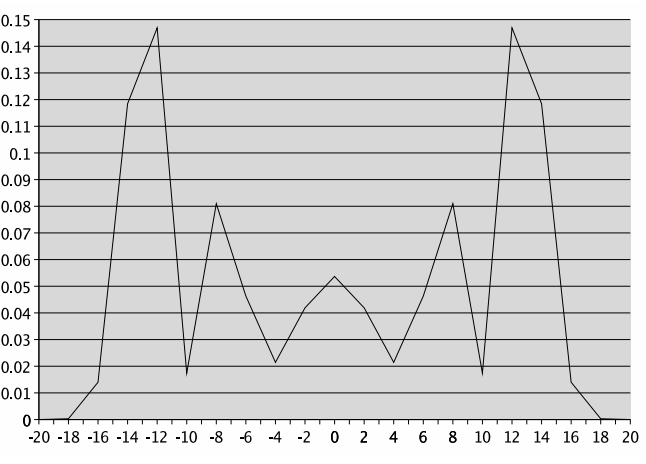}
\caption{Probability distribution for 1D quantum walk. Left side for the first coin and the right side for the second coin}\label{fig:1dqw}
\end{center}
\end{figure}
       
We can see that the probability distribution is not a normal distribution. So, we can think about some algorithmic applications that can use such a behavior.

At the same time, the two-dimensional case is more interesting than the  1D case. Let us define it in the next section.

\subsection{Two-dimensional Quantum Walk on a Grid}
Let us have an $n\times n$-grid as a graph. Assume that the grid is a torus. It means 
\begin{itemize}
\item vertices of the graph are nodes of the grid;
\item neighbors of a vertex are connected vertices in the grid;
\item (torus) the most left nodes and the most right nodes are neighbors;
\item (torus) the top nodes and the bottom nodes are neighbors.
\end{itemize}

As in the one-dimensional case, a quantum state corresponds to a directed edge. A state is $\ket{x,y,d}$, where $(x,y)$ is a coordinate of a point $(1\leq x,y\leq n)$ and $d$ is a direction of an outgoing edge. A direction $d\in\{0,\dots,3\}$ such that  $0$ is $\uparrow$, $1$ is $\downarrow$, $2$ is $\rightarrow$, and $3$ is $\leftarrow$.

As in the one-dimensional case, one step is ``coin flip'' $C$ and ``shift'' $S$.
 \begin{itemize}
 \item A coin flip $C=D$ is a Grover's diffusion (Section \ref{sec:grover}) for  $ \ket{x,y,\uparrow}, \ket{x,y,\downarrow}, \ket{x,y,\rightarrow}$ and $\ket{x,y,\leftarrow}$ states.
 
 The transformation means a redistribution of amplitudes.
 \item Shift $S$ is such that the following states are swapped
  \begin{itemize}
\item $ \ket{x,y,\uparrow}$ with $\ket{x,y+1,\downarrow},\quad\quad
\quad \ket{x,y,\downarrow}$ with $\ket{x,y-1,\uparrow}$,
\item $ \ket{x,y,\rightarrow}$with $\ket{x+1,y,\leftarrow}, \quad\quad
\quad \ket{x,y,\leftarrow}$ with $\ket{x-1,y,\rightarrow}$, 
 \end{itemize}
 The transformation means moving one step.
 \end{itemize}

 If we start from equal distribution $\ket{\psi}=\frac{1}{\sqrt{4n^2}}\sum_{x=1}^n\sum_{y=1}^n\sum_{d=1}^4\ket{x,y,d}$, then $\ket{\psi}$ is an eigenvector of $U=SU$ with eigenvalue $1$.
 
\subsection{Search Problem}

Let us consider a search problem on an $n\times n$-grid. Assume that there is a single node that is ``marked'' or ``target''. Our goal is to find the node. We can say that it is a generalization of the search problem (the problem that was solved by Grover's search algorithm, Section \ref{sec:grover}) to a 2D grid.

Classically we can solve the problem using the Random walks algorithm. Firstly, it randomly chooses one of the nodes. After that, it moves to one of the neighbor nodes with equal probability. If the agent reaches the ``marked'' vertex, then it stops. The complexity of the algorithm is $\Theta(n^2)$.

Let us discuss the quantum version of the algorithm. 
 \begin{itemize}
 \item We start from $\ket{\psi}=\frac{1}{\sqrt{4n^2}}\sum_{x=1}^n\sum_{y=1}^n\sum_{d=1}^4\ket{x,y,d}$
 \item One step is $Q$, $C$ and then $S$.
 \begin{itemize}
 \item A query $Q$ is such that $\ket{x,y,d}\to-\ket{x,y,d}$ iff $(x,y)$ is marked.
 \item A coin flip $C=D$ is a Grover's diffusion for  $ \ket{x,y,\uparrow}, \ket{x,y,\downarrow}, \ket{x,y,\rightarrow}$ and $\ket{x,y,\leftarrow}$ states.
 \item Shift $S$ is such that the following states are swapped
  \begin{itemize}
\item $ \ket{x,y,\uparrow}$ with $\ket{x,y+1,\downarrow}$,$\quad \ket{x,y,\downarrow}$ with $\ket{x,y-1,\uparrow}$,
\item $ \ket{x,y,\rightarrow}$with $\ket{x+1,y,\leftarrow}$, $\quad \ket{x,y,\leftarrow}$ with $\ket{x-1,y,\rightarrow}$, 
 \end{itemize}
 \end{itemize}
 \end{itemize}
 The complexity of searching the marked node by quantum walk is $O(n\log n)$. Detailed analysis of complexity is presented in \cite{akr2005,abnor2012,a2003}.

What if we use Grover's Search algorithm for the problem (Section \ref{sec:grover})? The search space size is $O(n^2)$. We can use $f:\{0,\dots,n-1\}\times\{0,\dots,n-1\}\to \{0,1\}$ search function, such that $f(x,y)=1$ iff $(x,y)$ is the ``marked'' vertex. So, the complexity of Grover's search algorithm is $O(\sqrt{n^2})=O(n)$. That is better than the presented complexity of the quantum walk algorithm. At the same time, Grover's search algorithm requires the quantum query model. It is an amplification of the random sampling algorithm (Section \ref{sec:grover}). It means we can sample one vertex $v$ and then another vertex $v'$ that is not a neighbor of $v$. The vertex $v'$ can be very far from $v$. If we restrict our computational model with the moving of an agent that is a generalization of 2-way quantum automata(\cite{sy2014}) or branching programs (Section \ref{sec:qbp}) on graph-like input as it was discussed in the begin of this section, then we are not allowed to ``jump'' from one vertex to another as we want. Moving to a vertex on a distance $M$ requires $M$ steps. If the diameter of the graph (the distance between the two most far vertices) is $Diam$, then the complexity of Grover's search algorithm in our restricted computational model with ``locality'' property is $O(\sqrt{n}\cdot Diam)$. It can be up to $O(n^{1.5})$. So, the Quantum walk algorithm has an advantage in the new computational model with the ``locality'' property.

Let us remind you that we solve the search problem for a single marked vertex. In the case of Grover's search algorithm, we know how to solve the search problem for several marked vertices (Section \ref{sec:grover-fixed-t}, Section \ref{sec:grover-t}). The situation with the Quantum Walks algorithm is different. Depending on different configurations of marked vertices in the grid, we can find a marked vertex or do not find it at all \cite{akr2005,ar2008,nr2018,nr2015}. This effect has algorithmic applications \cite{gnbk2021}.

The reader can find more details about the algorithm in \cite{a2003}. Discussion of different properties of the algorithm including properties according to stationary states are presented here \cite{ns2017,ns2021,pvw2016}.

The presented version of the Quantum Walk algorithm is called ``Coined quantum walk''. In fact, there are many different quantum versions of the Random Walk algorithm, including \cite{mnrs2007,s2004}.
In Section \ref{sec:qw-mnrs} we discuss one another popular version of the algorithm called MNRS-Quantum Walk.

\subsection{NAND-formula Evaluation}
In this section, we discuss one of the applications of the coined quantum walk. That is NAND-formula Evaluation.

Let us discuss the problem. It is known that any Boolean formula can be represented using several bases:
  \begin{itemize}
   \item $\&$, $\vee$, $\neg$ (In fact two operations $\&$ and $\neg$ or $\vee$ and $\neg$ are enough, but usage of all three operations is more popular and easy)
   \item $\oplus$, $\&$,$1$
   \item $\nand $(Sheffer stroke): $x\nand y = \neg (x \& y)$.
   \end{itemize}

We focus on the last one. It is interesting because the basis uses a single operation. Operations from other bases can be represented using $\nand$ operation:
\begin{itemize}
   \item $x\nand x = \neg x$ 
   \item $(x\nand x)\nand (y\nand y)= x \& y$
   \item $(x\nand y)\nand (x\nand y)= x \vee y$
    \item $x\nand (\neg x)=x\nand(x\nand x )= 1$
   \end{itemize}

The NAND formula can be naturally represented as a tree, where inner nodes correspond to NAND operation and leaves correspond to input variables. As an example, the tree representation of a formula \[((x_1 \nand x_2) \nand (x_3 \nand x_4)) \nand((x_5 \nand x_6) \nand (x_7 \nand x_8))\] is presented in Figure \ref{fig:nand-example}.

\begin{figure}[h]
\begin{center}
 \includegraphics[height=5cm]{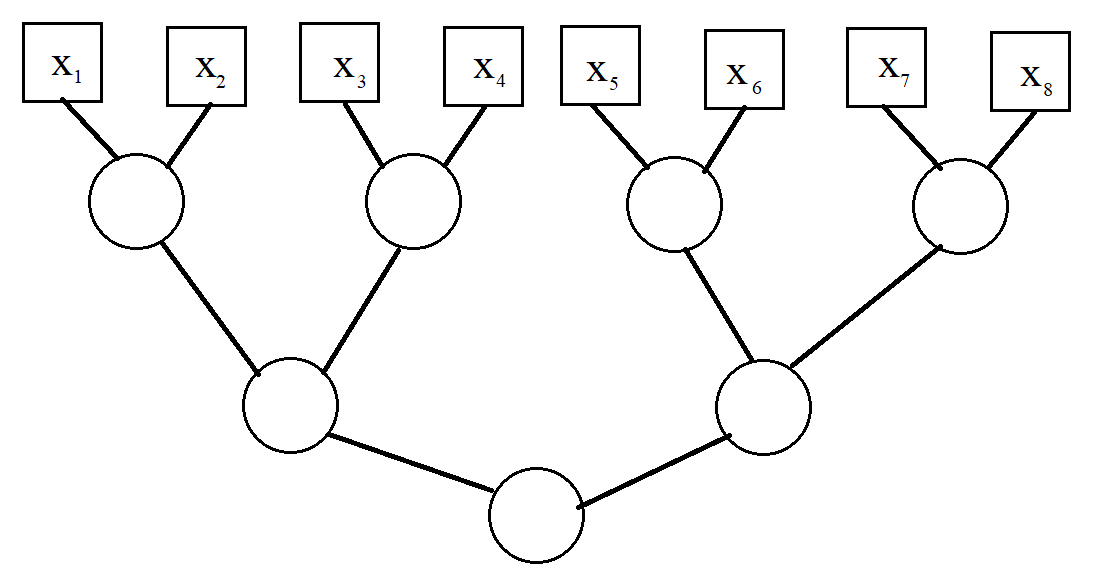}
\caption{The tree representation of a formula $((x_1 \nand x_2) \nand (x_3 \nand x_4)) \nand$ $((x_5 \nand x_6) \nand (x_7 \nand x_8))$}\label{fig:nand-example}
\end{center}
\end{figure}

The problem is the following. We have a NAND formula with $n$ non-repeated variables that is represented as a tree. We should compute the value of the formula for a given instance of input variables.

The important part of the problem is ``non-repeated variables''. If we use a variable in the formula several times, then it should be counted all these times.
Several papers discussed the case with ``shared variables'' \cite{bkt2018,ckk2012,ks2019}.

We represent the problem as a search problem. Let us consider a modification of a NAND tree. We add $L=2\lceil\sqrt{n}\rceil$ nodes to the root node. We call it tail and enumerate them $w_1,\dots,w_L$, where $w_1$ is connected with the root node and $w_2$, $w_L$ is connected with $w_{L-1}$ and $w_i$ is connected with $w_{i-1}$ and $w_{i+1}$ for $i\in\{2,\dots,L-1\}$.

We mark only leaves that correspond to $x_j=1$. Other nodes are unmarked.
the quantum state corresponds to directed edges. We define quantum states $\ket{v,d}$ of the following types:
\begin{itemize}
\item if $v=w_L$, then $d=\leftarrow$;
\item if $v\in\{w_1,\dots, w_{L-1}\}$, then $d\in\{\leftarrow,\rightarrow\}$;
\item if $v$ is an inner node of the original tree, then $d\in\{\leftarrow,\rightarrow, \downarrow\}$;
\item if $v$ is a leaf node of the original tree, then $d\in\{\downarrow\}$;
\end{itemize}

The modification of the tree is presented in Figure \ref{fig:nand-example2}.

\begin{figure}[h]
\begin{center}
 \includegraphics[height=5cm]{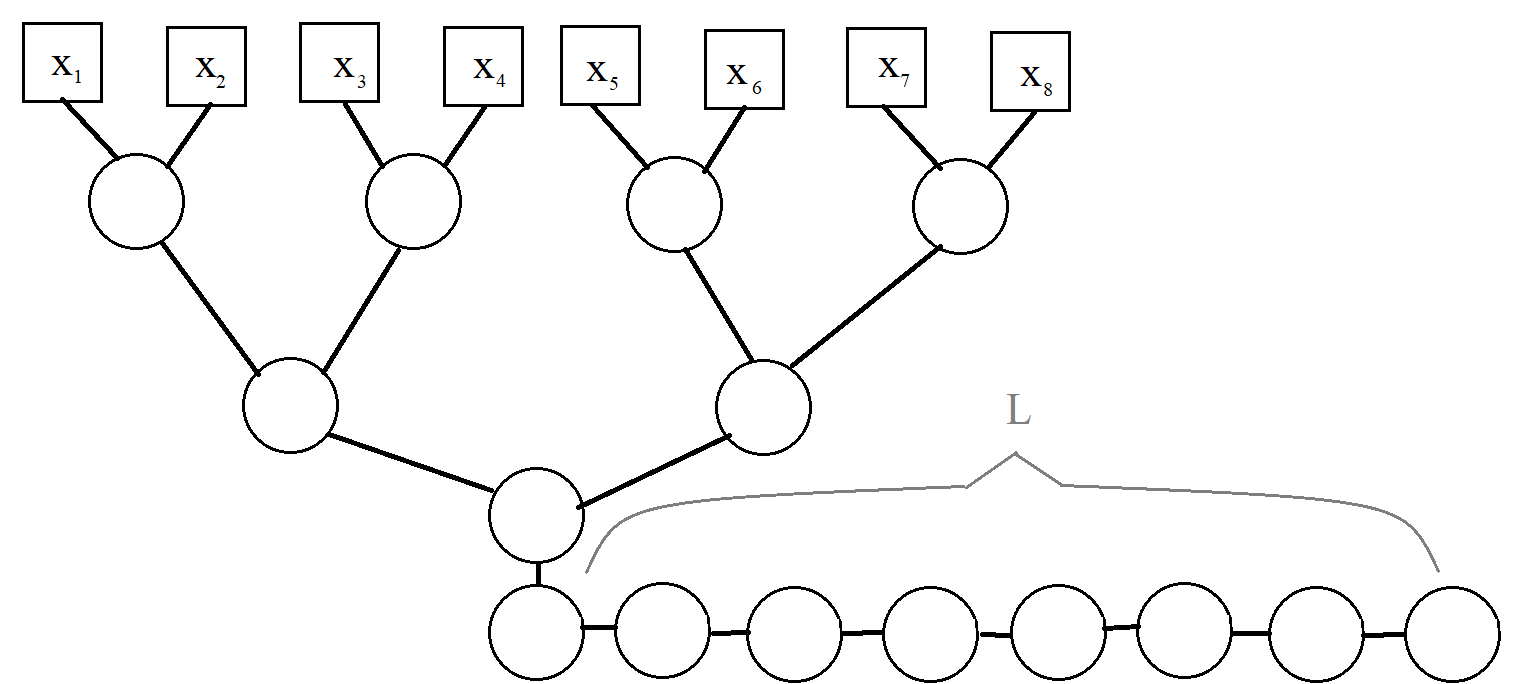}
\caption{The modified tree for a formula $((x_1 \nand x_2) \nand (x_3 \nand x_4)) \nand$ $((x_5 \nand x_6) \nand (x_7 \nand x_8))$}\label{fig:nand-example2}
\end{center}
\end{figure}

 One step of the walk is $Q$, $C$, and then $S$.
 \begin{itemize}
 \item A query $Q$ is such that $\ket{v,d}\to-\ket{v,d}$ if $v$ is marked. Note that the only marked vertices are leaf nodes that correspond to $1$-value variables. In other words, $\ket{v,d}\to(-1)^x\ket{v,d}$ if $v$ corresponds to a variable $x$.
 \item A coin flip $C=D_j$ is a Grover's diffusion for  $j$ states that correspond to $j$ outer edges of $v$. It means:
 \begin{itemize}
 \item For inner nodes, we have $C=D_3$.
 \item For a node $w_i$, $i\in\{1,\dots,L-1\}$, we have $C=D_2$.
 \item For leaf nodes and $w_L$, we have $C=D_1=I$. It is easy to see that $D_1=I$ is an identity matrix.
 \end{itemize}
 \item Shift $S$ is the redirection of edges. Formally, if $ \ket{v,d}$ and $\ket{v',d'}$ are two directed edges corresponding to one undirected edge, then $S$ swaps their amplitudes.
  \end{itemize}

After $O(T)$ steps of the Quantum walk, we measure states. If the resulting state $\ket{v,d}$ is such that $v$ is a leaf node corresponding to $x=1$ (a marked node), then the result of the formula is $1$, and $0$ otherwise.
The complexity of the algorithm is the following.
 \begin{itemize}
      \item If the tree is balanced, then complexity is $T=O(\sqrt{n})$.
      \item Otherwise, complexity is $T=O(n^{0.5+O(1/\sqrt{\log n})})$.     
   \end{itemize} 
The algorithm is presented in \cite{avrz2010,a2007,a2010}.

\section{Quantum Walk. MNRS-Quantum Walk}\label{sec:qw-mnrs}
This kind of quantum walk does not require a specific computational model but uses the quantum query model. At the same time, it can be considered as a walk by a meta graph and the walking by the meta graph requires the ``locality'' property. So, it is like a high-level algorithm with a restricted computational model for the high level, but the general query model for the main level.

The model was introduced and used in \cite{mnrs2007}. The name of this kind of Quantum walk was given by the first letters of the surnames of the authors. The interested reader can read more about the algorithm in \cite{s2008}.
\subsection{Relation of Markov Chains and Random and Quantum Walks}
This kind of quantum walk is based on the properties of Markov Chains. We can consider the Random walk algorithm as a Markov chain \cite{k2009}. States of the Markov chain are vertices of a graph. We assume that number of states is $n$ which is also the number of vertices of the graph. The probability of moving from a state $i$ to $j$ is the probability of moving from a vertex $i$ to $j$. So we use $p=(p_1,\dots,p_n)^T$ as a probability distribution vector for $n$ states of the Markov chain. An $n\times n$-matrix $M$ is a transition matrix such that  $M[i,j]$ is a probability of moving from a state $i$ to a state $j$. Let $M$ be symmetric. Then we can define $\delta$ which is a spectral gap (eigenvalue gap). The value is $\delta=1-|\lambda_2|$, where $\lambda_2$ is the eigenvalue of $M$ with the second largest magnitude.

Let us consider a search problem with a search function $f\{1,\dots,n\}\to\{0,1\}$.
We want to find any $x$ such that $f(x)=1$. Assume that there are $t$ arguments with $1$-value. Formally, $t=|\{x:f(x)=1, 1\leq x\leq n\}|$.

We can suggest three approaches to the problem. Let $\varepsilon=\frac{t}{n}$ that is the probability of finding a required $x$ on a random sampling of an argument from $\{1,\dots,n\}$.
 
 \paragraph{Solution 1. Random Sampling.} We use boosting of success probability technique for the random sampling algorithm.
      \begin{itemize}
      \item We do the following steps $O(\frac{1}{\varepsilon})$ times
      \begin{itemize}
          \item We uniformly choose $x\in_R\{1,\dots,n\}$.
          \item Then, check whether $f(x)=1$. If it is true, then we stop the process. If it is false, then continue.
      \end{itemize}
     \end{itemize}
    
     \paragraph{Solution 2. Markov Chain.} We use the Markov chain as the main search algorithm.
     \begin{itemize}
    \item We uniformly choose $x\in_R\{1,\dots,n\}$.
     \item We do the following steps $O(\frac{1}{\varepsilon})$ times:
    \begin{itemize}
     \item We check whether the walk reaches the state $x:f(x)=1$. If it is true, then we stop the process. If it is false, then continue.
     \item We do $O(\frac{1}{\delta})$ steps of the random process (random walk, Markov chain) with transition matrix $M$.
     \end{itemize}
       \end{itemize}
     
     \paragraph{Solution 2. Markov Chain. Greedy idea.} We use a greedy version of the algorithm based on the Markov chain.
      \begin{itemize}
    \item We uniformly choose $x\in_R\{1,\dots,n\}$.
     \item We do the following steps $O(\frac{1}{\varepsilon})$ times:
    \begin{itemize}
     \item We check whether the walk reaches the state $x:f(x)=1$. If it is true, then we stop the process. If it is false, then continue.
     \item We do $O(\frac{1}{\delta})$ steps of the random process (random walk, Markov chain) $M$. On each step, we check whether the walk reaches the state $x:f(x)=1$. If it is true, then we stop the process. If it is false, then continue.
     \end{itemize}
     \end{itemize}
     
     Assume that complexity of different steps is the following.
     \begin{itemize}
       \item Let $S$ be query complexity of sampling $x\in_R\{1,\dots,n\}$
     \item Let $U$ be query complexity of update (one step from a state $i$ to a neighbor state $j$)
     \item Let $C$ be query complexity of checking whether the current $x:f(x)=1$.
     \end{itemize}
     
     Then the presented solutions have the following complexity
     \begin{itemize}
     \item Solution 1. Random Sampling. We have one loop of $O(\frac{1}{\varepsilon})$ steps. Each step is sampling and checking. The complexity is \[O\left(\frac{1}{\varepsilon}(S+C)\right)\]
     
     \item Solution 2. Markov chain. Firstly, we have a sampling. Then, there are two nested loops.  The outer loop $O(\frac{1}{\varepsilon})$ steps. Each step is a checking and the nested loop of $O(\frac{1}{\varepsilon})$ steps that are updates. The complexity is \[O\left(S+\frac{1}{\varepsilon}\left(\frac{1}{\delta}U+C\right)\right)=O\left(S+\frac{1}{\varepsilon\delta}U+\frac{1}{\varepsilon}C)\right)\]
     \item Solution 3. Markov chain. Greedy idea. Firstly, we have a sampling. Then, there are two nested loops.  The outer loop of $O(\frac{1}{\varepsilon})$ steps. Each step is the nested loop of $O(\frac{1}{\varepsilon})$ steps that are updates and checking. The complexity is \[O\left(S+\frac{1}{\varepsilon\delta}(U+C)\right)=O\left(S+\frac{1}{\varepsilon\delta}U+\frac{1}{\varepsilon\delta}C)\right)\]
     
     \end{itemize}
The move from the random walk to the quantum walk is an application of the Amplitude Amplification algorithm (Section \ref{sec:amplampl}) to each repetition process. So, we have the following complexities of  the quantum versions of the algorithms:

\begin{itemize}
     \item Solution 1. Quantum version of Random Sampling. We obtain Grover's search algorithm (Section \ref{sec:grover-t}) \[O\left(\sqrt{\frac{1}{\varepsilon}}(S+C)\right)\]
     
     \item Solution 2. Quantum version of Markov chain. \[O\left(S+\sqrt{\frac{1}{\varepsilon}}\left(\sqrt{\frac{1}{\delta}}U+C\right)\right)=O\left(S+\sqrt{\frac{1}{\varepsilon\delta}}U+\sqrt{\frac{1}{\varepsilon}}C)\right)\]
     \item Solution 3. Quantum version of the greedy idea for the Markov chain.  \[O\left(S+\sqrt{\frac{1}{\varepsilon\delta}}(U+C)\right)=O\left(S+\sqrt{\frac{1}{\varepsilon\delta}}U+\sqrt{\frac{1}{\varepsilon\delta}}C)\right)\]
     
     \end{itemize}
     
     Let us discuss an application of the idea in the next section.
     \subsection{Quantum Algorithm for Element Distinctness Problem}\label{sec:elem-distinct}
     Let us define the Element Distinctness Problem
      \begin{itemize}
     \item Let us have $n$ integer variables $\{x_1,\dots,x_n\}$.
     \item We want to find $i\neq j: x_i=x_j$ or say that there is no such a pair of indexes (all elements are distinct).
      \end{itemize}
     
     The problem is connected with the Collision Problem (Section \ref{sec:collision}). The difference is the following. The statement of the Collision Problem claims that all elements are distinct ($1\!-to\!-1$ function), or we have $n/2$ pairs of equal elements ($2\!-to\!-1$ function). The statement of the Element distinctness problem claims that all elements are distinct, or we have \textbf{at least one} pair of equal elements. In the worst case, we have only one pair of equal elements, this makes the problem much harder. 
     
     It is known \cite{as2004} that if we have an algorithm for the Element Distinctness problem with query complexity $Q$, then we can construct an algorithm for the Collision problem with query complexity $O(\sqrt{Q})$. The idea is following. Let us take randomly $O(\sqrt{n})$ elements. Due to  \cite{l2005}, if the input is $2\!-to\!-1$ function, then with probability at least $0.9$ we can claim that there is at least one pair of equal elements among the taken elements. Then we can invoke an algorithm for Element Distinctness on the smaller input of size $O(\sqrt{n})$ and find the required pair. At the same time, there is no opposite connection. Therefore, we cannot use an algorithm from Section \ref{sec:collision}, and we should develop a new algorithm for Element Distinctness Problem.
     
      Firstly, let us discuss the obvious application of Grover's search algorithm. Let us define the search function $f:\{1,\dots,n\}\times\{1,\dots,n\}\to\{0,1\}$ such that $f(i,j)=1$ iff $x_i\neq x_j$. The search space size is $O(n^2)$. We can use Grover's search algorithm and obtain a quantum algorithm with $O(n)$ complexity. At the same time, classical query complexity is $\Theta(n)$ \cite{g1996} and we do not obtain any speed-up in the quantum case.
      More advanced usage of Grover's algorithm and Amplitude Amplification is \cite{b2001} that obtains $O(n^{3/4})$ query complexity.
      
      Secondly, we present the algorithm based on the quantum walks with $O(n^{2/3})$ query complexity (See \cite{a2007elementDist}) that satisfy the quantum lower bound $\Omega(n^{2/3})$ \cite{as2004}.
      
      First of all, let us present the graph that defines transitions of a Markov chain, or we can say that we invoke a random walk for this graph. Let us fix some parameter $r$ that is $1\leq r\leq n$. 
      
      \begin{itemize}
      \item a vertex $v$ corresponds to a subset $S_v\subset \{1,\dots, n\}$ of size $r$, i.e. $|S_v|=r$.
      \item two vertices $v$ and $w$ are connected by an edge iff there are two indexes $i$ and $j$ such that the difference between $S_v$ and $S_w$ only in $i$ and $j$. Formally,
      \begin{itemize}
      \item $i\in S_v$, and $j\in S_w$;
      \item $i\not\in S_w$, and $j\not\in S_v$;
      \item $S_v\backslash\{i\}=S_w\backslash\{j\}$.
\end{itemize}       
\item We assume that the edge from $v$ to $w$ is labeled by a directed pair $(i,j)$; and the edge from $w$ to $v$ is labeled by a directed pair $(j,i)$.
\end{itemize}       
The presented graph is called the Johnson graph and has $J(n,r)$ notation. 

We assume that a vertex $v$ is marked iff the corresponding set $S_v$ contains a duplicate. Formally, $v$ is marked iff $i,j\in S_v$ such that $i\neq j$ and $x_i=x_j$.

We can define the following Random walk algorithm.
\begin{itemize}
\item \textbf{Step 0}. We uniformly randomly choose a vertex of the graph. Query all variables from the set corresponding to the chosen vertex. If the vertex is marked, then we stop. Otherwise, we continue.
\item \textbf{Step t}. Assume, we are in a vertex $v$. We uniformly randomly choose one of the neighbor vertices and move to it. Let the new vertex is $w$, and the edge is labeled by $(i,j)$. For moving, it is enough to query $x_j$ and remove (forgot) $x_i$. After that, without any query, we can understand whether the vertex is marked or not because all variables of $S_w$ are already queried.
\end{itemize} 
Now we can use a quantum technique for obtaining a quantum walk.

Let us discuss each parameter of the algorithm's complexity for computing the complexity of the whole quantum algorithm. 

Let us compute the probability  $\varepsilon$ of sampling a marked vertex and the spectral gap (eigenvalue gap) $\delta$.
\begin{itemize}
\item Let us discuss the probability  $\varepsilon$. In the worst case, there is only one duplicate $x_i=x_j$, for $i\neq j$. The probability of obtaining $i$ among $r$ elements of the subset is $\frac{r}{n}$.  The probability of obtaining $j$ among other $r-1$ elements of the subset is $\frac{r-1}{n-1}$. The total probability is $\frac{r(r-1)}{n(n-1)}\approx \frac{r^2}{n^2}$. Because we are interested in asymptotic complexity, such an approximation is allowed. 
\item The spectral gap (eigenvalue gap) for Johnson graph $J(n,r)$ is already computed \cite{bh2011}[Section 12.3.2]. It is $\delta=\frac{n}{r(n-r)}$. 
\end{itemize}
Let us compute the complexity of sampling, updating, and checking ($S,U,C$).
\begin{itemize}
\item The complexity of  sampling of the first vertex $v$ is $S =O(r)$ because we should query all elements of $S_v$.
\item The complexity of update is $U=O(1)$ because we query only one new variable. 
\item The complexity of checking a vertex $v$ is $C=0$ because we already queried all variables of the set $S_v$ and can check the marking of the vertex without query.    
\end{itemize}
 
So, let us discuss the complexity of the whole algorithm:
\[\bigo{S+\sqrt{\frac{1}{\varepsilon}}\left(\sqrt{\frac{1}{\delta}}U+C\right)}=
\bigo{r+\sqrt{\frac{n^2}{r^2}}\left(\sqrt{\frac{r(n-r)}{n}}\cdot 1+0\right)}=\bigo{r+\sqrt{\frac{n(n-r)}{r}}}\]

For minimization the complexity, we can choose $r=n^{2/3}$, so the complexity is
\[\bigo{r+\sqrt{\frac{n(n-r)}{r}}}
=\bigo{n^{2/3}+\sqrt{\frac{n(n-n^{2/3})}{n^{2/3}}}}
=\bigo{n^{2/3}+n^{2/3}}=\bigo{n^{2/3}}
\]
%=\bigo{n^{2/3}+n^{2/3}}=\bigo{n^{2/3}}

In fact, if we care about time complexity, not only query complexity, then we can store all elements of a set $S_v$ for a current vertex in the set data structure. The set data structure that allows us to add, to remove elements, and to check duplicates can be implemented using a Self-balanced search tree (for example, Red-Black Tree), array, Hash Table and other data structures \cite{cormen2001}. Depending on the choice of the data structure, $U=O(1)\dots O(\log n)$, and $C=O(1)\dots O(\log n)$. So, the total complexity can be from $O(n^{2/3})$ to $O(n^{2/3}\log n)$.   

\section{Quantum Walks and Electronic Networks. Learning Graph}
Let us consider a technique of constructing a Random walk that is called Random Walks and Electronic Networks \cite{b1998,ds1984}.

Let us have a simple undirected graph $G=(V,E)$ with $n$ vertices and $m$ edges. We associate a weight $w_{(u,v)}$ with an edge $(u,v)\in E$. Let $W=\sum\limits_{(u,v)\in E}w_{(u,v)}$ be the total weight. We consider a random walk algorithm for this graph such that the probability of walking by an edge $(u,v)$ is
 $\frac{w_{(u,v)}}{\sum\limits_{(u,x)\in E}w_{(u,x)}}$.
 
 Assume that  $\xi=(\xi_1,\dots,\xi_{n})$ is a stationary state of the Random walk, where $\xi_u=\sum\limits_{(u,v)\in E}\frac{w_{(u,v)}}{2W}$. Let $\sigma=(\sigma_1,\dots,\sigma_{n})$ be an initial probability distribution on the vertices of the graph. Let $M\subseteq V$ be a set of market vertices.
We denote $H_{\sigma,M}$ as the \textbf{ hitting time} or an expected number of steps for reaching a vertex from $M$ when the initial vertex is sampled according to $\sigma$.

We define a \textbf{flow} $p_{(u,v)}$ on directed edges such that
 \begin{itemize}
 \item $p_{(u,v)}=-p_{(u,v)}$
 \item $\sigma_u = \sum\limits_{(u,v)\in E}p_{(u,v)}$
 \item $\sigma_u$ units of flow are injected into $u$, traverse through the graph, and
are removed from marked vertices.
 \end{itemize}
Energy of the flow is $\sum\limits_{(u,v)\in E}\frac{p_{(u,v)}^2}{w_{(u,v)}}$. 
\textbf{Effective resistance} $R_{\sigma,M}$ is the minimal energy of a flow from $\sigma$ to $M$.
We can discuss the electric network analog:
\begin{itemize}
\item  $w_{(u,v)}$ is Resistance.
\item  $p_{(u,v)}$ is Voltage.
\item  Effective resistance is the amount of energy that is dissipated by the electric flow.
\end{itemize}
The representation of the Random walk as an electronic network was presented in \cite{b1998,ds1984}. According to these papers
 \[H_{\xi, M}=2WR_{\xi,M}.\]
 
Let us discuss an application of the framework to quantum walks. Let us have a simple undirected graph $G=(V,E)$ that is \textbf{bipartite} with parts $A$ and $B$.

We consider the following problem. We should detect the existence of any market vertex, i.e. we should distinguish two cases: the set of marked vertices $M=\emptyset$ or $M\neq\emptyset$.

Let the initial set $A_{\sigma}=\{u:\sigma_u>0\}\subseteq A$ for a distribution $\sigma$.
We define quantum states $\{\ket{u}: u\in A_{\sigma}\}\cup\{\ket{(u,v)}: (u,v)\in E\}$.  The initial state is $\ket{\zeta}=\sum\limits_{u\in A_{\sigma}}\sqrt{\sigma_u}\ket{u}$.

The step of a quantum walk is $R_BR_A$, where $R_A=\bigoplus_{u\in A}D_u$ and $R_B=\bigoplus_{u\in B}D_u$, where  $\bigoplus$ is direct sum. The matrix $D_u$ is such that
 \begin{itemize}
 \item If $u$ is marked, then $D_u=I$
 \item If $u$ is not marked and $u\in A_{\sigma}$, then $D_u$ is the Grover's diffusion that is reflection near \[\ket{\psi}=\sqrt{\frac{\sigma_u}{C_1 R}}\ket{u} + \sum_{(u,v)\in E}\sqrt{w_{(u,v)}\ket{(u,v)}}\] for enough large constant $C_1$, and an upper bound $R$ known on the effective resistance $R_{\sigma,M}$ from $\sigma$ to $M$ for all possible sets $M$ of marked states that might appear.
 \item If $u$ is not marked and $u\not \in A_{\sigma}$, then $D_u$ is the Grover's diffusion that is reflection near \[\ket{\psi}=\sum_{(u,v)\in E}\sqrt{w_{(u,v)}\ket{(u,v)}}.\]
 \end{itemize}

The algorithm is following.
 \begin{itemize}
 \item We use Phase estimation Algorithm \cite{k1995, nc2010} for $U=R_BR_A$ transformation with preposition $\frac{1}{\sqrt{C\cdot RW}}$ for some constant $C$. The scheme is presented on Figure \ref{fig:phase-est-qw}
\begin{figure}[h]
\begin{center}
\includegraphics[width=0.5\textwidth]{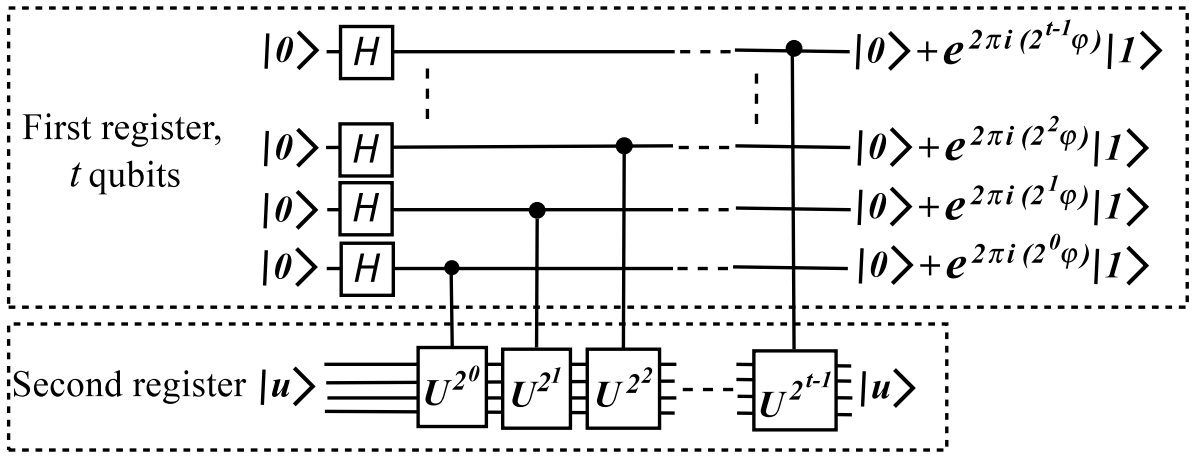}$\quad$ \includegraphics[width=0.35\textwidth]{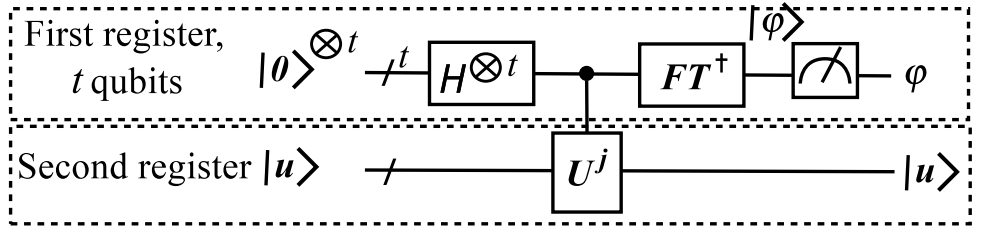}
\caption{Phase estimation for $U=R_BR_A$}\label{fig:phase-est-qw}
\end{center}
\end{figure}

\item The algorithm finds the angle $\varphi$ in eigenvalue $e^{2\pi i \varphi}$.
\item Such preposition enough to distinguish two cases:
\begin{itemize}
\item $\lambda=1$ eigenvalue that corresponds to $M=\emptyset$ and stationary state.
\item $\lambda\neq 1$ eigenvalue that corresponds to $M\neq \emptyset$.
 \end{itemize}
\end{itemize} 
 The complexity of the algorithm is $O(\sqrt{RW})$ because of complexity of the Phase estimation Algorithm.

Let us discuss two applications of this idea in Section \ref{sec:3-dist} and Section \ref{sec:qbacktracking}.
\subsection{Quantum Algorithm for 3-Distinctness Problem}\label{sec:3-dist}
We already discussed the Element distinctness problem (Section \ref{sec:elem-distinct}). Here we discuss a generalization of this problem which is the 3-Distinctness Problem:

Given $x_1,\dots,x_n\in\{0,\dots,M-1\}$ for some integer $M>n$. We want to find three different elements $i\neq j\neq z$ such that $x_i=x_j=x_z$.

Ambainis suggested \cite{a2007elementDist} an algorithm based on a technique similar to the algorithm for Element distinctness (Section \ref{sec:elem-distinct}). The complexity of the algorithm is $O(n^{3/4})$.
%todo write the correct complexity
 Here we present an algorithm based on the Electronic Networks technique \cite{bcjkm2013,b2013}. The complexity is $O^*(n^{5/7})$. Here $O^*$ hinds not only constants but $\log$ factor too.

Let us discuss the algorithm. Let a set $J\subseteq\{1,\dots,n\}$ be $\ell$-collision if $x_i=x_j$ for any $i\neq j, i,j\in J$, $|J|=\ell$, i.e. the set is indexes of equal elements. 
Let a set $J_j=\{i\in\{1,\dots,n\}: x_i=x_j\}$ be a set of elements that equal to $x_j$, including $j$ itself. The problem is detecting the existence of a $3$-collision.

Let us define some additional numerical characteristics of a sequence of numbers $x_1,\dots,x_n$.
For any $S\subseteq\{1,\dots,n\}$ and $i\in\{1,\dots,\ell\}$, let $S_i=\{j\in S:|J_j|=i\}$ be a set of indexes $j$ such that there are exactly $i$ elements equal to $x_j$. Then $r_i=|S_i|/i$ is a number of different $x_j$ such that $|J_j|=i$, i.e. there are exactly $i$ elements equal to $x_j$.
Let $\tau=(r_1,r_2,r_3)$ be the type of $S$. It is the numbers of elements that occur exactly once, twice, and three times.  
 
Assume that we have $\Omega(n)$ 2-Collisions. Let us fix some elements of type $r_1=r_2=n^{4/7}$, $r_3=0$.
Let sets with this type be
 \[V_0=\{S\subseteq\{1,\dots,n\}\mbox{ having type }(r_1,r_2,r_3)\}.\]

Let types with difference one in one of elements be $\tau_1=(r_1+1,r_2,r_3)$, $\tau_2=(r_1,r_2+1,r_3)$, $\tau_3=(r_1,r_2,r_3+1)$.
Let sets of corresponding sets and their union be
 \[V_i=\{S\subseteq\{1,\dots,n\}\mbox{ having type }\tau_i\},\mbox{ and }V=\bigcup_{i=0}^3 V_i.\]

For $i \in \{1,\dots,3\}$, let
 \[Z_i=\{(S,j):j\in\{1,\dots,n\},S\in V_{i-1},\mbox{ and }S\Delta \{j\}\not\in V\}\]
  be a set of sets from $V_{i-1}$ and elements $j$ that removes $S$ from $V$. Recall that $\Delta$ stands for the symmetric difference.  Let
\[Z=\bigcup_{i=1}^3 Z_i\] be the union of these sets.

Let us define a graph $G$ such that
\begin{itemize}
\item each vertex of the graph corresponds to one of the elements from $V\cup Z$. For simplicity of explanation, we say about vertices and corresponding elements from $V\cup Z$ as synonyms if the nature of discussing object is clear from the context.   
\item Edges of the graph are following.
\begin{itemize}
\item  A vertex $S\in V\backslash V_3$ is connected with all vertices $S\Delta\{j\}$ such that $S\Delta\{j\}\in V$, $j\in\{1,\dots,n\}$. For $S\Delta\{j\}\not\in V$ it connected with $(S,j)\in Z$.
\item A vertex $S\in V_3$ is connected to $3$ vertices from $V_{2}$ differing from $S$ in one element.
\item Each vertex $(S, j)\in Z$ is only connected to a vertex $S\in V$. 
\end{itemize}  
\item The weight of each edge is $1$.
\item The set of marked vertices is $V_3$.	
\end{itemize} 

Let us remind you that we define quantum states for vertices and edges of the graph.
The initial state is \[\ket{\zeta}=\frac{1}{\sqrt{|V_0|}}\sum_{S\in V_0}\ket{S}.\]

One step of the quantum walk is the following one.

\begin{itemize}
 \item If $u\in V_3$, then the transformation $D_u=I$, where $I$ is identity matrix.
 \item If $u\in Z$, then $D_u$ negates the amplitude of the only edge incident to $u$.
  \item Otherwise $D_u$ is Grover's diffusion.
 \end{itemize}

The complexity of the presented algorithm is $O^*(n^{5/7})$. You can find more details in \cite{bcjkm2013,b2013}.
\subsection{Quantum Backtracking}\label{sec:qbacktracking}
There is a class of problems that can be solved using a brute force algorithm of the following form.
We construct all possible sequences $a = (a_1,\dots,a_n)$, for $a_i\in\{1,\dots,d\}$ for some integer parameter $d$. Additionally, we  have a function $f':\{1,\dots,d\}^n\to \{0,1\}$ that checks whether an argument $a$ solution or not. Our goal is to find any or all arguments $a$ such that $f'(a)=1$. Let us assume that the problem is searching for any target argument. 

The complexity of checking all solutions is $O(d^n\cdot Q(f'))$, where $Q(f')$ is the complexity of computing $f'$. For simplicity, we assume that $Q(f')=O(n)$, at the same time it is not always true. If the complexity $Q(f')$ is different, then we should change this multiple in all statements of this section.  So, the complexity is $O(d^n\cdot n)$.
As an example, we can consider the Hamiltonian path problem (Section \ref{sec:hamilt}). In this problem, we generate sequences of vertices. The function $f$ checks that all vertices are presented exactly once and the sequence is a path. 

We can associate this solution with searching on $d$-nary tree of depth $n$. We call it a tree of variants.
Each vertex of the tree has $d$ outgoing edges for children. Each edge is labeled by the number from $1$ to $d$.
 A leaf vertex of the tree on level $n$ with labels on edges of the path from root $a = (a_1,\dots,a_n)$ is corresponding to the sequence $a$, and it is marked if $f(a)=1$. The goal is to search for any marked vertex in the tree. We can say that a node on level $i$ is checking all cases of assignment to a variable $a_i$.

The backtracking algorithm \cite{cormen2001} suggests us a modification of this approach. We consider partially sequences $a=(a_1,\dots,a_i,*,\dots,*)\in\{1,\dots,d,*\}^n$. Here $*$ means undefined value. In that case we define a new function $f:\{1,\dots,d,*\}^n\to\{0,1,2\}$ such that 
\begin{itemize}
\item $f(a)=0$ if $a$ is not a target sequence;
\item $f(a)=1$ if $a$ is a target sequence;
\item $f(a)=2$ if all possible complements of $a$ are not a target sequence and we should not check them. By a complement of $a$ we mean any replacing $*$ by a value from $\{1,\dots,d\}$.
\end{itemize}
 Note that $f(a)=2$ means that there are no target vertices in the subtree corresponding to $a$. At the same time, $f(a)=0$, does not guarantee that there is a target vertex in the subtree.  

The new idea for a solution is considering the ``cut'' tree of variants. A vertex corresponded to a sequence $(a_1,\dots,a_{i-1},*,\dots,*)$ on a level $i$ has only edges with value $x$ such that $f(a_1,\dots,a_{i-1},x,*,\dots,*)\neq 2$. In other words, we cut all branches with $2$-result of $f$. The algorithm is exploring the tree and is searching a marked vertex. This modification can be significantly faster. Let us consider the Hamiltonian path problem as an example again. We can understand that a sequence of vertices is wrong by a part of the sequence. If $(a_1,\dots,a_i,*,\dots,*)$ already has duplicates or it is not a path, then we can be sure that all compliments of the sequence are not valid. Such a solution has complexity $O(n!\cdot n)$ that is faster than a brute force solution with $O(n^n\cdot n)$ complexity. The complexity can be smaller if we have few edges in the considering graph.

In the general case of the backtracking algorithm the assignment of variables is not sequentially and not all the paths have the same order of assignment. We can define the technique in the following way.
Let us consider a tree such that each node of the tree corresponds to a sequence $a\in\{1,\dots,d,*\}^n$. Note, that stars and numbers can be shuffled in the sequence. Edges from the node corresponds to assignments $a_{j}\leftarrow \alpha_{j}$, where $\alpha_{j}\in\{1,\dots,d\}$ and there is $*$ on $j$-th position of $a$. The edge leads to a new node with a sequence $a'$ after the assignment $a_{j}\leftarrow \alpha_{j}$ such that $f(a')\neq 2$. We assume that we have an ``algorithm'' $h(a)$ that suggests the next assignments.
Note, that here we should not know the tree in advance, but it can be constructible using $h$. 
Assume that we know the upper bound for the size of the tree that is $T$.
The classical complexity of searching the solution or marked node in the tree is $O(T\cdot n)$, that is searching a node using the DFS algorithm (Section \ref{sec:dfs}, \cite{cormen2001}). Remind, we assume that the complexity of $f$ is $O(n)$, the same assumption we have for $h$. We assume that the total complexity of computing all outcomes of $h(a)$ is $O(n)$.

Let us discuss the quantum version of this idea. We can use Grover's search for brute force algorithm because we can enumerate all elements of search space and compute the elements by their number. Similar ideas were discussed in Section \ref{sec:hamilt} for Hamiltonian Path Problem. At the same time, usage of Grover's search for a general backtracking technique is a hard issue because we cannot enumerate all vertices and compute a vertex by its index without constructing whole the tree. If we do not know the tree in advance and the only way is constructing using the $h$ function, then we do not know the number of nodes in each brunch and we cannot compute an index of a node from the left, for example, using standard technique. 
In this section, we present an algorithm based on Quantum Walks, that achieves almost quadratic speed-up \cite{m2018}. 

We consider the tree as a bipartite graph (as the technique requires), where part $A$ is nodes on even levels (the distance from the root is even); part $B$ is nodes on odd levels. We start with full amplitude in the state corresponding to the root. The coin $C_v$ for a node $v$ is the following.
 \begin{itemize}
 \item If $v$ is marked, then $C_v=I$.
 \item If $v$ is not marked and it is not the root, then $C_v$ is Grover's diffusion for children and itself.
 \item If $v$ is the root, then $C_v$ is Grover's diffusion for children. 
 \end{itemize}
 We consider the detection of marked vertex problem. We want to check whether a marked vertex exists in the tree. 
Apply $O(\sqrt{T})$ and use phase estimation for the estimation of whether eigenvalue $1$ (stationary state) or not.
 We can find the solution itself using Binary Search on each level.

The query complexity of the algorithm is $O(\sqrt{T}n^{1.5}\log n\log(1/\varepsilon))$ for $0\leq \varepsilon\leq 1$ error probability.
The technique and algorithm were developed in \cite{m2018,
m2020}. A similar problem in Game theory was discussed in \cite{ak2017}.

\paragraph*{Acknowledgements.}
A part of the study was supported by Kazan Federal University for the state assignment in the sphere of
scientific activities, project No. 0671-2020-0065. 
\bibliographystyle{alpha}
\bibliography{tcs}

\end{document}